\documentclass[prx, aps, longbibliography, twocolumn , superscriptaddress]{revtex4-2}
\usepackage{hyperref}
\usepackage{amsmath}
\usepackage{mathtools}
\usepackage{physics}
\usepackage{tikz}
\usepackage{siunitx}
\usepackage{graphicx}
\usepackage{xcolor,colortbl}
\usepackage{floatrow}
\usepackage{amsthm}
\newtheorem{prop}{Proposition}

\usepackage{enumitem}
\usepackage{booktabs}
\usepackage{multirow}
\usetikzlibrary{matrix}
\usepackage{subfloat}
\usepackage{subcaption}
\usepackage{ragged2e}
\DeclareCaptionJustification{justified}{\justifying}
\usepackage{tabularray}
\usetikzlibrary{patterns, patterns.meta}
\usepackage{changes}

\definecolor{veryorange}{RGB}{213,89,46}
\definecolor{orange}{RGB}{222,119,66}
\definecolor{peach}{RGB}{231,155,70}
\definecolor{beige}{RGB}{247,219,92}
\definecolor{darkgreen}{RGB}{81,112,30}
\definecolor{quantumviolet}{HTML}{53257F} 
\definecolor{qpurple}{HTML}{946CBA} 
\definecolor{qviolet}{HTML}{8888C6} 
\definecolor{qbleuclair}{HTML}{6ab5db}
\definecolor{qbleugris}{HTML}{7aa0d2} 
\definecolor{qbleuunpeumoinsclair}{HTML}{58c4e1}
\definecolor{qgris}{HTML}{6e7075}

\definecolor{noncontextual}{RGB}{189,206,226}
\definecolor{npal}{RGB}{185,218,232}
\definecolor{nonsignalling}{RGB}{201,236,244}

\newfloatcommand{capbtabbox}{table}[][\FBwidth]

\newcommand{\CF}{\mathsf{CF}}
\newcommand{\SF}{\mathsf{SF}}

\newcommand{\NS}{\mathcal{NS}}
\newcommand{\NC}{\mathcal{NC}}

\newcommand{\SCHSH}{S_{\textrm{CHSH}}}

\begin{document}

\title{Certified randomness in tight space}

\author{Andreas Fyrillas}
\thanks{These authors contributed equally to this work.}
\affiliation{Quandela, 7 Rue Léonard de Vinci, 91300 Massy, France}

\author{Boris Bourdoncle}
\thanks{These authors contributed equally to this work.}
\affiliation{Quandela, 7 Rue Léonard de Vinci, 91300 Massy, France}

\author{Alexandre Maïnos}
\affiliation{Quandela, 7 Rue Léonard de Vinci, 91300 Massy, France}
\affiliation{Quantum Engineering Technology Labs, H. H. Wills Physics Laboratory and Department of Electrical and Electronic Engineering, University of Bristol, Bristol, BS81FD, UK}

\author{Pierre-Emmanuel Emeriau}

\author{Kayleigh Start}
\affiliation{Quandela, 7 Rue Léonard de Vinci, 91300 Massy, France}

\author{Nico Margaria}
\affiliation{Quandela, 7 Rue Léonard de Vinci, 91300 Massy, France}

\author{Martina Morassi}
\affiliation{Universit\'e Paris-Saclay, CNRS, Centre de Nanosciences et de nanotechnologies, 91120, Palaiseau, France}

\author{Aristide Lema\^itre}
\affiliation{Universit\'e Paris-Saclay, CNRS, Centre de Nanosciences et de nanotechnologies, 91120, Palaiseau, France}

\author{Isabelle Sagnes}
\affiliation{Universit\'e Paris-Saclay, CNRS, Centre de Nanosciences et de nanotechnologies, 91120, Palaiseau, France}

\author{Petr Stepanov}
\affiliation{Quandela, 7 Rue Léonard de Vinci, 91300 Massy, France}

\author{Thi Huong Au}
\affiliation{Quandela, 7 Rue Léonard de Vinci, 91300 Massy, France}

\author{S\'ebastien Boissier}
\affiliation{Quandela, 7 Rue Léonard de Vinci, 91300 Massy, France}

\author{Niccolo Somaschi}
\affiliation{Quandela, 7 Rue Léonard de Vinci, 91300 Massy, France}

\author{Nicolas Maring}
\affiliation{Quandela, 7 Rue Léonard de Vinci, 91300 Massy, France}

\author{Nadia Belabas}
\thanks{These authors jointly supervised this work.}
\affiliation{Universit\'e Paris-Saclay, CNRS, Centre de Nanosciences et de nanotechnologies, 91120, Palaiseau, France}

\author{Shane Mansfield}
\thanks{These authors jointly supervised this work.}
\affiliation{Quandela, 7 Rue Léonard de Vinci, 91300 Massy, France}

\begin{abstract}

Reliable randomness is a core ingredient in algorithms and applications ranging from numerical simulations to statistical sampling and cryptography. The outcomes of measurements on entangled quantum states can violate Bell inequalities, thus guaranteeing their intrinsic randomness. This constitutes the basis for certified randomness generation. However, this certification requires spacelike separated devices, making it unfit for a compact apparatus. Here we provide a general method for certified randomness generation on a small-scale application-ready device and perform an integrated photonic demonstration combining a solid-state emitter and a glass chip. In contrast to most existing certification protocols, which in the absence of spacelike separation are vulnerable to loopholes inherent to realistic devices, the protocol we implement accounts for information leakage and is thus compatible with emerging compact scalable devices. We demonstrate a 2-qubit photonic device that achieves the highest standard in randomness yet is cut out for real-world applications. The full 94.5-hour-long stabilized process harnesses a bright and stable single-photon quantum-dot based source, feeding into a reconfigurable photonic chip, with stability in the milliradian range on the implemented phases and consistent indistinguishability of the entangled photons above 93\%. Using the contextuality framework, we certify private randomness generation and achieve a rate compatible with randomness expansion secure against quantum adversaries. 

\end{abstract}

\maketitle

\section{Introduction}

The strictest requirements on randomness sources are typically destined to cryptographic applications. There, randomness should ideally be both unpredictable and private, so that no information about the generated sequence can be gained by an eavesdropper either prior to or immediately after its generation. Quantum sources admit certification of these properties, by exploiting links between the unpredictability of quantum behaviour and the violation of Bell inequalities. A guarantee that numbers have been sampled from empirical data exhibiting Bell nonlocality \cite{BCP+2014Bell} or, more generally, contextuality \cite{kochen1975, AB2011Sheaf, CSW2014Graph, AFL+2015Combinatorial} can suffice to certify unpredictability and privacy~\cite{Colbeck2007Quantum,PAM+2010Random}.

Randomness certification and other Bell inequality based protocols offer attainable practical advantages for quantum information processing with relatively low numbers of qubits, but they are nevertheless susceptible to loopholes \cite{Larsson2014Loopholes}. One way to close the locality, or more generally, the compatibility loophole is to ensure space-like separation between the players of the non-local game~\cite{BKG+2018Experimentally, LZL+2018Device, ZSB+2020Experimental, LZL+2021Experimental, SZB+2021Device}. However, that is not an option for a practical compact device. Merely asserting that the relevant parts of the device are shielded~\cite{LLR+2021Device} is unsatisfactory for users who would like to prevent themselves against a device deteriorating with time. For such a device, the compatibility loophole must be carefully addressed, because crosstalk can lead to detrimental information flow between components. This compromises theoretical analyses and security proofs even outside of adversarial scenarios. More broadly, all future on-chip quantum information processing will be susceptible to such effects, for which reason it is essential that they be taken into account in protocols and algorithms at the information processing level.

In this work, we introduce novel theoretical tools to address the locality loophole, which we demonstrate in a randomness certification protocol performed on a compact 2-qubit photonic processor. Idealised analyses typically lead to relations between relevant figures of merit (such as fidelities, rates, and guessing probabilities) on the one hand, and Bell violations or more general contextuality measures \cite{ABM2017Contextual} on the other. Here, however, we provide relations suited to realistic devices, which allow the evaluation of the relevant figures of merit in terms of both beneficial contextuality and detrimental crosstalk. Moreover, we introduce a method to upper bound the amount of crosstalk by computing how far the device's observed behaviour is from the set of quantum correlations approximated by the Navascués--Pironio--Acín (NPA) hierarchy \cite{NPA2007Bounding}. This enables detection of adversarial manipulation of the device which may seek to exploit the locality loophole to spoof certification.

We propose a certification method that is secure against quantum side information, meeting the highest security standards. This method requires acquiring large statistics while maintaining high photon purity and indistinguishability \cite{Marcikic2004, Gonzalez2022}, which puts knock-on constraints on hardware efficiency and stability. Our theoretical contribution bridging the gap between ideal situations and realistic implementations, combined with finely controlled and robust hardware, allow us to implement the first on-chip certified quantum random number generation protocol with full security proof.

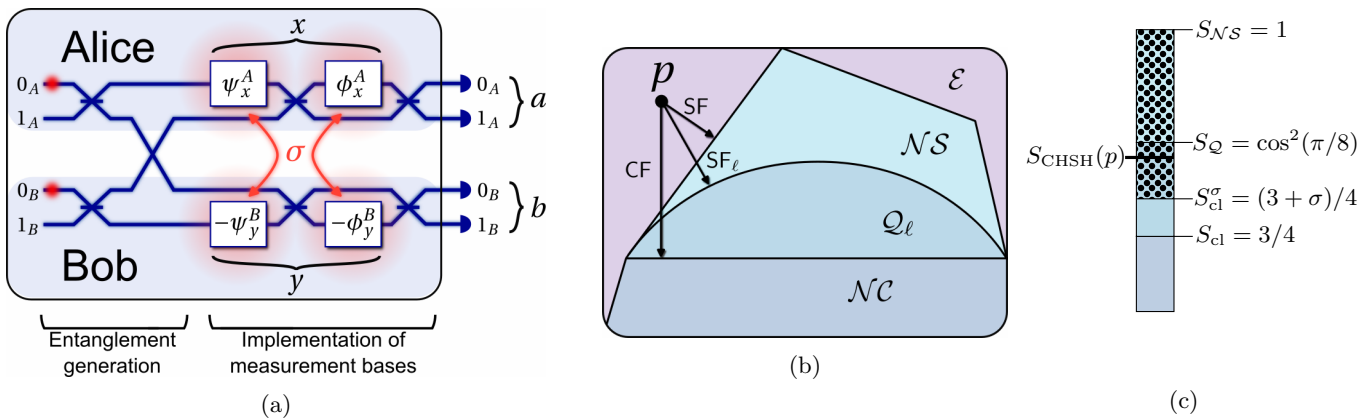
\begin{figure*}
\begin{subfigure}[c]{0.4\textwidth}
    \includegraphics[width=\textwidth]{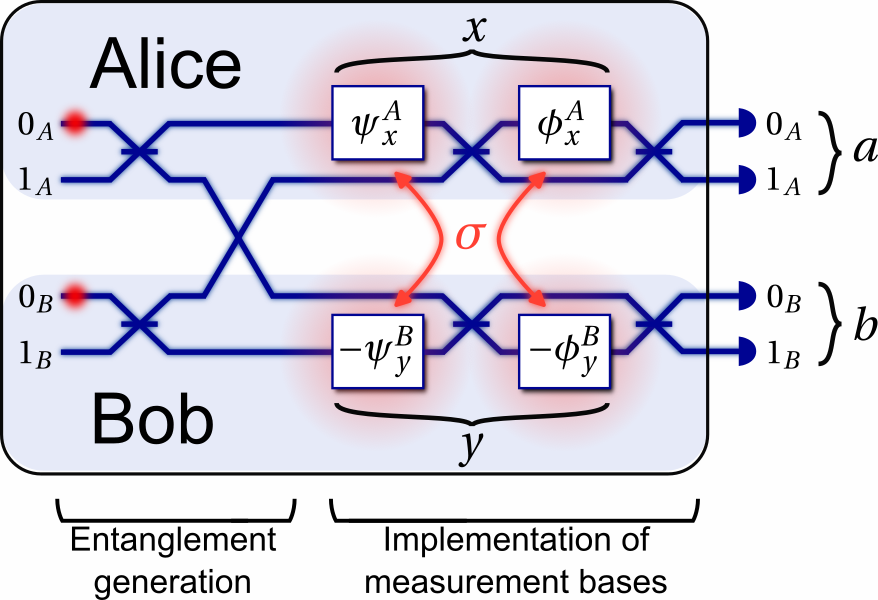}
    \caption{}
    \label{subfig:chip_scheme}
\end{subfigure}
\qquad
\begin{subfigure}[c]{0.3\textwidth}
    \includegraphics[width=\textwidth]{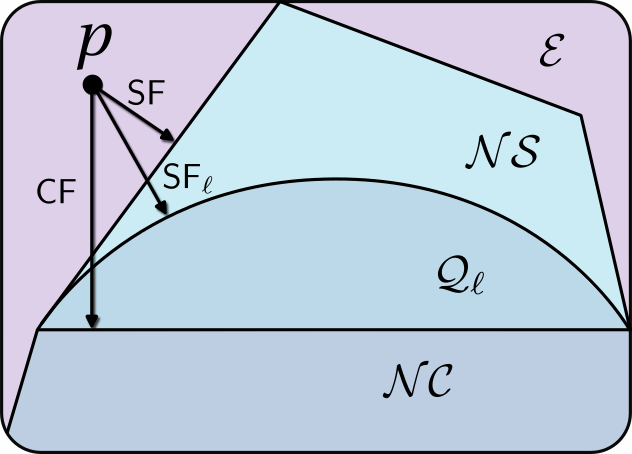}
    \caption{}
    \label{fig:NSSets}
\end{subfigure}
\begin{subfigure}[c]{0.25\textwidth}
    \begin{tikzpicture}[scale=0.5]
\draw[fill=noncontextual] plot coordinates{(0,1) (0,3) (1,3) (1,1)  }--cycle;
\draw[fill=npal] plot coordinates{(0,3) (0,5.5) (1,5.5) (1,3)  }--cycle;
\draw[fill=nonsignalling] plot coordinates{(0,5.5) (0,8.5) (1,8.5) (1,5.5)  }--cycle;
\node at (0,-0.4) {};
\node at (-3,0) {};
\node at (2.9,3) {$S_{\mathrm{cl}}=3/4$};
\node at (3.7,4) {$S_{\mathrm{cl}}^{\sigma}=(3+\sigma)/4$};
\node at (3.7,5.5) { $S_{\mathcal{Q}}=\cos^2(\pi/8)$};
\node at (2.8,8.5) {$S_{\mathcal{NS}}=1$};
\node at (-1.6,5.1) {$\SCHSH(p)$};
\draw (1,8.5) -- (1.5,8.5);
\draw (1,5.5) -- (1.5,5.5);
\draw (1,3) -- (1.5,3);
\draw[] (0,4) -- (1.5,4);
\draw[very thick] (-0.3,5.1) -- (1,5.1);
\fill[pattern={Dots[angle=45, radius=0.4mm]}] (0,4)--(0,8.5)--(1,8.5)--(1,4)--cycle;
\end{tikzpicture}
    \caption{}
    \label{subfig:scores}
\end{subfigure}
\caption{ 
\textbf{Bell test on a compact device.} 
    \textbf{(a)} CHSH game on a photonic chip. The glass chip uses two dual-rail encoded qubits: modes $0_A$ and $1_A$ correspond to Alice's qubit, modes $0_B$ and $1_B$ to Bob's qubit. The thick lines represent waveguides, the cross with horizontal lines represent symmetric beamsplitters.  
    The first part of the chip generates the Bell state $\frac{1}{\sqrt{2}}(\ket{0_A 1_B} + \ket{1_A 0_B})$ after post-selection on detection events to project into the dual-rail qubit basis (see Appendix~\hyperref[app:chsh_hom]{D} for details). The second part enables the agents to select the inputs $(x,y)$ of the Bell test, i.e. the measurement bases, by applying phases via thermo-optic phase shifters. $\sigma$ represents the crosstalk that can occur between Alice and Bob. The outputs $(a,b)$ correspond to where the photons are detected at the exit of the chip. The distributions of the outputs conditioned on the input is called the behaviour, denoted $p$. \textbf{(b)} Sketch of the set $\mathcal{E}$ of all possible behaviours and its subsets. The behaviour $p$ underlying the statistics of our implementation lies outside the no-signalling polytope $\NS$ because $\sigma$ is nonzero. $\SF$, $\CF$ and $\SF_{\ell}$ quantify how far $p$ is from, respectively, the no-signalling set $\NS$, the non-contextual set $\NC$, and the approximation of the quantum set defined by the $\ell$-th level of the Navascues-Pironio-Acín hierarchy $\mathcal{Q}_l$.  \textbf{(c)} Maximal scores for the CHSH game. Thanks to our bound relating contextuality and signalling, we know that a score above $S_{\mathrm{cl}}^{\sigma}$ (dotted region) cannot be achieved with a non-contextual behaviour even when an amount $\sigma$ of signalling is allowed: above that threshold, the behaviour $p$ has to be contextual and can be used to certify randomness. $S_{\mathrm{cl}}, S_{\mathcal{Q}}$ and $S_{\mathcal{NS}}$ are the maximal scores over, respectively, the non-contextual, the quantum and the no-signalling sets.}
\label{fig:crosstalk}
\end{figure*}

\section{Maximal scores in realistic contextual games}
\label{sec:theory}

In our protocol, two parties play the Clauser-Horne-Shimony-Holt (CHSH) game \cite{CHSH1969Proposed,CHTW2004Consequences} to guarantee the generation of randomness. The protocol is divided between test rounds that assess a Bell inequality violation and generation rounds that produce random numbers. During test rounds, Alice and Bob each perform one of two measurements
$x \in \{0,1\}$ for Alice and $y \in \{0,1\}$ for Bob. This yields one of two outputs
$a \in \{0,1\}$ for Alice and $b \in \{0,1\}$ for Bob. Following the terminology of contextuality, we call a set of measurement choices (or set of `inputs') a `context'. Our corresponding implementation is sketched in Fig.~\ref{subfig:chip_scheme} and detailed in Fig.~\ref{fig:setup}.

The winning condition for the CHSH game for a context $(x,y)$ and joint results $(a,b)$ is $a \oplus b = x\cdot y$. The probability of obtaining $(a,b)$ when measuring $(x,y)$ is denoted $p(ab|xy)$, and the set of the four conditional distributions $\{p(a,b|x,y)\}$ forms the behaviour of our device, denoted $p$. Its CHSH score $S_{\textrm{CHSH}}(p)$ and its CHSH inequality violation~\cite{CHSH1969Proposed} $I_{\textrm{CHSH}}(p)$ are defined as:
\begin{align}
    &S_{\textrm{CHSH}}(p)=\sum_{a,b,x,y} p(x,y) V(a,b,x,y) p(ab|xy), \label{eq:CHSH_ineq}
    \\
    &I_{\textrm{CHSH}}(p)=\expval{A_0B_0}+\expval{A_0B_1}+\expval{A_1B_0}-\expval{A_1B_1}, 
\label{eq:CHSH_score}
\end{align}
where $p(x,y)$ is the distribution of the measurement choices, $V$ is the scoring function for the game, defined by $V(a,b,x,y)=1$ if $a \oplus b = x\cdot y$ and $V(a,b,x,y)=0$ otherwise, and where $\expval{A_xB_y}$ $= p(a=b|x,y) - p(a\neq b|x,y)$. When the input distribution is uniform, $S_{\textrm{CHSH}}=(I_{\textrm{CHSH}}+4)/8$.
For a non-signalling behaviour, a value in the range $\SCHSH(p)>0.75$ ($I_{\textrm{CHSH}}(p)>2$) is a signature that $p$ is contextual.
For a given experimental setup, the set of all possible behaviours is often represented within a geometric space \cite{BCP+2014Bell}, in which behaviours with relevant properties typically form polytopes or convex subsets Fig.~\ref{fig:NSSets}.

In Bell-based quantum information processing, the key element to examine the intrinsic properties of a behaviour is its decomposition into hidden-variable models (HVMs). Operationally, one can think of an HVM as a description of the devices held by an eavesdropper that reproduces the observed behaviour, thus giving the eavesdropper more predictive power while being indistinguishable from the original behaviour for the users. A behaviour is non-contextual if it admits a factorisable HVM~\cite{AB2011Sheaf}, and we quantify the departure from the set of non-contextual behaviours with the contextual fraction $\CF$~\cite{ABM2017Contextual}. When no-signalling is guaranteed, a behaviour is non-contextual if and only if it can be decomposed as a mixture of deterministic models~\cite{Fine1982Hidden}. Conversely, this means that the underlying behaviour is inherently nondeterministic, and randomness can be certified from it, when one observes a score $S(p)$ higher than the non-contextual (or `classical') maximum $S_{\mathrm{cl}}$, defined as: 
\begin{equation} 
S_{\mathrm{cl}} := \max_p \{ S(p) \, \mid \, \CF(p) = 0 \} \, .
\label{eq:Scl}
\end{equation}
Note that the score $S(p)$ we define here is not limited to the CHSH game: this definition applies more generally to any contextual game, which can have more than two measurement choices, more than two measurement results, more than two players, and any scoring function mapping the input-output tuples to $\{0,1\}$. One can also define $S^{\xi}_{\textrm{cl}}$, maximal score over HVMs with contextual fraction at most $\xi$: 
\begin{equation}
S_{\mathrm{cl}}^{\xi} := \max_p \{S(p) \, \mid \, \CF(p) \leq \xi \}\,;
\label{eq:Sxi}
\end{equation}
and relate it to the consistency of the game~\cite{AH2012Logical}.

For every context, the scoring function of a game defines a logical formula on the outputs that has to be satisfied to win the game (for the CHSH game, $a \oplus b =1$ for inputs $(1,1)$ and $a \oplus b =0$ for the other inputs). A game is said to be $k$-consistent if at most $k$ formulae can be simultaneously satisfied by an assignment of the outputs. For CHSH, $a=b=0$ satisfies the formulae defined by all inputs except $(1,1)$, and that is the best one can do: the CHSH game is 3-consistent. 

\begin{prop}[Theorem 4 in~\cite{ABM2017Contextual}\footnote{Note that there is a typo in the statement of Theorem 4 in \cite{ABM2017Contextual}: the denominator on the right-hand side should be $n$, not $k$, which is equal to $m$ in our case.}]
Let $S_{\mathrm{cl}}^{\xi}$ be the maximal score for a contextual game with $m$ measurement choices that is $k$-consistent. Then 
\begin{equation}
    S_{\mathrm{cl}}^{\xi} \leq \frac{k+\xi(m-k) }{m}.
\label{eq:SclBound}
\end{equation}
\label{prop:SclBound}
\end{prop}

Unfortunately, the measurements performed in practical scenarios and on devices without space-like separation are not perfectly compatible~\cite{GHK+2021Incompatible}, which implies that the no-signalling conditions~\cite{GRW1980General} are not met and the tools that were developed to certify randomness~\cite{Colbeck2007Quantum, PAM+2010Random, MS2017Universal} do not apply in a straightforward way. To overcome this limitation, the set of admissible HVMs needs to be extended. One can for instance admit more general measurements~\cite{SPM2013Device} or more general behaviours~\cite{UZZ+2020Randomness}. Here, we allow a limited amount of signalling at the hidden-variable level, and we compute new maximal scores achievable over these models. More precisely, we define the signalling fraction $\SF$ (simultaneously introduced in~\cite{EBMM2022Corrected}) to quantify how signalling a behaviour is. Like the contextual fraction, $\SF$ is the solution of a linear program and can thus be computed efficiently. Its formal definitions can be found in Appendix \hyperref[app:fractions]{A}. We then define $S^{\sigma}_{x,y}$, maximal score over HVMs with signalling fraction at most $\sigma$, and the further requirement that the HVMs be deterministic on a specific context $(x,y)$, called distinguished context:
\begin{align}
S^{\sigma}_{x,y} := \max_p \{S(p) \mid & \ \exists \ (a,b). \, p(ab|xy)=1  \\
& \ \mathrm{ and } \ \, \SF(p) \leq \sigma\}\, , \nonumber
\label{eq:Ssigmaxy}
\end{align} 

The score with a distinguished context is the relevant quantity to certify a lower bound on randomness against quantum side information using the techniques introduced by Miller and Shi~\cite{MS2017Universal}. We extend the score definition to the case of nonzero signalling ($\sigma>0$), in order to design a protocol for realistic devices, and we relate $S^{\sigma}_{x,y}$ to $S_{\mathrm{cl}}^{\xi}$ with the following proposition. 

\begin{prop} 
For a contextual bipartite game with binary input choices for both players, let $(x,y)$ be a distinguished context. Then 
\begin{equation}
    S^{\sigma}_{x,y} \leq S^{\xi = \sigma}_{\mathrm{cl}}.
\label{eq:SclSxi}
\end{equation}
\label{prop:SclSxi}
\end{prop}
The proof can be found in Appendix~\hyperref[app:scores]{B}, where the bound is formulated in the more general case of $n$-partite games with binary inputs. This relation can be interpreted as follow: $S^{\sigma}_{x,y}$ is an extension of the distinguished context score to behaviours with signalling fraction at most $\sigma$; $S^{\xi}_{\mathrm{cl}}$ is an extension of the classical score to behaviours with contextual fraction at most $\xi$; the latter is an upper bound on the former, with $\xi=\sigma$. In other words, the extent to which an amount of signalling $\sigma$ can improve the distinguished context score is bounded by the extent to which an amount $\xi=\sigma$ of contextuality can improve the classical score. For our protocol to generate randomness, one has to observe a score above $S^{\sigma}_{x,y}$. In the CHSH case, combining Propositions~\ref{prop:SclBound} and~\ref{prop:SclSxi}, this means that $\SCHSH(p)>\frac{3+\sigma}{4}$ is a sufficient condition, as depicted in Fig.~\ref{subfig:scores}. Compared to Ref. \cite{UZZ+2020Randomness}, which also uses the spot-checking security proof of Miller and Shi~\cite{MS2017Universal} and takes into account imperfect compatibility between the measurements, our analysis is more general, as it applies to any $n$-partite nonlocal games with binary inputs.

\section{Estimating physical crosstalk $\sigma$}

\begin{figure*}
\includegraphics[width=\textwidth]{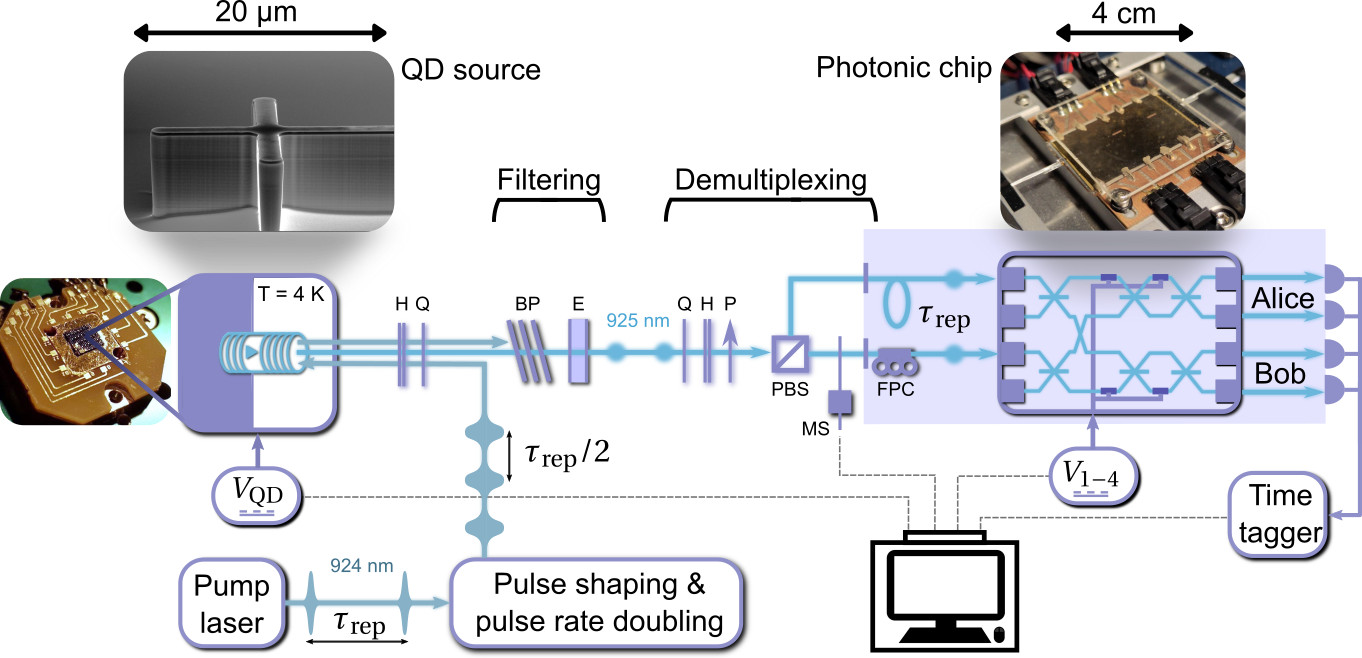}
\caption{
    \textbf{Compact implementation of a certified quantum random number generator.}
    The quantum dot (QD) photon emitter generates photons at 925 nm via a phonon-assisted excitation scheme (see Methods). H, Q: half and quarter-wave plates. BP:  bandpass filters. E: etalon. P: polarizer. After a demultiplexing stage (see Methods), the outputs of a polarizing beamsplitter (PBS) are collected with collimators. The setup is entirely fibered or waveguided in the blue area. A fibered delay $\tau_{\text{rep}}$, allows to synchronize pairs of photons sent into the photonic chip. A motorized shutter (MS) enables chip voltage calibration. The fibered polarisation controller (FPC), ensures both photons enter the photonic chip with the same polarisation. Dashed grey lines indicate that elements of the setup are automated to implement the randomness generation protocol, by adapting the voltage on the photon source for optimal brightness and periodic calibrations of the thermo-optic phase-shifter voltages. $V_{1-4}$ control the phases on chip and hence measurement bases of Alice and Bob. $V_\text{QD}$ feedback loop ensures the QD emission remains bright and the emitted photons indistinguishable.
    }
\label{fig:setup}
\end{figure*}

In order to define the set of admissible HVMs associated to our implementation, we allow a limited amount of information between the subsystems of Alice and Bob, due to crosstalk between the components (see Fig.\ref{fig:crosstalk}). This means that the admissible HVMs can have a positive signalling fraction $\sigma$, for some $\sigma$ that can be characterised either by an upstream partial characterisation of the devices or through an on-the-fly estimation based on the statistics of the inputs and outputs collected during the protocol. The first approach requires a device-dependent physical analysis of the setup, while the second approach is semi-device-independent: it derives a crosstalk estimate only from the observed input-output correlations, hence its device-independent characterisation, and then relate it to the underlying crosstalk via an assumption on the device. We follow the second one here and assume that $\sigma$ is related to the amount of signalling observed empirically. This assumption is well-founded if the devices were fabricated by an honest provider, i.e. were not programmed to act maliciously in order to function with a high level of crosstalk while keeping the empirically observable signalling low. In that case, an eavesdropper can only take advantage of flaws in the implementation and deterioration of the devices with time to try and predict the outputs. 

More precisely, we assume that our implementation obeys the laws of quantum mechanics, and we thus take $\sigma = \mathsf{SF}_{\ell}$, where $\mathsf{SF}_{\ell}$ is the extension of $\SF$ obtained by replacing the non-signalling set with the set of quantum correlations approximated at the $\ell^{th}$ level of the NPA hierarchy~\cite{NPA2007Bounding, NPA2008Convergent}. The reasoning for that choice is the following: if we instead took $\sigma = \mathsf{SF} \leq \SF_{\ell}$, an adversarial quantum strategy consisting in preparing a distribution $p_{\mathrm{adv}}$ such that $S(p_{\mathrm{adv}})$ is high, $\SF(p_{\mathrm{adv}})$ is low and the output distribution for the distinguished context is biased for the benefit of the adversary wouldn't be discarded by our analysis, while preparing such a behaviour quantumly requires more than $\sigma$ signalling. 

We give the formal definition of $\mathsf{SF}_{\ell}$ in Appendix~\hyperref[app:fractions]{A}, along with some properties of this measure and a comparison to the crosstalk measure introduced by Silman et al. in~\cite{SPM2013Device}. 

\section{Protocol for certified randomness generation}
\label{sec:MS}

We build upon Miller and Shi's spot-checking random number generation protocol~\cite{MS2017Universal} to compute a lower bound on the min-entropy of the total string of bits obtained during the execution of our protocol. This protocol is valid for general contextual games and is secure in the presence of the most general kind of information accessible to an eavesdropper, i.e.\ quantum side information. The complete description of this protocol is given in Appendix~\hyperref[app:protocol]{C}, along with the assumptions required for the min-entropy bound to be valid.  In that Appendix, we also describe how our idea can be combined with a protocol for randomness certification against classical side information such as the ones introduced in~\cite{PM2013Security, FGS2013Security, NBSP2018Device}, because in the trusted provider scenario this restriction on the side information is reasonable.
 
When the protocol succeeds, the min-entropy of the output sequences $\mathbf{AB}$ conditioned on the inputs sequence $\mathbf{XY}$ and all information potentially available to an eavesdropper $E$~\cite{Renner2005Security} satisfies 
\begin{equation}
 H^{\delta}_{\textrm{min}}(\mathbf{AB}|\mathbf{XY},E)  \geq N (\pi(\chi)-\Delta).
\label{eq:BoundMain}
\end{equation}
 
$H^{\delta}_{\textrm{min}}$ quantifies how many bits can be extracted from the outputs sequence of $N$ bits, such that these sifted bits are uniformly random and uncorrelated to the quantum side information held by the eavesdropper. $\chi$ is the score threshold fixed prior to the execution of the protocol, $\pi$ is the rate curve, and $\Delta$ is the correction term, whose expressions are given in Appendix~\hyperref[app:protocol]{C}.

\section{Description of our setup}
\label{sec:exp}

A scheme of the setup we use to implement our certified randomness generation is provided in Fig~\ref{fig:setup}: an electrically controlled semiconductor quantum dot in a 2-µm-diameter micropilar cavity generates single photons that are sent to a reconfigurable glass chip, implementing the CHSH game by varying the measurement context via optical phases and measuring output  coincidences. 

By a periodical calibration during the experiment, we maintain a high precision over the implemented measurement bases to limit signalling between the two agents, quantified by $\SF_{\ell}$. We use a bright and stable Quandela semiconductor quantum dot (QD) based single photon source \cite{Somaschi2016} that delivers indistinguishable single photons, allowing us to obtain high Bell inequality violations.

The polarized fibered-device brightness of our QD-based single-photon source, i.e.\ the probability to detect after the filtering stage with a polarizer an emitted photon following an excitation pulse, is $\SI{8.3(8)}{\%}$ (all error bars represent one standard deviation).

We quantify the purity of the single photons, i.e.\ the proportion of $\ket{1}$ Fock states compared to $\ket{2}$, with the second-order normalized correlation function $g^{(2)}(0)\approx \SI{2.31(3)}{\%}$ \cite{Loudon2000} and their indistinguishability with the Hong-Ou-Mandel (HOM) visibility $V_{\text{HOM}}=\SI{93.09(4)}{\%}$ \cite{Ollivier2021}.
The train of emitted photons is converted with a passive demultiplexing stage (see Methods) into pairs of photons entering simultaneously the photonic chip. 

On exiting the chip, the photons are detected by high-efficiency single photon detectors and time tagged. The overall transmission of the setup, i.e.\ the probability that a pump pulse results in a detected photon, is $\SI{2.7}{\%}$. 

The details about the source, the chip and the selection of the measurement bases can be in found in Methods. 

\section{Relation between photon distinguishability and CHSH violation}

\begin{figure}[h]
    \centering
    \includegraphics[width=\textwidth]{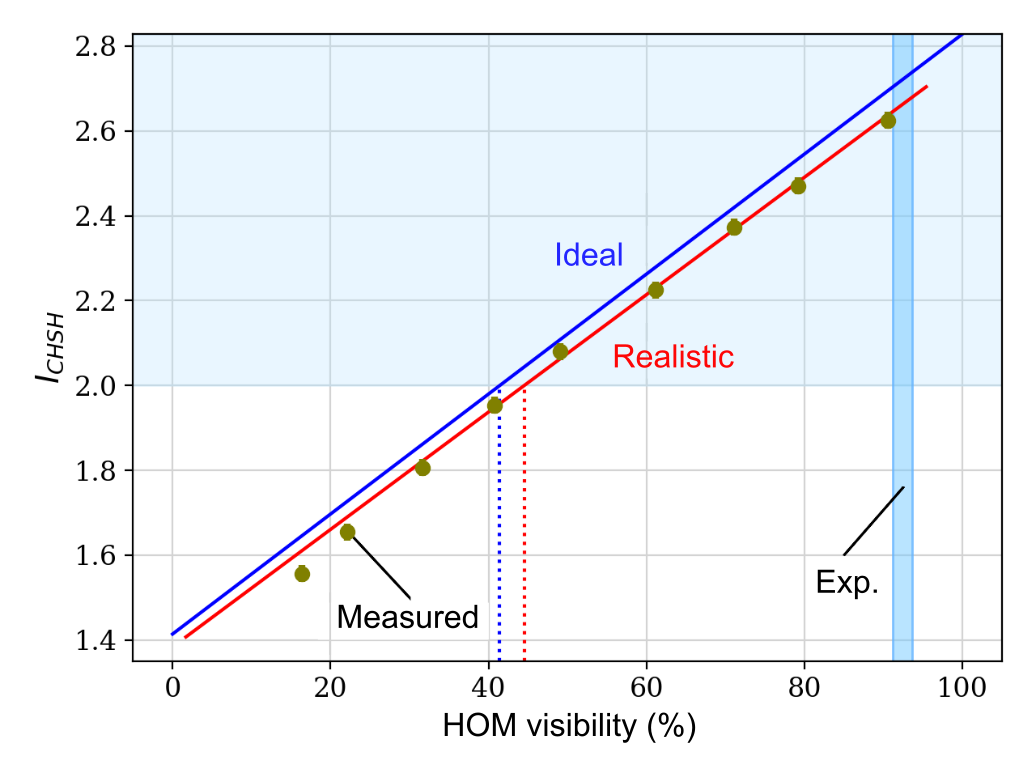}
    \caption{\textbf{CHSH expression value $I_\text{CHSH}$ as a function of HOM visibility} 
    Blue line: analytical dependence derived for a photon source emitting only pure single photons $g^{(2)}(0)=0$ in a lossless circuit  T=100\%. Red line: realistic curve simulated with the Perceval package for $g^{(2)}(0)=0.023$ and an optical circuit transmission of T=2.7 \% which are the experimental values. The blue and red dotted lines indicate the minimum HOM visibility required to certify quantum correlations for the ideal and realistic case respectively. The minimum values are respectively $\approx \SI{41.4}{\%}$ and $\approx \SI{44.5}{\%}$. Green dots: measured data points. The HOM visibility is decreased by changing the polarisation of one of the two photons entering the photonic chip. The error bars extending over $\pm 0.02$ representing one standard deviation are contained in the plot markers. Light blue area indicates CHSH violations certifying quantum correlations in the no-signaling case. Blue vertical bar indicates the range of HOM visibility values measured during our 94.5 hours main experiment. }
    \label{fig:chsh_vs_hom}
\end{figure}

Achieving certified randomness requires witnessing correlations that violate the CHSH inequality.
This places requirements on the purity and indistinguishability of the photons emitted by the single-photon source. For the ideal case of a source emitting only pure photons in a lossless optical circuit, we derive the relation 
\begin{equation}
I_\mathrm{CHSH} = \sqrt{2} (V_\mathrm{HOM}+1),
\end{equation}
where $V_\mathrm{HOM}$ is the HOM visibility of the photons. We derive this relation in Appendix~\hyperref[app:chsh_hom]{D}. We simulate the relationship between $I_\text{CHSH}$ and $V_\text{HOM}$ taking into account photon impurity and circuit losses with Perceval, a software platform specialised in simulations of photonic circuits in the discrete variable paradigm \cite{Heurtel2022}. In our case, the main causes of photon distinguishability arise from polarisation fluctuations in the optical fibers and charge noise around the QD \cite{Kuhlmann2013}. 

The comparison of the CHSH value obtained with our experimental setup and the simulated one is given in Fig.~\ref{fig:chsh_vs_hom}. We acquire $I_\text{CHSH}$ experimentally as a function of $V_\text{HOM}$. The photon distinguishability is adjusted by manually shifting the polarisation of one of the two photons entering the photonic chip with a fibered polarisation controller (see FPC in Fig. \ref{fig:setup}). 
$V_\text{HOM}$ is measured independently from $I_\text{CHSH}$ by setting Alice's on-chip interferometer phases $\psi^A=0$ and $\phi^A=-\pi/2$ on Fig.\ \ref{fig:setup}. Her interferometer then implements the following unitary matrix on the first two spatial modes (up to a global phase, see "Selection of the measurement bases" in Methods)

\begin{equation}
    \frac{1}{\sqrt{2}}
    \left[\begin{array}{ll}
    1 &
    1 \\
    1 &
    -1
    \end{array}\right],
\end{equation}
which is a 50:50 beamsplitter matrix like in the canonical Hong-Ou-Mandel experiment. The HOM visibility is then $V_\text{HOM} = 1 - 2p_\text{coinc}$ where $p_\text{coinc}$ is the probability of detecting simultaneously a photon on both of Alice's detectors in this configuration. The close match between simulated and measured data points validates the performance of our setup.

\begin{figure*}
\includegraphics[width=0.9\textwidth]{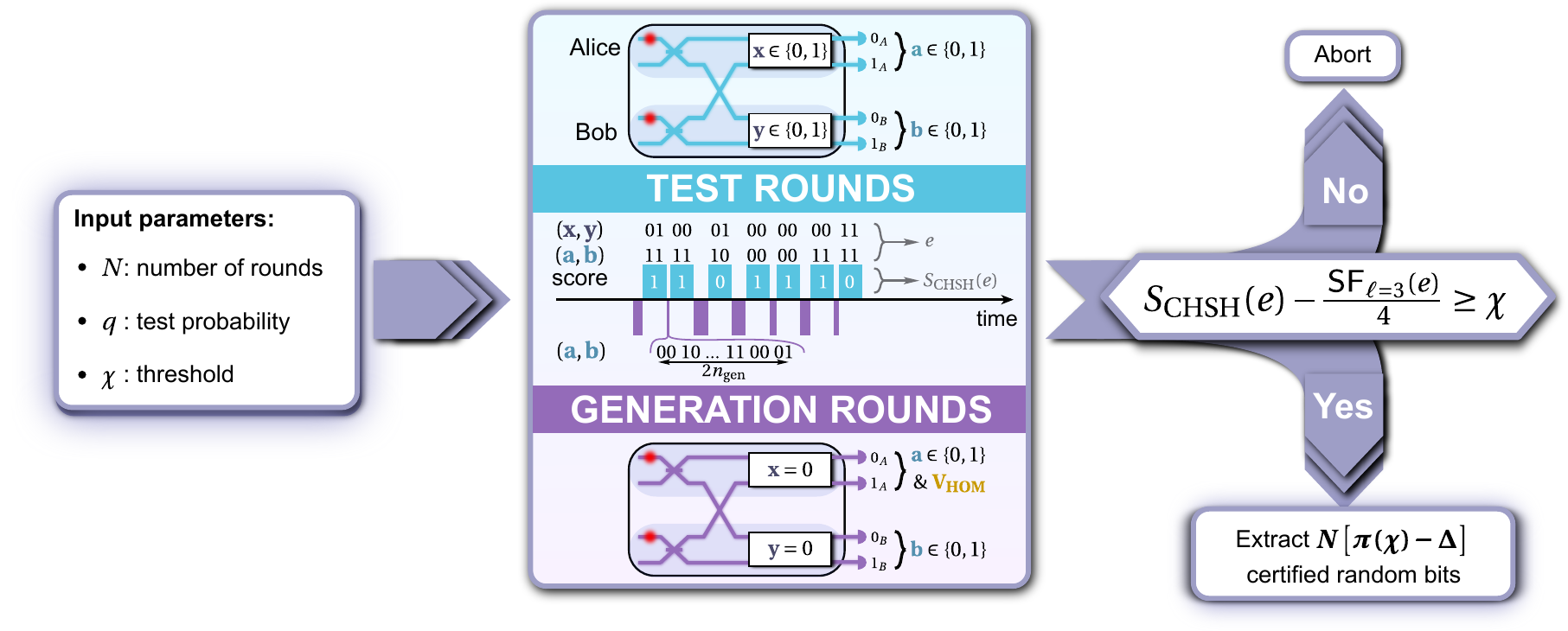}
\caption{
    \textbf{Certification protocol.}
    The protocol certifies privacy and unpredictability of the output sequence.
    The main loop consists in two steps: generation rounds (purple area) followed by a test round (blue area). The number of generation rounds $n_\text{gen}$ to acquire in each iteration of the main loop is determined by a geometric probability distribution of parameter $q$. Each generation round consists in measuring a coincidence $(a,b)$ in the generation context $(x=0,y=0)$. After $n_\text{gen}$ generation rounds, a measurement context $(x,y)$ is chosen randomly among (0,0), (0,1), (1,0) and (1,1), and a coincidence $(a,b)$ is measured in a test round. The test rounds are used to compute the CHSH game score $S_\text{CHSH}(e)=\frac{c}{Nq}$, where $e$ is the behaviour constructed from the test round measurements and $c$ is the sum of all individual test round scores. When the total number of rounds carried out reaches $N$, the protocol exits the main loop. If $S_\text{CHSH}(e)-\SF_{l=3} < \chi$ the protocol aborts, else we perform randomness extraction on the generation rounds to yield $N \left[ \pi(\chi)-\Delta \right]$ certified random bits. In parallel to the generation rounds, we monitor the photon indistinguishability via the HOM visibility $V_\text{HOM}.$  
   }
\label{fig:flowchart}
\end{figure*}

\section{Implementation of the protocol and experimental results}

The full protocol is described in Fig.~\ref{fig:flowchart}. Before each round of the experiment, Alice and Bob choose the measurement contexts $(x,y)$ via phases $(\psi_x^A,\phi_x^A)$ and $(\psi_y^B,\phi_y^B)$ indicated on Fig.\ \ref{subfig:chip_scheme}. The corresponding values of the phases are given in Methods, Table \ref{table:context_table}. Alice and Bob also agree that $(x=0, y=0)$ is the generation context, i.e.\ the measurement context in which the random number will be generated. We choose $\phi_0^A=-\pi/2$, so that $(x=0, y=0)$ allows us to acquire the HOM visibility in parallel to the generation rounds because Alice's interferometer acts as a symmetric beamsplitter in that configuration. For our experiment, $\psi_x^A=\psi_y^B=0$. To implement in practice on the chip a set of phases, we solve the phase-voltage matrix equation written in Appendix~\hyperref[app:crosstalk]{E} using standard optimization procedures. Alice (resp. Bob) maps her (his) result to `0' when a photon is detected in mode $0_A$ (resp.\ $0_B$) and to `1' when a photon is detected in mode $1_A$ (resp.\ $1_B$), according to the mode labelling of Fig.\ \ref{subfig:chip_scheme}.

The protocol alternates between generation rounds and test rounds. A round consists in the measurement of a coincidence between Alice and Bob, whose result is stored in the form 00, 01, 10 or 11, where the first bit describes Alice's result and the second one Bob's result.  For the generation rounds, Alice and Bob set their phase to the generation context, and for the test rounds, they choose randomly a measurement context among the four. The protocol outputs three binary sequences: generation round results, test round results and test round contexts. The first stores the coincidence results for the generation rounds and represents the sequence of random bits before randomness extraction. The test round results and contexts sequences are used to determine the experiment's behaviour $p$. From it, we can compute $S_\text{CHSH}(p)$ and 
$\SF(p)$.
We assume that the detected photons are representative of the whole optical setup behaviour, i.e.\ that the sampling is fair.

The predicted coincidence rate between Alice and Bob, which corresponds to the rounds processing rate is computed as follows:

\begin{align} 
&\SI{158e6}{} && \text{(pump laser pulse rate)} \nonumber\\
& \ \times (0.0266)^2 && \text{(overall photon transmission)} \nonumber\\
& \ \times 1/4 && \text{(passive demultiplexing)} \nonumber\\
& \ \times 1/2 && \text{(state post-selection)} \nonumber \\
& \ = \SI{14000(3000)}{\per\second}
\label{eq:coinc}
\end{align}

The measured coincidence rate is about $\SI{14200( 600)}{\per\second}$. Considering the thermalization waiting time of 250 ms after each measurement context switch, the expected protocol rounds processing rate from the measured coincidence rate is around $\SI{8300\pm 300}{\per\second}$, which is close to the measured rate of $\SI{7300}{\per\second}$. The disparity indicates an opportunity for improvement in the programming implementation of the randomness generation protocol.

The polarisation of photons reaching the polarisation-sensitive detectors exhibits temporal fluctuations, leading to asymmetric count rate variations on Alice's and Bob's pair of detectors. This causes a drift on the order of 0.25 mrad/h on the measured phases $\phi^A$ and $\phi^B$ of Alice's and Bob's interferometer.
To compensate this effect and additional phase errors, the voltages used to implement the corresponding phases for each measurement context with the heating resistors are calibrated at the beginning of the protocol and then after every 6 hours of operation (see Appendix~\hyperref[app:calibration]{F}). As a result, all phases stayed confined within an interval of 3 mrad around the target phases during an almost 100 hour-long run.

\begin{figure}
\includegraphics[width=\textwidth]{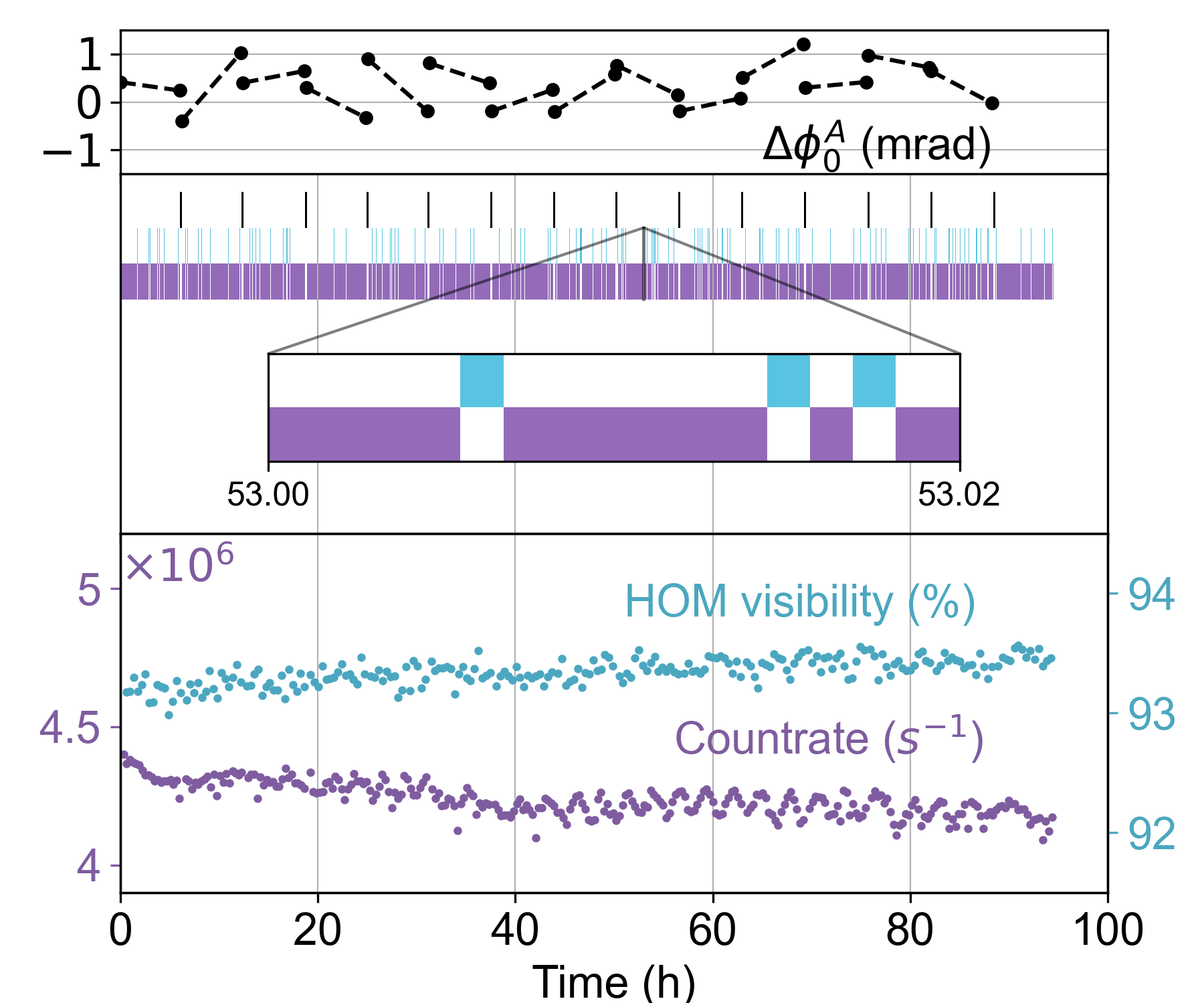}
\caption{
    \textbf{Stability, brightness and indistinguishability of photons emitted by the source during a 94.5-hour-long experiment.}
    Top plot: difference $\Delta \phi ^A_0 = \phi^A_0 - ( - \pi/2) $ between Alice's measured and target interferometer phase in context $(0,0)$. Phase is measured using the interferometer splitting before and after each context voltages calibration. Middle plot: typical sequence of calibration, test and generation sequences generated randomly using parameter values of the experiment. Bottom plot: total countrate and HOM visibility at output of the chip.}
\label{fig:run}
\end{figure}

Before data acquisition we measured a violation of the CHSH inequality $I_\text{CHSH}=2.68$, which we used to fix the optimal parameters for the protocol. The total number of rounds $N=2.4\times 10^9$ was determined by the desired duration of the acquisition together with the expected rate. We optimised the test probability $q$ (cf. Fig. \ref{fig:flowchart}) to maximise the min-entropy bound given by Eq.~\eqref{eq:BoundMain} and fixed $q=\SI{1.34e-4}{}$. As a result, the number of test rounds carried out to construct the behaviour of our device and estimate the CHSH violation is $3.2\times 10^5$. We obtained $I_{\text{CHSH}}=2.685$ on the test rounds (or $S_{\text{CHSH}}=0.8356$ in the CHSH game picture) and $\SF_3(e)= 0.005$. In Fig.~\ref{fig:run}, we present a few key figures recorded during the acquisition, showcasing the stability and precision of our setup.
According to Eq.~\eqref{eq:BoundMain}, these values certify that the $\delta$-smooth min-entropy $H^{\delta}_{\textrm{min}}(\mathbf{AB}|\mathbf{XY},E)$ of the obtained sequence of results $\mathbf{AB}$, conditioned on the sequence of input measurement choices $\mathbf{XY}$ and the quantum side information held by any potential eavesdropper $E$, is at least
\begin{equation}
 H^{\delta}_{\textrm{min}}(\mathbf{AB}|\mathbf{XY},E)  \geq \SI{7.21e6}{},
\end{equation}
where we chose as security parameter $\delta = 10^{-10}$. We can hence extract $\SI{7.21e6}{}$ random bits with a Toeplitz matrix hashing randomness extractor \cite{Krawczyk1995}. 

Note that this amount of certified randomness is compatible with randomness expansion, in the sense that if we were to use the interval algorithm to generate our strongly biased input bits from a small number of uniform bits~\cite{KY76Algorithms, HH1997Interval}, the input randomness required for our implementation would be: 
\begin{equation}
\textrm{rand}_{in}=N(h(q) + 2 q) = \SI{5,24e6}{},
\end{equation}
where $h$ is the binary Shannon entropy, which is smaller than the amount of randomness we generate at the output.

\section{Discussion and outlook}

This work focuses on certifying randomness with a high standard of security and a compact device, as is needed for real-world applications. Previous implementations of quantum random number generation on a chip focused on the device-dependent framework~\cite{BHQGbps2021} or did not provide a complete security proof~\cite{WPD2018Multidimensional}, while we report here the first on chip Bell-based randomness generation protocol with a complete security proof.
Improvements, both theoretical and experimental, would allow us to maintain this level of security at higher rates. 

On the theoretical side, the Miller--Shi protocol is valid for generic nonlocal and contextual games but it is not optimal. Recent results on how entropy accumulates \cite{DFR2020Entropy, DF2019Entropy} have been used to obtain optimal rates for randomness expansion with the CHSH game \cite{ARV2019Simple} and better rates for arbitrary nonlocal games \cite{BRC2020Framework}. The general framework of \cite{BRC2020Framework} in combination with new techniques for efficiently lower-bounding the conditional von Neumann entropy \cite{BFF2021Device} and with our randomness-versus-signalling tradeoff can lead to order-of-magnitude improvements in min-entropy rates. This will require extending the aforementioned results to contextuality scenarios with limited but nonzero signalling. We could also use a similar setup to certify randomness against classical side information only. In that case, biasing the input distribution is useful in two cases: when an output distribution for some choice of inputs contains more randomness than others, and when one aims to do  randomness expansion (that's the goal of spot-checking). When one wishes to simply generate private randomness, and when the randomness-bounding functions are identical for all input choices because of symmetries of the underlying behaviour, it is more favourable to generate randomness from all inputs. However, using a biased input distribution is still beneficial to avoid lengthy recalibration time between each round. By implementing a protocol against classical side information on a silicon nitride chip, where thermo-optic phase shifters react faster than in glass, and with an optimised choice of bias for the inputs, we could certify bits against classical side information at a rate of kbits/s after only a few minutes of acquisition.

On the hardware side, the efficiency of deterministic single photon sources does not have fundamental limits and can be increased while keeping single photon purity high. The main limitations are optical losses within the setup and the reconfiguration time of the photonic circuit. We estimate the experiment duration with short-term hardware improvements. With active photon demultiplexing, the photon coincidence rate can be doubled, and using a photonic chip featuring electro-optic \cite{Janner2009} or mechanical \cite{Quack2020} phase-shifters, the round processing rate equals the coincidence rate because it removes the need for thermalization waiting times. Commercially available single-photon detectors can reach 90 \% detection efficiency. For the single-photon source, these characteristics are achieved currently in the lab: 50 \% polarized first-lens brightness, $g^{(2)}(0) = 0.9 \%$ and $V_\text{HOM} = 94 \%$. As a consequence, the etalon filter is not needed. Factoring in all of these improvements, the total setup transmission increases from 2.66 \% to 10.5 \%, and the CHSH inequality value increases from 2.66 to 2.73. Repeating our randomness generation experiment with the same number of rounds to process, that is $N=\SI{2.4e9}{}$, it would take only 1.3 hours and generate $\SI{9e6}{}$ certified random bits, resulting in a kbits/s rate.

\bibliography{Bibliography}

\begin{thebibliography}{61}%
\makeatletter
\providecommand \@ifxundefined [1]{%
 \@ifx{#1\undefined}
}%
\providecommand \@ifnum [1]{%
 \ifnum #1\expandafter \@firstoftwo
 \else \expandafter \@secondoftwo
 \fi
}%
\providecommand \@ifx [1]{%
 \ifx #1\expandafter \@firstoftwo
 \else \expandafter \@secondoftwo
 \fi
}%
\providecommand \natexlab [1]{#1}%
\providecommand \enquote  [1]{``#1''}%
\providecommand \bibnamefont  [1]{#1}%
\providecommand \bibfnamefont [1]{#1}%
\providecommand \citenamefont [1]{#1}%
\providecommand \href@noop [0]{\@secondoftwo}%
\providecommand \href [0]{\begingroup \@sanitize@url \@href}%
\providecommand \@href[1]{\@@startlink{#1}\@@href}%
\providecommand \@@href[1]{\endgroup#1\@@endlink}%
\providecommand \@sanitize@url [0]{\catcode `\\12\catcode `\$12\catcode
  `\&12\catcode `\#12\catcode `\^12\catcode `\_12\catcode `\%12\relax}%
\providecommand \@@startlink[1]{}%
\providecommand \@@endlink[0]{}%
\providecommand \url  [0]{\begingroup\@sanitize@url \@url }%
\providecommand \@url [1]{\endgroup\@href {#1}{\urlprefix }}%
\providecommand \urlprefix  [0]{URL }%
\providecommand \Eprint [0]{\href }%
\providecommand \doibase [0]{https://doi.org/}%
\providecommand \selectlanguage [0]{\@gobble}%
\providecommand \bibinfo  [0]{\@secondoftwo}%
\providecommand \bibfield  [0]{\@secondoftwo}%
\providecommand \translation [1]{[#1]}%
\providecommand \BibitemOpen [0]{}%
\providecommand \bibitemStop [0]{}%
\providecommand \bibitemNoStop [0]{.\EOS\space}%
\providecommand \EOS [0]{\spacefactor3000\relax}%
\providecommand \BibitemShut  [1]{\csname bibitem#1\endcsname}%
\let\auto@bib@innerbib\@empty
\bibitem [{\citenamefont {Brunner}\ \emph {et~al.}(2014)\citenamefont
  {Brunner}, \citenamefont {Cavalcanti}, \citenamefont {Pironio}, \citenamefont
  {Scarani},\ and\ \citenamefont {Wehner}}]{BCP+2014Bell}%
  \BibitemOpen
  \bibfield  {author} {\bibinfo {author} {\bibfnamefont {N.}~\bibnamefont
  {Brunner}}, \bibinfo {author} {\bibfnamefont {D.}~\bibnamefont {Cavalcanti}},
  \bibinfo {author} {\bibfnamefont {S.}~\bibnamefont {Pironio}}, \bibinfo
  {author} {\bibfnamefont {V.}~\bibnamefont {Scarani}},\ and\ \bibinfo {author}
  {\bibfnamefont {S.}~\bibnamefont {Wehner}},\ }\bibfield  {title} {\bibinfo
  {title} {Bell nonlocality},\ }\href
  {https://doi.org/10.1103/RevModPhys.86.419} {\bibfield  {journal} {\bibinfo
  {journal} {Review of Modern Physics}\ }\textbf {\bibinfo {volume} {86}},\
  \bibinfo {pages} {419} (\bibinfo {year} {2014})}\BibitemShut {NoStop}%
\bibitem [{\citenamefont {Kochen}\ and\ \citenamefont
  {Specker}(1975)}]{kochen1975}%
  \BibitemOpen
  \bibfield  {author} {\bibinfo {author} {\bibfnamefont {S.}~\bibnamefont
  {Kochen}}\ and\ \bibinfo {author} {\bibfnamefont {E.~P.}\ \bibnamefont
  {Specker}},\ }\bibfield  {title} {\bibinfo {title} {The {{Problem}} of
  {{Hidden Variables}} in {{Quantum Mechanics}}},\ }in\ \href
  {https://doi.org/10.1007/978-94-010-1795-4_17} {\emph {\bibinfo {booktitle}
  {The {{Logico-Algebraic Approach}} to {{Quantum Mechanics}}}}},\ \bibinfo
  {series and number} {The {{University}} of {{Western Ontario Series}} in
  {{Philosophy}} of {{Science}}}\ (\bibinfo  {publisher} {{Springer,
  Dordrecht}},\ \bibinfo {year} {1975})\ pp.\ \bibinfo {pages}
  {293--328}\BibitemShut {NoStop}%
\bibitem [{\citenamefont {Abramsky}\ and\ \citenamefont
  {Brandenburger}(2011)}]{AB2011Sheaf}%
  \BibitemOpen
  \bibfield  {author} {\bibinfo {author} {\bibfnamefont {S.}~\bibnamefont
  {Abramsky}}\ and\ \bibinfo {author} {\bibfnamefont {A.}~\bibnamefont
  {Brandenburger}},\ }\bibfield  {title} {\bibinfo {title} {The sheaf-theoretic
  structure of non-locality and contextuality},\ }\href
  {https://doi.org/10.1088/1367-2630/13/11/113036} {\bibfield  {journal}
  {\bibinfo  {journal} {New Journal of Physics}\ }\textbf {\bibinfo {volume}
  {13}},\ \bibinfo {pages} {113036} (\bibinfo {year} {2011})}\BibitemShut
  {NoStop}%
\bibitem [{\citenamefont {Cabello}\ \emph {et~al.}(2014)\citenamefont
  {Cabello}, \citenamefont {Severini},\ and\ \citenamefont
  {Winter}}]{CSW2014Graph}%
  \BibitemOpen
  \bibfield  {author} {\bibinfo {author} {\bibfnamefont {A.}~\bibnamefont
  {Cabello}}, \bibinfo {author} {\bibfnamefont {S.}~\bibnamefont {Severini}},\
  and\ \bibinfo {author} {\bibfnamefont {A.}~\bibnamefont {Winter}},\
  }\bibfield  {title} {\bibinfo {title} {Graph-theoretic approach to quantum
  correlations},\ }\href {https://doi.org/10.1103/PhysRevLett.112.040401}
  {\bibfield  {journal} {\bibinfo  {journal} {Phys. Rev. Lett.}\ }\textbf
  {\bibinfo {volume} {112}},\ \bibinfo {pages} {040401} (\bibinfo {year}
  {2014})}\BibitemShut {NoStop}%
\bibitem [{\citenamefont {Ac{\'\i}n}\ \emph {et~al.}(2015)\citenamefont
  {Ac{\'\i}n}, \citenamefont {Fritz}, \citenamefont {Leverrier},\ and\
  \citenamefont {Sainz}}]{AFL+2015Combinatorial}%
  \BibitemOpen
  \bibfield  {author} {\bibinfo {author} {\bibfnamefont {A.}~\bibnamefont
  {Ac{\'\i}n}}, \bibinfo {author} {\bibfnamefont {T.}~\bibnamefont {Fritz}},
  \bibinfo {author} {\bibfnamefont {A.}~\bibnamefont {Leverrier}},\ and\
  \bibinfo {author} {\bibfnamefont {A.~B.}\ \bibnamefont {Sainz}},\ }\bibfield
  {title} {\bibinfo {title} {A combinatorial approach to nonlocality and
  contextuality},\ }\href {https://doi.org/10.1007/s00220-014-2260-1}
  {\bibfield  {journal} {\bibinfo  {journal} {Communications in Mathematical
  Physics}\ }\textbf {\bibinfo {volume} {334}},\ \bibinfo {pages} {533}
  (\bibinfo {year} {2015})}\BibitemShut {NoStop}%
\bibitem [{\citenamefont {Colbeck}(2007)}]{Colbeck2007Quantum}%
  \BibitemOpen
  \bibfield  {author} {\bibinfo {author} {\bibfnamefont {R.}~\bibnamefont
  {Colbeck}},\ }\emph {\bibinfo {title} {Quantum and Relativistic Protocols For
  Secure Multi-Party Computation}},\ \href@noop {} {Ph.D. thesis},\ \bibinfo
  {school} {University of Cambridge} (\bibinfo {year} {2007})\BibitemShut
  {NoStop}%
\bibitem [{\citenamefont {Pironio}\ \emph {et~al.}(2010)\citenamefont
  {Pironio}, \citenamefont {Ac{\'\i}n}, \citenamefont {Massar}, \citenamefont
  {de~la Giroday}, \citenamefont {Matsukevich}, \citenamefont {Maunz},
  \citenamefont {Olmschenk}, \citenamefont {Hayes}, \citenamefont {Luo},
  \citenamefont {Manning},\ and\ \citenamefont {Monroe}}]{PAM+2010Random}%
  \BibitemOpen
  \bibfield  {author} {\bibinfo {author} {\bibfnamefont {S.}~\bibnamefont
  {Pironio}}, \bibinfo {author} {\bibfnamefont {A.}~\bibnamefont {Ac{\'\i}n}},
  \bibinfo {author} {\bibfnamefont {S.}~\bibnamefont {Massar}}, \bibinfo
  {author} {\bibfnamefont {A.~B.}\ \bibnamefont {de~la Giroday}}, \bibinfo
  {author} {\bibfnamefont {D.~N.}\ \bibnamefont {Matsukevich}}, \bibinfo
  {author} {\bibfnamefont {P.}~\bibnamefont {Maunz}}, \bibinfo {author}
  {\bibfnamefont {S.}~\bibnamefont {Olmschenk}}, \bibinfo {author}
  {\bibfnamefont {D.}~\bibnamefont {Hayes}}, \bibinfo {author} {\bibfnamefont
  {L.}~\bibnamefont {Luo}}, \bibinfo {author} {\bibfnamefont {T.~A.}\
  \bibnamefont {Manning}},\ and\ \bibinfo {author} {\bibfnamefont
  {C.}~\bibnamefont {Monroe}},\ }\bibfield  {title} {\bibinfo {title} {Random
  numbers certified by bell's theorem},\ }\href
  {https://doi.org/10.1038/nature09008} {\bibfield  {journal} {\bibinfo
  {journal} {Nature}\ }\textbf {\bibinfo {volume} {464}},\ \bibinfo {pages}
  {1021} (\bibinfo {year} {2010})}\BibitemShut {NoStop}%
\bibitem [{\citenamefont {Larsson}(2014)}]{Larsson2014Loopholes}%
  \BibitemOpen
  \bibfield  {author} {\bibinfo {author} {\bibfnamefont {J.-{\AA}.}\
  \bibnamefont {Larsson}},\ }\bibfield  {title} {\bibinfo {title} {Loopholes in
  bell inequality tests of local realism},\ }\href
  {https://doi.org/10.1088/1751-8113/47/42/424003} {\bibfield  {journal}
  {\bibinfo  {journal} {Journal of Physics A: Mathematical and Theoretical}\
  }\textbf {\bibinfo {volume} {47}},\ \bibinfo {pages} {424003} (\bibinfo
  {year} {2014})}\BibitemShut {NoStop}%
\bibitem [{\citenamefont {Bierhorst}\ \emph {et~al.}(2018)\citenamefont
  {Bierhorst}, \citenamefont {Knill}, \citenamefont {Glancy}, \citenamefont
  {Zhang}, \citenamefont {Mink}, \citenamefont {Jordan}, \citenamefont
  {Rommal}, \citenamefont {Liu}, \citenamefont {Christensen}, \citenamefont
  {Nam}, \citenamefont {Stevens},\ and\ \citenamefont
  {Shalm}}]{BKG+2018Experimentally}%
  \BibitemOpen
  \bibfield  {author} {\bibinfo {author} {\bibfnamefont {P.}~\bibnamefont
  {Bierhorst}}, \bibinfo {author} {\bibfnamefont {E.}~\bibnamefont {Knill}},
  \bibinfo {author} {\bibfnamefont {S.}~\bibnamefont {Glancy}}, \bibinfo
  {author} {\bibfnamefont {Y.}~\bibnamefont {Zhang}}, \bibinfo {author}
  {\bibfnamefont {A.}~\bibnamefont {Mink}}, \bibinfo {author} {\bibfnamefont
  {S.}~\bibnamefont {Jordan}}, \bibinfo {author} {\bibfnamefont
  {A.}~\bibnamefont {Rommal}}, \bibinfo {author} {\bibfnamefont {Y.-K.}\
  \bibnamefont {Liu}}, \bibinfo {author} {\bibfnamefont {B.}~\bibnamefont
  {Christensen}}, \bibinfo {author} {\bibfnamefont {S.~W.}\ \bibnamefont
  {Nam}}, \bibinfo {author} {\bibfnamefont {M.~J.}\ \bibnamefont {Stevens}},\
  and\ \bibinfo {author} {\bibfnamefont {L.~K.}\ \bibnamefont {Shalm}},\
  }\bibfield  {title} {\bibinfo {title} {Experimentally generated randomness
  certified by the impossibility of superluminal signals},\ }\href
  {https://doi.org/10.1038/s41586-018-0019-0} {\bibfield  {journal} {\bibinfo
  {journal} {Nature}\ }\textbf {\bibinfo {volume} {556}},\ \bibinfo {pages}
  {223} (\bibinfo {year} {2018})}\BibitemShut {NoStop}%
\bibitem [{\citenamefont {Liu}\ \emph {et~al.}(2018)\citenamefont {Liu},
  \citenamefont {Zhao}, \citenamefont {Li}, \citenamefont {Guan}, \citenamefont
  {Zhang}, \citenamefont {Bai}, \citenamefont {Zhang}, \citenamefont {Liu},
  \citenamefont {Wu}, \citenamefont {Yuan}, \citenamefont {Li}, \citenamefont
  {Munro}, \citenamefont {Wang}, \citenamefont {You}, \citenamefont {Zhang},
  \citenamefont {Ma}, \citenamefont {Fan}, \citenamefont {Zhang},\ and\
  \citenamefont {Pan}}]{LZL+2018Device}%
  \BibitemOpen
  \bibfield  {author} {\bibinfo {author} {\bibfnamefont {Y.}~\bibnamefont
  {Liu}}, \bibinfo {author} {\bibfnamefont {Q.}~\bibnamefont {Zhao}}, \bibinfo
  {author} {\bibfnamefont {M.-H.}\ \bibnamefont {Li}}, \bibinfo {author}
  {\bibfnamefont {J.-Y.}\ \bibnamefont {Guan}}, \bibinfo {author}
  {\bibfnamefont {Y.}~\bibnamefont {Zhang}}, \bibinfo {author} {\bibfnamefont
  {B.}~\bibnamefont {Bai}}, \bibinfo {author} {\bibfnamefont {W.}~\bibnamefont
  {Zhang}}, \bibinfo {author} {\bibfnamefont {W.-Z.}\ \bibnamefont {Liu}},
  \bibinfo {author} {\bibfnamefont {C.}~\bibnamefont {Wu}}, \bibinfo {author}
  {\bibfnamefont {X.}~\bibnamefont {Yuan}}, \bibinfo {author} {\bibfnamefont
  {H.}~\bibnamefont {Li}}, \bibinfo {author} {\bibfnamefont {W.~J.}\
  \bibnamefont {Munro}}, \bibinfo {author} {\bibfnamefont {Z.}~\bibnamefont
  {Wang}}, \bibinfo {author} {\bibfnamefont {L.}~\bibnamefont {You}}, \bibinfo
  {author} {\bibfnamefont {J.}~\bibnamefont {Zhang}}, \bibinfo {author}
  {\bibfnamefont {X.}~\bibnamefont {Ma}}, \bibinfo {author} {\bibfnamefont
  {J.}~\bibnamefont {Fan}}, \bibinfo {author} {\bibfnamefont {Q.}~\bibnamefont
  {Zhang}},\ and\ \bibinfo {author} {\bibfnamefont {J.-W.}\ \bibnamefont
  {Pan}},\ }\bibfield  {title} {\bibinfo {title} {Device-independent quantum
  random-number generation},\ }\href
  {https://doi.org/10.1038/s41586-018-0559-3} {\bibfield  {journal} {\bibinfo
  {journal} {Nature}\ }\textbf {\bibinfo {volume} {562}},\ \bibinfo {pages}
  {548} (\bibinfo {year} {2018})}\BibitemShut {NoStop}%
\bibitem [{\citenamefont {Zhang}\ \emph {et~al.}(2020)\citenamefont {Zhang},
  \citenamefont {Shalm}, \citenamefont {Bienfang}, \citenamefont {Stevens},
  \citenamefont {Mazurek}, \citenamefont {Nam}, \citenamefont {Abell{\'a}n},
  \citenamefont {Amaya}, \citenamefont {Mitchell}, \citenamefont {Fu},
  \citenamefont {Miller}, \citenamefont {Mink},\ and\ \citenamefont
  {Knill}}]{ZSB+2020Experimental}%
  \BibitemOpen
  \bibfield  {author} {\bibinfo {author} {\bibfnamefont {Y.}~\bibnamefont
  {Zhang}}, \bibinfo {author} {\bibfnamefont {L.~K.}\ \bibnamefont {Shalm}},
  \bibinfo {author} {\bibfnamefont {J.~C.}\ \bibnamefont {Bienfang}}, \bibinfo
  {author} {\bibfnamefont {M.~J.}\ \bibnamefont {Stevens}}, \bibinfo {author}
  {\bibfnamefont {M.~D.}\ \bibnamefont {Mazurek}}, \bibinfo {author}
  {\bibfnamefont {S.~W.}\ \bibnamefont {Nam}}, \bibinfo {author} {\bibfnamefont
  {C.}~\bibnamefont {Abell{\'a}n}}, \bibinfo {author} {\bibfnamefont
  {W.}~\bibnamefont {Amaya}}, \bibinfo {author} {\bibfnamefont {M.~W.}\
  \bibnamefont {Mitchell}}, \bibinfo {author} {\bibfnamefont {H.}~\bibnamefont
  {Fu}}, \bibinfo {author} {\bibfnamefont {C.~A.}\ \bibnamefont {Miller}},
  \bibinfo {author} {\bibfnamefont {A.}~\bibnamefont {Mink}},\ and\ \bibinfo
  {author} {\bibfnamefont {E.}~\bibnamefont {Knill}},\ }\bibfield  {title}
  {\bibinfo {title} {Experimental {{Low-Latency Device-Independent Quantum
  Randomness}}},\ }\href {https://doi.org/10.1103/PhysRevLett.124.010505}
  {\bibfield  {journal} {\bibinfo  {journal} {Physical Review Letters}\
  }\textbf {\bibinfo {volume} {124}},\ \bibinfo {pages} {010505} (\bibinfo
  {year} {2020})}\BibitemShut {NoStop}%
\bibitem [{\citenamefont {Li}\ \emph {et~al.}(2021)\citenamefont {Li},
  \citenamefont {Zhang}, \citenamefont {Liu}, \citenamefont {Zhao},
  \citenamefont {Bai}, \citenamefont {Liu}, \citenamefont {Zhao}, \citenamefont
  {Peng}, \citenamefont {Zhang}, \citenamefont {Zhang}, \citenamefont {Munro},
  \citenamefont {Ma}, \citenamefont {Zhang}, \citenamefont {Fan},\ and\
  \citenamefont {Pan}}]{LZL+2021Experimental}%
  \BibitemOpen
  \bibfield  {author} {\bibinfo {author} {\bibfnamefont {M.-H.}\ \bibnamefont
  {Li}}, \bibinfo {author} {\bibfnamefont {X.}~\bibnamefont {Zhang}}, \bibinfo
  {author} {\bibfnamefont {W.-Z.}\ \bibnamefont {Liu}}, \bibinfo {author}
  {\bibfnamefont {S.-R.}\ \bibnamefont {Zhao}}, \bibinfo {author}
  {\bibfnamefont {B.}~\bibnamefont {Bai}}, \bibinfo {author} {\bibfnamefont
  {Y.}~\bibnamefont {Liu}}, \bibinfo {author} {\bibfnamefont {Q.}~\bibnamefont
  {Zhao}}, \bibinfo {author} {\bibfnamefont {Y.}~\bibnamefont {Peng}}, \bibinfo
  {author} {\bibfnamefont {J.}~\bibnamefont {Zhang}}, \bibinfo {author}
  {\bibfnamefont {Y.}~\bibnamefont {Zhang}}, \bibinfo {author} {\bibfnamefont
  {W.~J.}\ \bibnamefont {Munro}}, \bibinfo {author} {\bibfnamefont
  {X.}~\bibnamefont {Ma}}, \bibinfo {author} {\bibfnamefont {Q.}~\bibnamefont
  {Zhang}}, \bibinfo {author} {\bibfnamefont {J.}~\bibnamefont {Fan}},\ and\
  \bibinfo {author} {\bibfnamefont {J.-W.}\ \bibnamefont {Pan}},\ }\bibfield
  {title} {\bibinfo {title} {Experimental {{Realization}} of
  {{Device-Independent Quantum Randomness Expansion}}},\ }\href
  {https://doi.org/10.1103/PhysRevLett.126.050503} {\bibfield  {journal}
  {\bibinfo  {journal} {Physical Review Letters}\ }\textbf {\bibinfo {volume}
  {126}},\ \bibinfo {pages} {050503} (\bibinfo {year} {2021})}\BibitemShut
  {NoStop}%
\bibitem [{\citenamefont {Shalm}\ \emph {et~al.}(2021)\citenamefont {Shalm},
  \citenamefont {Zhang}, \citenamefont {Bienfang}, \citenamefont {Schlager},
  \citenamefont {Stevens}, \citenamefont {Mazurek}, \citenamefont
  {Abell{\'a}n}, \citenamefont {Amaya}, \citenamefont {Mitchell}, \citenamefont
  {Alhejji}, \citenamefont {Fu}, \citenamefont {Ornstein}, \citenamefont
  {Mirin}, \citenamefont {Nam},\ and\ \citenamefont {Knill}}]{SZB+2021Device}%
  \BibitemOpen
  \bibfield  {author} {\bibinfo {author} {\bibfnamefont {L.~K.}\ \bibnamefont
  {Shalm}}, \bibinfo {author} {\bibfnamefont {Y.}~\bibnamefont {Zhang}},
  \bibinfo {author} {\bibfnamefont {J.~C.}\ \bibnamefont {Bienfang}}, \bibinfo
  {author} {\bibfnamefont {C.}~\bibnamefont {Schlager}}, \bibinfo {author}
  {\bibfnamefont {M.~J.}\ \bibnamefont {Stevens}}, \bibinfo {author}
  {\bibfnamefont {M.~D.}\ \bibnamefont {Mazurek}}, \bibinfo {author}
  {\bibfnamefont {C.}~\bibnamefont {Abell{\'a}n}}, \bibinfo {author}
  {\bibfnamefont {W.}~\bibnamefont {Amaya}}, \bibinfo {author} {\bibfnamefont
  {M.~W.}\ \bibnamefont {Mitchell}}, \bibinfo {author} {\bibfnamefont {M.~A.}\
  \bibnamefont {Alhejji}}, \bibinfo {author} {\bibfnamefont {H.}~\bibnamefont
  {Fu}}, \bibinfo {author} {\bibfnamefont {J.}~\bibnamefont {Ornstein}},
  \bibinfo {author} {\bibfnamefont {R.~P.}\ \bibnamefont {Mirin}}, \bibinfo
  {author} {\bibfnamefont {S.~W.}\ \bibnamefont {Nam}},\ and\ \bibinfo {author}
  {\bibfnamefont {E.}~\bibnamefont {Knill}},\ }\bibfield  {title} {\bibinfo
  {title} {Device-independent randomness expansion with entangled photons},\
  }\href {https://doi.org/10.1038/s41567-020-01153-4} {\bibfield  {journal}
  {\bibinfo  {journal} {Nature Physics}\ }\textbf {\bibinfo {volume} {17}},\
  \bibinfo {pages} {452} (\bibinfo {year} {2021})}\BibitemShut {NoStop}%
\bibitem [{\citenamefont {Liu}\ \emph {et~al.}(2021)\citenamefont {Liu},
  \citenamefont {Li}, \citenamefont {Ragy}, \citenamefont {Zhao}, \citenamefont
  {Bai}, \citenamefont {Liu}, \citenamefont {Brown}, \citenamefont {Zhang},
  \citenamefont {Colbeck}, \citenamefont {Fan}, \citenamefont {Zhang},\ and\
  \citenamefont {Pan}}]{LLR+2021Device}%
  \BibitemOpen
  \bibfield  {author} {\bibinfo {author} {\bibfnamefont {W.-Z.}\ \bibnamefont
  {Liu}}, \bibinfo {author} {\bibfnamefont {M.-H.}\ \bibnamefont {Li}},
  \bibinfo {author} {\bibfnamefont {S.}~\bibnamefont {Ragy}}, \bibinfo {author}
  {\bibfnamefont {S.-R.}\ \bibnamefont {Zhao}}, \bibinfo {author}
  {\bibfnamefont {B.}~\bibnamefont {Bai}}, \bibinfo {author} {\bibfnamefont
  {Y.}~\bibnamefont {Liu}}, \bibinfo {author} {\bibfnamefont {P.~J.}\
  \bibnamefont {Brown}}, \bibinfo {author} {\bibfnamefont {J.}~\bibnamefont
  {Zhang}}, \bibinfo {author} {\bibfnamefont {R.}~\bibnamefont {Colbeck}},
  \bibinfo {author} {\bibfnamefont {J.}~\bibnamefont {Fan}}, \bibinfo {author}
  {\bibfnamefont {Q.}~\bibnamefont {Zhang}},\ and\ \bibinfo {author}
  {\bibfnamefont {J.-W.}\ \bibnamefont {Pan}},\ }\bibfield  {title} {\bibinfo
  {title} {Device-independent randomness expansion against quantum side
  information},\ }\href {https://doi.org/10.1038/s41567-020-01147-2} {\bibfield
   {journal} {\bibinfo  {journal} {Nature Physics}\ }\textbf {\bibinfo {volume}
  {17}},\ \bibinfo {pages} {448} (\bibinfo {year} {2021})}\BibitemShut
  {NoStop}%
\bibitem [{\citenamefont {Abramsky}\ \emph {et~al.}(2017)\citenamefont
  {Abramsky}, \citenamefont {Barbosa},\ and\ \citenamefont
  {Mansfield}}]{ABM2017Contextual}%
  \BibitemOpen
  \bibfield  {author} {\bibinfo {author} {\bibfnamefont {S.}~\bibnamefont
  {Abramsky}}, \bibinfo {author} {\bibfnamefont {R.~S.}\ \bibnamefont
  {Barbosa}},\ and\ \bibinfo {author} {\bibfnamefont {S.}~\bibnamefont
  {Mansfield}},\ }\bibfield  {title} {\bibinfo {title} {Contextual fraction as
  a measure of contextuality},\ }\href
  {https://doi.org/10.1103/PhysRevLett.119.050504} {\bibfield  {journal}
  {\bibinfo  {journal} {Physical Review Letters}\ }\textbf {\bibinfo {volume}
  {119}},\ \bibinfo {pages} {050504} (\bibinfo {year} {2017})}\BibitemShut
  {NoStop}%
\bibitem [{\citenamefont {Navascu\'es}\ \emph {et~al.}(2007)\citenamefont
  {Navascu\'es}, \citenamefont {Pironio},\ and\ \citenamefont
  {Ac\'{\i}n}}]{NPA2007Bounding}%
  \BibitemOpen
  \bibfield  {author} {\bibinfo {author} {\bibfnamefont {M.}~\bibnamefont
  {Navascu\'es}}, \bibinfo {author} {\bibfnamefont {S.}~\bibnamefont
  {Pironio}},\ and\ \bibinfo {author} {\bibfnamefont {A.}~\bibnamefont
  {Ac\'{\i}n}},\ }\bibfield  {title} {\bibinfo {title} {Bounding the set of
  quantum correlations},\ }\href
  {https://doi.org/10.1103/PhysRevLett.98.010401} {\bibfield  {journal}
  {\bibinfo  {journal} {Phys. Rev. Lett.}\ }\textbf {\bibinfo {volume} {98}},\
  \bibinfo {pages} {010401} (\bibinfo {year} {2007})}\BibitemShut {NoStop}%
\bibitem [{\citenamefont {Marcikic}\ \emph {et~al.}(2004)\citenamefont
  {Marcikic}, \citenamefont {de~Riedmatten}, \citenamefont {Tittel},
  \citenamefont {Zbinden}, \citenamefont {Legr\'e},\ and\ \citenamefont
  {Gisin}}]{Marcikic2004}%
  \BibitemOpen
  \bibfield  {author} {\bibinfo {author} {\bibfnamefont {I.}~\bibnamefont
  {Marcikic}}, \bibinfo {author} {\bibfnamefont {H.}~\bibnamefont
  {de~Riedmatten}}, \bibinfo {author} {\bibfnamefont {W.}~\bibnamefont
  {Tittel}}, \bibinfo {author} {\bibfnamefont {H.}~\bibnamefont {Zbinden}},
  \bibinfo {author} {\bibfnamefont {M.}~\bibnamefont {Legr\'e}},\ and\ \bibinfo
  {author} {\bibfnamefont {N.}~\bibnamefont {Gisin}},\ }\bibfield  {title}
  {\bibinfo {title} {Distribution of time-bin entangled qubits over 50 km of
  optical fiber},\ }\href {https://doi.org/10.1103/PhysRevLett.93.180502}
  {\bibfield  {journal} {\bibinfo  {journal} {Phys. Rev. Lett.}\ }\textbf
  {\bibinfo {volume} {93}},\ \bibinfo {pages} {180502} (\bibinfo {year}
  {2004})}\BibitemShut {NoStop}%
\bibitem [{\citenamefont {González-Ruiz}\ \emph {et~al.}(2022)\citenamefont
  {González-Ruiz}, \citenamefont {Das}, \citenamefont {Lodahl},\ and\
  \citenamefont {Sørensen}}]{Gonzalez2022}%
  \BibitemOpen
  \bibfield  {author} {\bibinfo {author} {\bibfnamefont {E.~M.}\ \bibnamefont
  {González-Ruiz}}, \bibinfo {author} {\bibfnamefont {S.~K.}\ \bibnamefont
  {Das}}, \bibinfo {author} {\bibfnamefont {P.}~\bibnamefont {Lodahl}},\ and\
  \bibinfo {author} {\bibfnamefont {A.~S.}\ \bibnamefont {Sørensen}},\
  }\bibfield  {title} {\bibinfo {title} {Violation of {Bell}'s inequality with
  quantum-dot single-photon sources},\ }\href
  {https://doi.org/10.1103/PhysRevA.106.012222} {\bibfield  {journal} {\bibinfo
   {journal} {Physical Review A}\ }\textbf {\bibinfo {volume} {106}},\ \bibinfo
  {pages} {012222} (\bibinfo {year} {2022})},\ \bibinfo {note}
  {arXiv:2109.14712 [quant-ph]}\BibitemShut {NoStop}%
\bibitem [{\citenamefont {Clauser}\ \emph {et~al.}(1969)\citenamefont
  {Clauser}, \citenamefont {Horne}, \citenamefont {Shimony},\ and\
  \citenamefont {Holt}}]{CHSH1969Proposed}%
  \BibitemOpen
  \bibfield  {author} {\bibinfo {author} {\bibfnamefont {J.~F.}\ \bibnamefont
  {Clauser}}, \bibinfo {author} {\bibfnamefont {M.~A.}\ \bibnamefont {Horne}},
  \bibinfo {author} {\bibfnamefont {A.}~\bibnamefont {Shimony}},\ and\ \bibinfo
  {author} {\bibfnamefont {R.~A.}\ \bibnamefont {Holt}},\ }\bibfield  {title}
  {\bibinfo {title} {Proposed experiment to test local hidden-variable
  theories},\ }\href {https://doi.org/10.1103/PhysRevLett.23.880} {\bibfield
  {journal} {\bibinfo  {journal} {Physical Review Letters}\ }\textbf {\bibinfo
  {volume} {23}},\ \bibinfo {pages} {880} (\bibinfo {year} {1969})}\BibitemShut
  {NoStop}%
\bibitem [{\citenamefont {Cleve}\ \emph {et~al.}(2004)\citenamefont {Cleve},
  \citenamefont {Hoyer}, \citenamefont {Toner},\ and\ \citenamefont
  {Watrous}}]{CHTW2004Consequences}%
  \BibitemOpen
  \bibfield  {author} {\bibinfo {author} {\bibfnamefont {R.}~\bibnamefont
  {Cleve}}, \bibinfo {author} {\bibfnamefont {P.}~\bibnamefont {Hoyer}},
  \bibinfo {author} {\bibfnamefont {B.}~\bibnamefont {Toner}},\ and\ \bibinfo
  {author} {\bibfnamefont {J.}~\bibnamefont {Watrous}},\ }\bibfield  {title}
  {\bibinfo {title} {Consequences and limits of nonlocal strategies}} (\bibinfo
  {year} {2004}),\ \bibinfo {note} {arXiv:0404076}\BibitemShut {NoStop}%
\bibitem [{\citenamefont {Fine}(1982)}]{Fine1982Hidden}%
  \BibitemOpen
  \bibfield  {author} {\bibinfo {author} {\bibfnamefont {A.}~\bibnamefont
  {Fine}},\ }\bibfield  {title} {\bibinfo {title} {Hidden variables, joint
  probability, and the bell inequalities},\ }\href
  {https://doi.org/10.1103/PhysRevLett.48.291} {\bibfield  {journal} {\bibinfo
  {journal} {Phys. Rev. Lett.}\ }\textbf {\bibinfo {volume} {48}},\ \bibinfo
  {pages} {291} (\bibinfo {year} {1982})}\BibitemShut {NoStop}%
\bibitem [{\citenamefont {Abramsky}\ and\ \citenamefont
  {Hardy}(2012)}]{AH2012Logical}%
  \BibitemOpen
  \bibfield  {author} {\bibinfo {author} {\bibfnamefont {S.}~\bibnamefont
  {Abramsky}}\ and\ \bibinfo {author} {\bibfnamefont {L.}~\bibnamefont
  {Hardy}},\ }\bibfield  {title} {\bibinfo {title} {Logical bell
  inequalities},\ }\href {https://doi.org/10.1103/PhysRevA.85.062114}
  {\bibfield  {journal} {\bibinfo  {journal} {Phys. Rev. A}\ }\textbf {\bibinfo
  {volume} {85}},\ \bibinfo {pages} {062114} (\bibinfo {year}
  {2012})}\BibitemShut {NoStop}%
\bibitem [{Note1()}]{Note1}%
  \BibitemOpen
  \bibinfo {note} {Note that there is a typo in the statement of Theorem 4 in
  \cite {ABM2017Contextual}: the denominator on the right-hand side should be
  $n$, not $k$, which is equal to $m$ in our case.}\BibitemShut {Stop}%
\bibitem [{\citenamefont {G{\"u}hne}\ \emph {et~al.}(2021)\citenamefont
  {G{\"u}hne}, \citenamefont {Haapasalo}, \citenamefont {Kraft}, \citenamefont
  {Pellonp{\"a}{\"a}},\ and\ \citenamefont {Uola}}]{GHK+2021Incompatible}%
  \BibitemOpen
  \bibfield  {author} {\bibinfo {author} {\bibfnamefont {O.}~\bibnamefont
  {G{\"u}hne}}, \bibinfo {author} {\bibfnamefont {E.}~\bibnamefont
  {Haapasalo}}, \bibinfo {author} {\bibfnamefont {T.}~\bibnamefont {Kraft}},
  \bibinfo {author} {\bibfnamefont {J.-P.}\ \bibnamefont {Pellonp{\"a}{\"a}}},\
  and\ \bibinfo {author} {\bibfnamefont {R.}~\bibnamefont {Uola}},\ }\bibfield
  {title} {\bibinfo {title} {Incompatible measurements in quantum information
  science}} (\bibinfo {year} {2021}),\ \bibinfo {note}
  {arXiv:2112.06784}\BibitemShut {NoStop}%
\bibitem [{\citenamefont {Ghirardi}\ \emph {et~al.}(1980)\citenamefont
  {Ghirardi}, \citenamefont {Rimini},\ and\ \citenamefont
  {Weber}}]{GRW1980General}%
  \BibitemOpen
  \bibfield  {author} {\bibinfo {author} {\bibfnamefont {G.~C.}\ \bibnamefont
  {Ghirardi}}, \bibinfo {author} {\bibfnamefont {A.}~\bibnamefont {Rimini}},\
  and\ \bibinfo {author} {\bibfnamefont {T.}~\bibnamefont {Weber}},\ }\bibfield
   {title} {\bibinfo {title} {A general argument against superluminal
  transmission through the quantum mechanical measurement process},\ }\href
  {https://doi.org/10.1007/BF02817189} {\bibfield  {journal} {\bibinfo
  {journal} {Lettere al Nuovo Cimento (1971-1985)}\ }\textbf {\bibinfo {volume}
  {27}},\ \bibinfo {pages} {293} (\bibinfo {year} {1980})}\BibitemShut
  {NoStop}%
\bibitem [{\citenamefont {Miller}\ and\ \citenamefont
  {Shi}(2017)}]{MS2017Universal}%
  \BibitemOpen
  \bibfield  {author} {\bibinfo {author} {\bibfnamefont {C.~A.}\ \bibnamefont
  {Miller}}\ and\ \bibinfo {author} {\bibfnamefont {Y.}~\bibnamefont {Shi}},\
  }\bibfield  {title} {\bibinfo {title} {Universal security for randomness
  expansion from the spot-checking protocol},\ }\href
  {https://doi.org/10.1137/15M1044333} {\bibfield  {journal} {\bibinfo
  {journal} {SIAM Journal on Computing}\ }\textbf {\bibinfo {volume} {46}},\
  \bibinfo {pages} {1304} (\bibinfo {year} {2017})},\ \Eprint
  {https://arxiv.org/abs/https://doi.org/10.1137/15M1044333}
  {https://doi.org/10.1137/15M1044333} \BibitemShut {NoStop}%
\bibitem [{\citenamefont {Silman}\ \emph {et~al.}(2013)\citenamefont {Silman},
  \citenamefont {Pironio},\ and\ \citenamefont {Massar}}]{SPM2013Device}%
  \BibitemOpen
  \bibfield  {author} {\bibinfo {author} {\bibfnamefont {J.}~\bibnamefont
  {Silman}}, \bibinfo {author} {\bibfnamefont {S.}~\bibnamefont {Pironio}},\
  and\ \bibinfo {author} {\bibfnamefont {S.}~\bibnamefont {Massar}},\
  }\bibfield  {title} {\bibinfo {title} {Device-independent randomness
  generation in the presence of weak cross-talk},\ }\href
  {https://doi.org/10.1103/PhysRevLett.110.100504} {\bibfield  {journal}
  {\bibinfo  {journal} {Phys. Rev. Lett.}\ }\textbf {\bibinfo {volume} {110}},\
  \bibinfo {pages} {100504} (\bibinfo {year} {2013})}\BibitemShut {NoStop}%
\bibitem [{\citenamefont {Um}\ \emph {et~al.}(2020)\citenamefont {Um},
  \citenamefont {Zhao}, \citenamefont {Zhang}, \citenamefont {Wang},
  \citenamefont {Wang}, \citenamefont {Qiao}, \citenamefont {Zhou},
  \citenamefont {Ma},\ and\ \citenamefont {Kim}}]{UZZ+2020Randomness}%
  \BibitemOpen
  \bibfield  {author} {\bibinfo {author} {\bibfnamefont {M.}~\bibnamefont
  {Um}}, \bibinfo {author} {\bibfnamefont {Q.}~\bibnamefont {Zhao}}, \bibinfo
  {author} {\bibfnamefont {J.}~\bibnamefont {Zhang}}, \bibinfo {author}
  {\bibfnamefont {P.}~\bibnamefont {Wang}}, \bibinfo {author} {\bibfnamefont
  {Y.}~\bibnamefont {Wang}}, \bibinfo {author} {\bibfnamefont {M.}~\bibnamefont
  {Qiao}}, \bibinfo {author} {\bibfnamefont {H.}~\bibnamefont {Zhou}}, \bibinfo
  {author} {\bibfnamefont {X.}~\bibnamefont {Ma}},\ and\ \bibinfo {author}
  {\bibfnamefont {K.}~\bibnamefont {Kim}},\ }\bibfield  {title} {\bibinfo
  {title} {Randomness expansion secured by quantum contextuality},\ }\href
  {https://doi.org/10.1103/PhysRevApplied.13.034077} {\bibfield  {journal}
  {\bibinfo  {journal} {Phys. Rev. Applied}\ }\textbf {\bibinfo {volume}
  {13}},\ \bibinfo {pages} {034077} (\bibinfo {year} {2020})}\BibitemShut
  {NoStop}%
\bibitem [{\citenamefont {Vallée}\ \emph {et~al.}(2023)\citenamefont
  {Vallée}, \citenamefont {Emeriau}, \citenamefont {Bourdoncle}, \citenamefont
  {Sohbi}, \citenamefont {Mansfield},\ and\ \citenamefont
  {Markham}}]{EBMM2022Corrected}%
  \BibitemOpen
  \bibfield  {author} {\bibinfo {author} {\bibfnamefont {K.}~\bibnamefont
  {Vallée}}, \bibinfo {author} {\bibfnamefont {P.-E.}\ \bibnamefont
  {Emeriau}}, \bibinfo {author} {\bibfnamefont {B.}~\bibnamefont {Bourdoncle}},
  \bibinfo {author} {\bibfnamefont {A.}~\bibnamefont {Sohbi}}, \bibinfo
  {author} {\bibfnamefont {S.}~\bibnamefont {Mansfield}},\ and\ \bibinfo
  {author} {\bibfnamefont {D.}~\bibnamefont {Markham}},\ }\href@noop {}
  {\bibinfo {title} {Corrected {B}ell and noncontextuality inequalities for
  realistic experiments}} (\bibinfo {year} {2023}),\ \Eprint
  {https://arxiv.org/abs/2310.19383} {arXiv:2310.19383 [quant-ph]} \BibitemShut
  {NoStop}%
\bibitem [{\citenamefont {Navascu{\'{e}}s}\ \emph {et~al.}(2008)\citenamefont
  {Navascu{\'{e}}s}, \citenamefont {Pironio},\ and\ \citenamefont
  {Ac{\'{\i}}n}}]{NPA2008Convergent}%
  \BibitemOpen
  \bibfield  {author} {\bibinfo {author} {\bibfnamefont {M.}~\bibnamefont
  {Navascu{\'{e}}s}}, \bibinfo {author} {\bibfnamefont {S.}~\bibnamefont
  {Pironio}},\ and\ \bibinfo {author} {\bibfnamefont {A.}~\bibnamefont
  {Ac{\'{\i}}n}},\ }\bibfield  {title} {\bibinfo {title} {A convergent
  hierarchy of semidefinite programs characterizing the set of quantum
  correlations},\ }\href {https://doi.org/10.1088/1367-2630/10/7/073013}
  {\bibfield  {journal} {\bibinfo  {journal} {New Journal of Physics}\ }\textbf
  {\bibinfo {volume} {10}},\ \bibinfo {pages} {073013} (\bibinfo {year}
  {2008})}\BibitemShut {NoStop}%
\bibitem [{\citenamefont {Pironio}\ and\ \citenamefont
  {Massar}(2013)}]{PM2013Security}%
  \BibitemOpen
  \bibfield  {author} {\bibinfo {author} {\bibfnamefont {S.}~\bibnamefont
  {Pironio}}\ and\ \bibinfo {author} {\bibfnamefont {S.}~\bibnamefont
  {Massar}},\ }\bibfield  {title} {\bibinfo {title} {Security of practical
  private randomness generation},\ }\href
  {https://doi.org/10.1103/PhysRevA.87.012336} {\bibfield  {journal} {\bibinfo
  {journal} {Physical Review A}\ }\textbf {\bibinfo {volume} {87}},\ \bibinfo
  {pages} {012336} (\bibinfo {year} {2013})}\BibitemShut {NoStop}%
\bibitem [{\citenamefont {Fehr}\ \emph {et~al.}(2013)\citenamefont {Fehr},
  \citenamefont {Gelles},\ and\ \citenamefont {Schaffner}}]{FGS2013Security}%
  \BibitemOpen
  \bibfield  {author} {\bibinfo {author} {\bibfnamefont {S.}~\bibnamefont
  {Fehr}}, \bibinfo {author} {\bibfnamefont {R.}~\bibnamefont {Gelles}},\ and\
  \bibinfo {author} {\bibfnamefont {C.}~\bibnamefont {Schaffner}},\ }\bibfield
  {title} {\bibinfo {title} {Security and composability of randomness expansion
  from bell inequalities},\ }\href {https://doi.org/10.1103/PhysRevA.87.012335}
  {\bibfield  {journal} {\bibinfo  {journal} {Phys. Rev. A}\ }\textbf {\bibinfo
  {volume} {87}},\ \bibinfo {pages} {012335} (\bibinfo {year}
  {2013})}\BibitemShut {NoStop}%
\bibitem [{\citenamefont {Nieto-Silleras}\ \emph {et~al.}(2018)\citenamefont
  {Nieto-Silleras}, \citenamefont {Bamps}, \citenamefont {Silman},\ and\
  \citenamefont {Pironio}}]{NBSP2018Device}%
  \BibitemOpen
  \bibfield  {author} {\bibinfo {author} {\bibfnamefont {O.}~\bibnamefont
  {Nieto-Silleras}}, \bibinfo {author} {\bibfnamefont {C.}~\bibnamefont
  {Bamps}}, \bibinfo {author} {\bibfnamefont {J.}~\bibnamefont {Silman}},\ and\
  \bibinfo {author} {\bibfnamefont {S.}~\bibnamefont {Pironio}},\ }\bibfield
  {title} {\bibinfo {title} {Device-independent randomness generation from
  several {B}ell estimators},\ }\href
  {https://doi.org/10.1088/1367-2630/aaaa06} {\bibfield  {journal} {\bibinfo
  {journal} {New Journal of Physics}\ }\textbf {\bibinfo {volume} {20}},\
  \bibinfo {pages} {023049} (\bibinfo {year} {2018})}\BibitemShut {NoStop}%
\bibitem [{\citenamefont {Renner}(2005)}]{Renner2005Security}%
  \BibitemOpen
  \bibfield  {author} {\bibinfo {author} {\bibfnamefont {R.}~\bibnamefont
  {Renner}},\ }\emph {\bibinfo {title} {Security of quantum key distrubtion}},\
  \href@noop {} {Ph.D. thesis},\ \bibinfo  {school} {ETH Zurich} (\bibinfo
  {year} {2005})\BibitemShut {NoStop}%
\bibitem [{\citenamefont {Somaschi}\ \emph {et~al.}(2016)\citenamefont
  {Somaschi}, \citenamefont {Giesz}, \citenamefont {De~Santis}, \citenamefont
  {Loredo}, \citenamefont {Almeida}, \citenamefont {Hornecker}, \citenamefont
  {Portalupi}, \citenamefont {Grange}, \citenamefont {Antón}, \citenamefont
  {Demory}, \citenamefont {Gómez}, \citenamefont {Sagnes}, \citenamefont
  {Lanzillotti-Kimura}, \citenamefont {Lemaítre}, \citenamefont {Auffeves},
  \citenamefont {White}, \citenamefont {Lanco},\ and\ \citenamefont
  {Senellart}}]{Somaschi2016}%
  \BibitemOpen
  \bibfield  {author} {\bibinfo {author} {\bibfnamefont {N.}~\bibnamefont
  {Somaschi}}, \bibinfo {author} {\bibfnamefont {V.}~\bibnamefont {Giesz}},
  \bibinfo {author} {\bibfnamefont {L.}~\bibnamefont {De~Santis}}, \bibinfo
  {author} {\bibfnamefont {J.~C.}\ \bibnamefont {Loredo}}, \bibinfo {author}
  {\bibfnamefont {M.~P.}\ \bibnamefont {Almeida}}, \bibinfo {author}
  {\bibfnamefont {G.}~\bibnamefont {Hornecker}}, \bibinfo {author}
  {\bibfnamefont {S.~L.}\ \bibnamefont {Portalupi}}, \bibinfo {author}
  {\bibfnamefont {T.}~\bibnamefont {Grange}}, \bibinfo {author} {\bibfnamefont
  {C.}~\bibnamefont {Antón}}, \bibinfo {author} {\bibfnamefont
  {J.}~\bibnamefont {Demory}}, \bibinfo {author} {\bibfnamefont
  {C.}~\bibnamefont {Gómez}}, \bibinfo {author} {\bibfnamefont
  {I.}~\bibnamefont {Sagnes}}, \bibinfo {author} {\bibfnamefont {N.~D.}\
  \bibnamefont {Lanzillotti-Kimura}}, \bibinfo {author} {\bibfnamefont
  {A.}~\bibnamefont {Lemaítre}}, \bibinfo {author} {\bibfnamefont
  {A.}~\bibnamefont {Auffeves}}, \bibinfo {author} {\bibfnamefont {A.~G.}\
  \bibnamefont {White}}, \bibinfo {author} {\bibfnamefont {L.}~\bibnamefont
  {Lanco}},\ and\ \bibinfo {author} {\bibfnamefont {P.}~\bibnamefont
  {Senellart}},\ }\bibfield  {title} {\bibinfo {title} {Near-optimal
  single-photon sources in the solid state},\ }\href
  {https://doi.org/10.1038/nphoton.2016.23} {\bibfield  {journal} {\bibinfo
  {journal} {Nature Photonics}\ }\textbf {\bibinfo {volume} {10}},\ \bibinfo
  {pages} {340} (\bibinfo {year} {2016})}\BibitemShut {NoStop}%
\bibitem [{\citenamefont {Loudon}(2000)}]{Loudon2000}%
  \BibitemOpen
  \bibfield  {author} {\bibinfo {author} {\bibfnamefont {R.}~\bibnamefont
  {Loudon}},\ }\href {https://books.google.fr/books?id=AEkfajgqldoC} {\emph
  {\bibinfo {title} {{The Quantum Theory of Light}}}}\ (\bibinfo  {publisher}
  {OUP Oxford},\ \bibinfo {year} {2000})\ pp.\ \bibinfo {pages} {107--117,
  227--231}\BibitemShut {NoStop}%
\bibitem [{\citenamefont {Ollivier}\ \emph {et~al.}(2021)\citenamefont
  {Ollivier}, \citenamefont {Thomas}, \citenamefont {Wein}, \citenamefont {{de
  Buy Wenniger}}, \citenamefont {Coste}, \citenamefont {Loredo}, \citenamefont
  {Somaschi}, \citenamefont {Harouri}, \citenamefont {Lemaitre}, \citenamefont
  {Sagnes}, \citenamefont {Lanco}, \citenamefont {Simon}, \citenamefont
  {Anton}, \citenamefont {Krebs},\ and\ \citenamefont
  {Senellart}}]{Ollivier2021}%
  \BibitemOpen
  \bibfield  {author} {\bibinfo {author} {\bibfnamefont {H.}~\bibnamefont
  {Ollivier}}, \bibinfo {author} {\bibfnamefont {S.~E.}\ \bibnamefont
  {Thomas}}, \bibinfo {author} {\bibfnamefont {S.~C.}\ \bibnamefont {Wein}},
  \bibinfo {author} {\bibfnamefont {I.~M.}\ \bibnamefont {{de Buy Wenniger}}},
  \bibinfo {author} {\bibfnamefont {N.}~\bibnamefont {Coste}}, \bibinfo
  {author} {\bibfnamefont {J.~C.}\ \bibnamefont {Loredo}}, \bibinfo {author}
  {\bibfnamefont {N.}~\bibnamefont {Somaschi}}, \bibinfo {author}
  {\bibfnamefont {A.}~\bibnamefont {Harouri}}, \bibinfo {author} {\bibfnamefont
  {A.}~\bibnamefont {Lemaitre}}, \bibinfo {author} {\bibfnamefont
  {I.}~\bibnamefont {Sagnes}}, \bibinfo {author} {\bibfnamefont
  {L.}~\bibnamefont {Lanco}}, \bibinfo {author} {\bibfnamefont
  {C.}~\bibnamefont {Simon}}, \bibinfo {author} {\bibfnamefont
  {C.}~\bibnamefont {Anton}}, \bibinfo {author} {\bibfnamefont
  {O.}~\bibnamefont {Krebs}},\ and\ \bibinfo {author} {\bibfnamefont
  {P.}~\bibnamefont {Senellart}},\ }\bibfield  {title} {\bibinfo {title}
  {{Hong-Ou-Mandel Interference with Imperfect Single Photon Sources}},\ }\href
  {https://doi.org/10.1103/PhysRevLett.126.063602} {\bibfield  {journal}
  {\bibinfo  {journal} {Physical Review Letters}\ }\textbf {\bibinfo {volume}
  {126}},\ \bibinfo {pages} {63602} (\bibinfo {year} {2021})}\BibitemShut
  {NoStop}%
\bibitem [{\citenamefont {Heurtel}\ \emph {et~al.}(2023)\citenamefont
  {Heurtel}, \citenamefont {Fyrillas}, \citenamefont {Gliniasty}, \citenamefont
  {Le~Bihan}, \citenamefont {Malherbe}, \citenamefont {Pailhas}, \citenamefont
  {Bertasi}, \citenamefont {Bourdoncle}, \citenamefont {Emeriau}, \citenamefont
  {Mezher}, \citenamefont {Music}, \citenamefont {Belabas}, \citenamefont
  {Valiron}, \citenamefont {Senellart}, \citenamefont {Mansfield},\ and\
  \citenamefont {Senellart}}]{Heurtel2022}%
  \BibitemOpen
  \bibfield  {author} {\bibinfo {author} {\bibfnamefont {N.}~\bibnamefont
  {Heurtel}}, \bibinfo {author} {\bibfnamefont {A.}~\bibnamefont {Fyrillas}},
  \bibinfo {author} {\bibfnamefont {G.~d.}\ \bibnamefont {Gliniasty}}, \bibinfo
  {author} {\bibfnamefont {R.}~\bibnamefont {Le~Bihan}}, \bibinfo {author}
  {\bibfnamefont {S.}~\bibnamefont {Malherbe}}, \bibinfo {author}
  {\bibfnamefont {M.}~\bibnamefont {Pailhas}}, \bibinfo {author} {\bibfnamefont
  {E.}~\bibnamefont {Bertasi}}, \bibinfo {author} {\bibfnamefont
  {B.}~\bibnamefont {Bourdoncle}}, \bibinfo {author} {\bibfnamefont {P.-E.}\
  \bibnamefont {Emeriau}}, \bibinfo {author} {\bibfnamefont {R.}~\bibnamefont
  {Mezher}}, \bibinfo {author} {\bibfnamefont {L.}~\bibnamefont {Music}},
  \bibinfo {author} {\bibfnamefont {N.}~\bibnamefont {Belabas}}, \bibinfo
  {author} {\bibfnamefont {B.}~\bibnamefont {Valiron}}, \bibinfo {author}
  {\bibfnamefont {P.}~\bibnamefont {Senellart}}, \bibinfo {author}
  {\bibfnamefont {S.}~\bibnamefont {Mansfield}},\ and\ \bibinfo {author}
  {\bibfnamefont {J.}~\bibnamefont {Senellart}},\ }\bibfield  {title} {\bibinfo
  {title} {Perceval: {A} {S}oftware {P}latform for {D}iscrete {V}ariable
  {P}hotonic {Q}uantum {C}omputing},\ }\href
  {https://doi.org/10.22331/q-2023-02-21-931} {\bibfield  {journal} {\bibinfo
  {journal} {{Quantum}}\ }\textbf {\bibinfo {volume} {7}},\ \bibinfo {pages}
  {931} (\bibinfo {year} {2023})}\BibitemShut {NoStop}%
\bibitem [{\citenamefont {Kuhlmann}\ \emph {et~al.}(2013)\citenamefont
  {Kuhlmann}, \citenamefont {Houel}, \citenamefont {Ludwig}, \citenamefont
  {Greuter}, \citenamefont {Reuter}, \citenamefont {Wieck}, \citenamefont
  {Poggio},\ and\ \citenamefont {Warburton}}]{Kuhlmann2013}%
  \BibitemOpen
  \bibfield  {author} {\bibinfo {author} {\bibfnamefont {A.~V.}\ \bibnamefont
  {Kuhlmann}}, \bibinfo {author} {\bibfnamefont {J.}~\bibnamefont {Houel}},
  \bibinfo {author} {\bibfnamefont {A.}~\bibnamefont {Ludwig}}, \bibinfo
  {author} {\bibfnamefont {L.}~\bibnamefont {Greuter}}, \bibinfo {author}
  {\bibfnamefont {D.}~\bibnamefont {Reuter}}, \bibinfo {author} {\bibfnamefont
  {A.~D.}\ \bibnamefont {Wieck}}, \bibinfo {author} {\bibfnamefont
  {M.}~\bibnamefont {Poggio}},\ and\ \bibinfo {author} {\bibfnamefont {R.~J.}\
  \bibnamefont {Warburton}},\ }\bibfield  {title} {\bibinfo {title} {Charge
  noise and spin noise in a semiconductor quantum device},\ }\href
  {https://doi.org/10.1038/nphys2688} {\bibfield  {journal} {\bibinfo
  {journal} {Nature Physics}\ }\textbf {\bibinfo {volume} {9}},\ \bibinfo
  {pages} {570} (\bibinfo {year} {2013})}\BibitemShut {NoStop}%
\bibitem [{\citenamefont {Krawczyk}(1995)}]{Krawczyk1995}%
  \BibitemOpen
  \bibfield  {author} {\bibinfo {author} {\bibfnamefont {H.}~\bibnamefont
  {Krawczyk}},\ }\bibfield  {title} {\bibinfo {title} {New hash functions for
  message authentication},\ }in\ \href@noop {} {\emph {\bibinfo {booktitle}
  {Advances in Cryptology --- EUROCRYPT '95}}},\ \bibinfo {editor} {edited by\
  \bibinfo {editor} {\bibfnamefont {L.~C.}\ \bibnamefont {Guillou}}\ and\
  \bibinfo {editor} {\bibfnamefont {J.-J.}\ \bibnamefont {Quisquater}}}\
  (\bibinfo  {publisher} {Springer Berlin Heidelberg},\ \bibinfo {address}
  {Berlin, Heidelberg},\ \bibinfo {year} {1995})\ pp.\ \bibinfo {pages}
  {301--310}\BibitemShut {NoStop}%
\bibitem [{\citenamefont {Knuth}\ and\ \citenamefont
  {Yao}(1976)}]{KY76Algorithms}%
  \BibitemOpen
  \bibfield  {author} {\bibinfo {author} {\bibfnamefont {D.}~\bibnamefont
  {Knuth}}\ and\ \bibinfo {author} {\bibfnamefont {A.}~\bibnamefont {Yao}},\
  }\bibfield  {title} {\bibinfo {title} {The complexity of nonuniform random
  number generation},\ }in\ \href@noop {} {\emph {\bibinfo {booktitle}
  {Algorithms and Complexity: New Directions and Recent Results}}}\ (\bibinfo
  {publisher} {Academic Press},\ \bibinfo {year} {1976})\ pp.\ \bibinfo {pages}
  {357--428}\BibitemShut {NoStop}%
\bibitem [{\citenamefont {Hao}\ and\ \citenamefont
  {Hoshi}(1997)}]{HH1997Interval}%
  \BibitemOpen
  \bibfield  {author} {\bibinfo {author} {\bibfnamefont {T.~S.}\ \bibnamefont
  {Hao}}\ and\ \bibinfo {author} {\bibfnamefont {M.}~\bibnamefont {Hoshi}},\
  }\bibfield  {title} {\bibinfo {title} {Interval algorithm for random number
  generation},\ }\href@noop {} {\bibfield  {journal} {\bibinfo  {journal} {IEEE
  Transactions on Information Theory}\ }\textbf {\bibinfo {volume} {43}},\
  \bibinfo {pages} {599} (\bibinfo {year} {1997})}\BibitemShut {NoStop}%
\bibitem [{\citenamefont {Bai}\ \emph {et~al.}(2021)\citenamefont {Bai},
  \citenamefont {Huang}, \citenamefont {Qiao}, \citenamefont {Nie},
  \citenamefont {Tang}, \citenamefont {Chu}, \citenamefont {Zhang},\ and\
  \citenamefont {Pan}}]{BHQGbps2021}%
  \BibitemOpen
  \bibfield  {author} {\bibinfo {author} {\bibfnamefont {B.}~\bibnamefont
  {Bai}}, \bibinfo {author} {\bibfnamefont {J.}~\bibnamefont {Huang}}, \bibinfo
  {author} {\bibfnamefont {G.-R.}\ \bibnamefont {Qiao}}, \bibinfo {author}
  {\bibfnamefont {Y.-Q.}\ \bibnamefont {Nie}}, \bibinfo {author} {\bibfnamefont
  {W.}~\bibnamefont {Tang}}, \bibinfo {author} {\bibfnamefont {T.}~\bibnamefont
  {Chu}}, \bibinfo {author} {\bibfnamefont {J.}~\bibnamefont {Zhang}},\ and\
  \bibinfo {author} {\bibfnamefont {J.-W.}\ \bibnamefont {Pan}},\ }\bibfield
  {title} {\bibinfo {title} {18.8 {G}bps real-time quantum random number
  generator with a photonic integrated chip},\ }\href@noop {} {\bibfield
  {journal} {\bibinfo  {journal} {Applied Physics Letter}\ }\textbf {\bibinfo
  {volume} {118}},\ \bibinfo {pages} {264001} (\bibinfo {year}
  {2021})}\BibitemShut {NoStop}%
\bibitem [{\citenamefont {Wang}\ \emph {et~al.}(2018)\citenamefont {Wang},
  \citenamefont {Paesani}, \citenamefont {Ding}, \citenamefont {Santagati},
  \citenamefont {Skrzypczyk}, \citenamefont {Salavrakos}, \citenamefont {Tura},
  \citenamefont {Augusiak}, \citenamefont {Man{\v c}inska}, \citenamefont
  {Bacco}, \citenamefont {Bonneau}, \citenamefont {Silverstone}, \citenamefont
  {Gong}, \citenamefont {Ac{\'\i}n}, \citenamefont {Rottwitt}, \citenamefont
  {Oxenl{\o}we}, \citenamefont {O'Brien}, \citenamefont {Laing},\ and\
  \citenamefont {Thompson}}]{WPD2018Multidimensional}%
  \BibitemOpen
  \bibfield  {author} {\bibinfo {author} {\bibfnamefont {J.}~\bibnamefont
  {Wang}}, \bibinfo {author} {\bibfnamefont {S.}~\bibnamefont {Paesani}},
  \bibinfo {author} {\bibfnamefont {Y.}~\bibnamefont {Ding}}, \bibinfo {author}
  {\bibfnamefont {R.}~\bibnamefont {Santagati}}, \bibinfo {author}
  {\bibfnamefont {P.}~\bibnamefont {Skrzypczyk}}, \bibinfo {author}
  {\bibfnamefont {A.}~\bibnamefont {Salavrakos}}, \bibinfo {author}
  {\bibfnamefont {J.}~\bibnamefont {Tura}}, \bibinfo {author} {\bibfnamefont
  {R.}~\bibnamefont {Augusiak}}, \bibinfo {author} {\bibfnamefont
  {L.}~\bibnamefont {Man{\v c}inska}}, \bibinfo {author} {\bibfnamefont
  {D.}~\bibnamefont {Bacco}}, \bibinfo {author} {\bibfnamefont
  {D.}~\bibnamefont {Bonneau}}, \bibinfo {author} {\bibfnamefont {J.~W.}\
  \bibnamefont {Silverstone}}, \bibinfo {author} {\bibfnamefont
  {Q.}~\bibnamefont {Gong}}, \bibinfo {author} {\bibfnamefont {A.}~\bibnamefont
  {Ac{\'\i}n}}, \bibinfo {author} {\bibfnamefont {K.}~\bibnamefont {Rottwitt}},
  \bibinfo {author} {\bibfnamefont {L.~K.}\ \bibnamefont {Oxenl{\o}we}},
  \bibinfo {author} {\bibfnamefont {J.~L.}\ \bibnamefont {O'Brien}}, \bibinfo
  {author} {\bibfnamefont {A.}~\bibnamefont {Laing}},\ and\ \bibinfo {author}
  {\bibfnamefont {M.~G.}\ \bibnamefont {Thompson}},\ }\bibfield  {title}
  {\bibinfo {title} {Multidimensional quantum entanglement with large-scale
  integrated optics},\ }\href {https://doi.org/10.1126/science.aar7053}
  {\bibfield  {journal} {\bibinfo  {journal} {Science}\ }\textbf {\bibinfo
  {volume} {360}},\ \bibinfo {pages} {285} (\bibinfo {year} {2018})},\ \Eprint
  {https://arxiv.org/abs/https://www.science.org/doi/pdf/10.1126/science.aar7053}
  {https://www.science.org/doi/pdf/10.1126/science.aar7053} \BibitemShut
  {NoStop}%
\bibitem [{\citenamefont {Dupuis}\ \emph {et~al.}(2020)\citenamefont {Dupuis},
  \citenamefont {Fawzi},\ and\ \citenamefont {Renner}}]{DFR2020Entropy}%
  \BibitemOpen
  \bibfield  {author} {\bibinfo {author} {\bibfnamefont {F.}~\bibnamefont
  {Dupuis}}, \bibinfo {author} {\bibfnamefont {O.}~\bibnamefont {Fawzi}},\ and\
  \bibinfo {author} {\bibfnamefont {R.}~\bibnamefont {Renner}},\ }\bibfield
  {title} {\bibinfo {title} {Entropy accumulation},\ }\href
  {https://doi.org/10.1007/s00220-020-03839-5} {\bibfield  {journal} {\bibinfo
  {journal} {Communications in Mathematical Physics}\ }\textbf {\bibinfo
  {volume} {379}},\ \bibinfo {pages} {867} (\bibinfo {year}
  {2020})}\BibitemShut {NoStop}%
\bibitem [{\citenamefont {Dupuis}\ and\ \citenamefont
  {Fawzi}(2019)}]{DF2019Entropy}%
  \BibitemOpen
  \bibfield  {author} {\bibinfo {author} {\bibfnamefont {F.}~\bibnamefont
  {Dupuis}}\ and\ \bibinfo {author} {\bibfnamefont {O.}~\bibnamefont {Fawzi}},\
  }\bibfield  {title} {\bibinfo {title} {Entropy accumulation with improved
  second-order term},\ }\href {https://doi.org/10.1109/TIT.2019.2929564}
  {\bibfield  {journal} {\bibinfo  {journal} {IEEE Transactions on Information
  Theory}\ }\textbf {\bibinfo {volume} {65}},\ \bibinfo {pages} {7596}
  (\bibinfo {year} {2019})}\BibitemShut {NoStop}%
\bibitem [{\citenamefont {Arnon-Friedman}\ \emph {et~al.}(2019)\citenamefont
  {Arnon-Friedman}, \citenamefont {Renner},\ and\ \citenamefont
  {Vidick}}]{ARV2019Simple}%
  \BibitemOpen
  \bibfield  {author} {\bibinfo {author} {\bibfnamefont {R.}~\bibnamefont
  {Arnon-Friedman}}, \bibinfo {author} {\bibfnamefont {R.}~\bibnamefont
  {Renner}},\ and\ \bibinfo {author} {\bibfnamefont {T.}~\bibnamefont
  {Vidick}},\ }\bibfield  {title} {\bibinfo {title} {Simple and tight
  device-independent security proofs},\ }\href
  {https://doi.org/10.1137/18M1174726} {\bibfield  {journal} {\bibinfo
  {journal} {SIAM Journal on Computing}\ }\textbf {\bibinfo {volume} {48}},\
  \bibinfo {pages} {181} (\bibinfo {year} {2019})},\ \Eprint
  {https://arxiv.org/abs/https://doi.org/10.1137/18M1174726}
  {https://doi.org/10.1137/18M1174726} \BibitemShut {NoStop}%
\bibitem [{\citenamefont {Brown}\ \emph {et~al.}(2020)\citenamefont {Brown},
  \citenamefont {Ragy},\ and\ \citenamefont {Colbeck}}]{BRC2020Framework}%
  \BibitemOpen
  \bibfield  {author} {\bibinfo {author} {\bibfnamefont {P.~J.}\ \bibnamefont
  {Brown}}, \bibinfo {author} {\bibfnamefont {S.}~\bibnamefont {Ragy}},\ and\
  \bibinfo {author} {\bibfnamefont {R.}~\bibnamefont {Colbeck}},\ }\bibfield
  {title} {\bibinfo {title} {A framework for quantum-secure device-independent
  randomness expansion},\ }\href {https://doi.org/10.1109/TIT.2019.2960252}
  {\bibfield  {journal} {\bibinfo  {journal} {IEEE Transactions on Information
  Theory}\ }\textbf {\bibinfo {volume} {66}},\ \bibinfo {pages} {2964}
  (\bibinfo {year} {2020})}\BibitemShut {NoStop}%
\bibitem [{\citenamefont {Brown}\ \emph {et~al.}(2021)\citenamefont {Brown},
  \citenamefont {Fawzi},\ and\ \citenamefont {Fawzi}}]{BFF2021Device}%
  \BibitemOpen
  \bibfield  {author} {\bibinfo {author} {\bibfnamefont {P.}~\bibnamefont
  {Brown}}, \bibinfo {author} {\bibfnamefont {H.}~\bibnamefont {Fawzi}},\ and\
  \bibinfo {author} {\bibfnamefont {O.}~\bibnamefont {Fawzi}},\ }\bibfield
  {title} {\bibinfo {title} {Device-independent lower bounds on the conditional
  von {N}eumann entropy}} (\bibinfo {year} {2021}),\ \bibinfo {note}
  {arXiv:2106.13692}\BibitemShut {NoStop}%
\bibitem [{\citenamefont {Janner}\ \emph {et~al.}(2009)\citenamefont {Janner},
  \citenamefont {Tulli}, \citenamefont {García-Granda}, \citenamefont
  {Belmonte},\ and\ \citenamefont {Pruneri}}]{Janner2009}%
  \BibitemOpen
  \bibfield  {author} {\bibinfo {author} {\bibfnamefont {D.}~\bibnamefont
  {Janner}}, \bibinfo {author} {\bibfnamefont {D.}~\bibnamefont {Tulli}},
  \bibinfo {author} {\bibfnamefont {M.}~\bibnamefont {García-Granda}},
  \bibinfo {author} {\bibfnamefont {M.}~\bibnamefont {Belmonte}},\ and\
  \bibinfo {author} {\bibfnamefont {V.}~\bibnamefont {Pruneri}},\ }\bibfield
  {title} {\bibinfo {title} {Micro-structured integrated electro-optic {LiNbO3}
  modulators},\ }\href {https://doi.org/10.1002/lpor.200810073} {\bibfield
  {journal} {\bibinfo  {journal} {Laser \& Photonics Reviews}\ }\textbf
  {\bibinfo {volume} {3}},\ \bibinfo {pages} {301} (\bibinfo {year}
  {2009})}\BibitemShut {NoStop}%
\bibitem [{\citenamefont {Quack}\ \emph {et~al.}(2020)\citenamefont {Quack},
  \citenamefont {Sattari}, \citenamefont {Takabayashi}, \citenamefont {Zhang},
  \citenamefont {Verheyen}, \citenamefont {Bogaerts}, \citenamefont {Edinger},
  \citenamefont {Errando-Herranz},\ and\ \citenamefont {Gylfason}}]{Quack2020}%
  \BibitemOpen
  \bibfield  {author} {\bibinfo {author} {\bibfnamefont {N.}~\bibnamefont
  {Quack}}, \bibinfo {author} {\bibfnamefont {H.}~\bibnamefont {Sattari}},
  \bibinfo {author} {\bibfnamefont {A.~Y.}\ \bibnamefont {Takabayashi}},
  \bibinfo {author} {\bibfnamefont {Y.}~\bibnamefont {Zhang}}, \bibinfo
  {author} {\bibfnamefont {P.}~\bibnamefont {Verheyen}}, \bibinfo {author}
  {\bibfnamefont {W.}~\bibnamefont {Bogaerts}}, \bibinfo {author}
  {\bibfnamefont {P.}~\bibnamefont {Edinger}}, \bibinfo {author} {\bibfnamefont
  {C.}~\bibnamefont {Errando-Herranz}},\ and\ \bibinfo {author} {\bibfnamefont
  {K.~B.}\ \bibnamefont {Gylfason}},\ }\bibfield  {title} {\bibinfo {title}
  {{MEMS}-{Enabled} {Silicon} {Photonic} {Integrated} {Devices} and
  {Circuits}},\ }\href {https://doi.org/10.1109/JQE.2019.2946841} {\bibfield
  {journal} {\bibinfo  {journal} {IEEE Journal of Quantum Electronics}\
  }\textbf {\bibinfo {volume} {56}},\ \bibinfo {pages} {1} (\bibinfo {year}
  {2020})},\ \bibinfo {note} {conference Name: IEEE Journal of Quantum
  Electronics}\BibitemShut {NoStop}%
\bibitem [{\citenamefont {Barth}\ \emph {et~al.}(2016)\citenamefont {Barth},
  \citenamefont {L{\"{u}}ker}, \citenamefont {Vagov}, \citenamefont {Reiter},
  \citenamefont {Kuhn},\ and\ \citenamefont {Axt}}]{Barth2016}%
  \BibitemOpen
  \bibfield  {author} {\bibinfo {author} {\bibfnamefont {A.~M.}\ \bibnamefont
  {Barth}}, \bibinfo {author} {\bibfnamefont {S.}~\bibnamefont {L{\"{u}}ker}},
  \bibinfo {author} {\bibfnamefont {A.}~\bibnamefont {Vagov}}, \bibinfo
  {author} {\bibfnamefont {D.~E.}\ \bibnamefont {Reiter}}, \bibinfo {author}
  {\bibfnamefont {T.}~\bibnamefont {Kuhn}},\ and\ \bibinfo {author}
  {\bibfnamefont {V.~M.}\ \bibnamefont {Axt}},\ }\bibfield  {title} {\bibinfo
  {title} {{Fast and selective phonon-assisted state preparation of a quantum
  dot by adiabatic undressing}},\ }\href
  {https://doi.org/10.1103/PhysRevB.94.045306} {\bibfield  {journal} {\bibinfo
  {journal} {Phys. Rev. B}\ }\textbf {\bibinfo {volume} {94}},\ \bibinfo
  {pages} {45306} (\bibinfo {year} {2016})}\BibitemShut {NoStop}%
\bibitem [{\citenamefont {Cosacchi}\ \emph {et~al.}(2019)\citenamefont
  {Cosacchi}, \citenamefont {Ungar}, \citenamefont {Cygorek}, \citenamefont
  {Vagov},\ and\ \citenamefont {Axt}}]{Cosacchi2019}%
  \BibitemOpen
  \bibfield  {author} {\bibinfo {author} {\bibfnamefont {M.}~\bibnamefont
  {Cosacchi}}, \bibinfo {author} {\bibfnamefont {F.}~\bibnamefont {Ungar}},
  \bibinfo {author} {\bibfnamefont {M.}~\bibnamefont {Cygorek}}, \bibinfo
  {author} {\bibfnamefont {A.}~\bibnamefont {Vagov}},\ and\ \bibinfo {author}
  {\bibfnamefont {V.~M.}\ \bibnamefont {Axt}},\ }\bibfield  {title} {\bibinfo
  {title} {{Emission-Frequency Separated High Quality Single-Photon Sources
  Enabled by Phonons}},\ }\href
  {https://doi.org/10.1103/PhysRevLett.123.017403} {\bibfield  {journal}
  {\bibinfo  {journal} {Phys. Rev. Lett.}\ }\textbf {\bibinfo {volume} {123}},\
  \bibinfo {pages} {17403} (\bibinfo {year} {2019})}\BibitemShut {NoStop}%
\bibitem [{\citenamefont {Gustin}\ and\ \citenamefont
  {Hughes}(2020)}]{Gustin2020}%
  \BibitemOpen
  \bibfield  {author} {\bibinfo {author} {\bibfnamefont {C.}~\bibnamefont
  {Gustin}}\ and\ \bibinfo {author} {\bibfnamefont {S.}~\bibnamefont
  {Hughes}},\ }\bibfield  {title} {\bibinfo {title} {{Efficient
  Pulse-Excitation Techniques for Single Photon Sources from Quantum Dots in
  Optical Cavities}},\ }\href
  {https://doi.org/https://doi.org/10.1002/qute.201900073} {\bibfield
  {journal} {\bibinfo  {journal} {Advanced Quantum Technologies}\ }\textbf
  {\bibinfo {volume} {3}},\ \bibinfo {pages} {1900073} (\bibinfo {year}
  {2020})}\BibitemShut {NoStop}%
\bibitem [{\citenamefont {Thomas}\ \emph {et~al.}(2021)\citenamefont {Thomas},
  \citenamefont {Billard}, \citenamefont {Coste}, \citenamefont {Wein},
  \citenamefont {Priya}, \citenamefont {Ollivier}, \citenamefont {Krebs},
  \citenamefont {Taza{\"{i}}rt}, \citenamefont {Harouri}, \citenamefont
  {Lemaitre}, \citenamefont {Sagnes}, \citenamefont {Anton}, \citenamefont
  {Lanco}, \citenamefont {Somaschi}, \citenamefont {Loredo},\ and\
  \citenamefont {Senellart}}]{Thomas2021}%
  \BibitemOpen
  \bibfield  {author} {\bibinfo {author} {\bibfnamefont {S.~E.}\ \bibnamefont
  {Thomas}}, \bibinfo {author} {\bibfnamefont {M.}~\bibnamefont {Billard}},
  \bibinfo {author} {\bibfnamefont {N.}~\bibnamefont {Coste}}, \bibinfo
  {author} {\bibfnamefont {S.~C.}\ \bibnamefont {Wein}}, \bibinfo {author}
  {\bibnamefont {Priya}}, \bibinfo {author} {\bibfnamefont {H.}~\bibnamefont
  {Ollivier}}, \bibinfo {author} {\bibfnamefont {O.}~\bibnamefont {Krebs}},
  \bibinfo {author} {\bibfnamefont {L.}~\bibnamefont {Taza{\"{i}}rt}}, \bibinfo
  {author} {\bibfnamefont {A.}~\bibnamefont {Harouri}}, \bibinfo {author}
  {\bibfnamefont {A.}~\bibnamefont {Lemaitre}}, \bibinfo {author}
  {\bibfnamefont {I.}~\bibnamefont {Sagnes}}, \bibinfo {author} {\bibfnamefont
  {C.}~\bibnamefont {Anton}}, \bibinfo {author} {\bibfnamefont
  {L.}~\bibnamefont {Lanco}}, \bibinfo {author} {\bibfnamefont
  {N.}~\bibnamefont {Somaschi}}, \bibinfo {author} {\bibfnamefont {J.~C.}\
  \bibnamefont {Loredo}},\ and\ \bibinfo {author} {\bibfnamefont
  {P.}~\bibnamefont {Senellart}},\ }\bibfield  {title} {\bibinfo {title}
  {{Bright Polarized Single-Photon Source Based on a Linear Dipole}},\ }\href
  {https://doi.org/10.1103/PhysRevLett.126.233601} {\bibfield  {journal}
  {\bibinfo  {journal} {Physical Review Letters}\ }\textbf {\bibinfo {volume}
  {126}},\ \bibinfo {pages} {233601} (\bibinfo {year} {2021})}\BibitemShut
  {NoStop}%
\bibitem [{\citenamefont {Elitzur}\ \emph {et~al.}(1992)\citenamefont
  {Elitzur}, \citenamefont {Popescu},\ and\ \citenamefont
  {Rohrlich}}]{EPR1992Quantum}%
  \BibitemOpen
  \bibfield  {author} {\bibinfo {author} {\bibfnamefont {A.~C.}\ \bibnamefont
  {Elitzur}}, \bibinfo {author} {\bibfnamefont {S.}~\bibnamefont {Popescu}},\
  and\ \bibinfo {author} {\bibfnamefont {D.}~\bibnamefont {Rohrlich}},\
  }\bibfield  {title} {\bibinfo {title} {Quantum nonlocality for each pair in
  an ensemble},\ }\href
  {https://doi.org/https://doi.org/10.1016/0375-9601(92)90952-I} {\bibfield
  {journal} {\bibinfo  {journal} {Physics Letters A}\ }\textbf {\bibinfo
  {volume} {162}},\ \bibinfo {pages} {25} (\bibinfo {year} {1992})}\BibitemShut
  {NoStop}%
\bibitem [{\citenamefont {Barbosa}\ \emph {et~al.}(2022)\citenamefont
  {Barbosa}, \citenamefont {Douce}, \citenamefont {Emeriau}, \citenamefont
  {Kashefi},\ and\ \citenamefont {Mansfield}}]{barbosa2022continuous}%
  \BibitemOpen
  \bibfield  {author} {\bibinfo {author} {\bibfnamefont {R.~S.}\ \bibnamefont
  {Barbosa}}, \bibinfo {author} {\bibfnamefont {T.}~\bibnamefont {Douce}},
  \bibinfo {author} {\bibfnamefont {P.-E.}\ \bibnamefont {Emeriau}}, \bibinfo
  {author} {\bibfnamefont {E.}~\bibnamefont {Kashefi}},\ and\ \bibinfo {author}
  {\bibfnamefont {S.}~\bibnamefont {Mansfield}},\ }\bibfield  {title} {\bibinfo
  {title} {Continuous-variable nonlocality and contextuality},\ }\href@noop {}
  {\bibfield  {journal} {\bibinfo  {journal} {Communications in Mathematical
  Physics}\ }\textbf {\bibinfo {volume} {391}},\ \bibinfo {pages} {1047}
  (\bibinfo {year} {2022})}\BibitemShut {NoStop}%
\bibitem [{\citenamefont {Bancal}\ \emph {et~al.}(2014)\citenamefont {Bancal},
  \citenamefont {Sheridan},\ and\ \citenamefont {Scarani}}]{BSS2014More}%
  \BibitemOpen
  \bibfield  {author} {\bibinfo {author} {\bibfnamefont {J.-D.}\ \bibnamefont
  {Bancal}}, \bibinfo {author} {\bibfnamefont {L.}~\bibnamefont {Sheridan}},\
  and\ \bibinfo {author} {\bibfnamefont {V.}~\bibnamefont {Scarani}},\
  }\bibfield  {title} {\bibinfo {title} {More randomness from the same data},\
  }\href {https://doi.org/10.1088/1367-2630/16/3/033011} {\bibfield  {journal}
  {\bibinfo  {journal} {New Journal of Physics}\ }\textbf {\bibinfo {volume}
  {16}},\ \bibinfo {pages} {033011} (\bibinfo {year} {2014})}\BibitemShut
  {NoStop}%
\bibitem [{\citenamefont {Nieto-Silleras}\ \emph {et~al.}(2014)\citenamefont
  {Nieto-Silleras}, , \citenamefont {Pironio},\ and\ \citenamefont
  {Silman}}]{NPS2014Using}%
  \BibitemOpen
  \bibfield  {author} {\bibinfo {author} {\bibfnamefont {O.}~\bibnamefont
  {Nieto-Silleras}}, , \bibinfo {author} {\bibfnamefont {S.}~\bibnamefont
  {Pironio}},\ and\ \bibinfo {author} {\bibfnamefont {J.}~\bibnamefont
  {Silman}},\ }\bibfield  {title} {\bibinfo {title} {Using complete measurement
  statistics for optimal device-independent randomness evaluation},\ }\href
  {https://doi.org/10.1088/1367-2630/16/1/013035} {\bibfield  {journal}
  {\bibinfo  {journal} {New Journal of Physics}\ }\textbf {\bibinfo {volume}
  {16}},\ \bibinfo {pages} {013035} (\bibinfo {year} {2014})}\BibitemShut
  {NoStop}%
\bibitem [{\citenamefont {Ac\'{\i}n}\ \emph {et~al.}(2012)\citenamefont
  {Ac\'{\i}n}, \citenamefont {Massar},\ and\ \citenamefont
  {Pironio}}]{AMP2012Randomness}%
  \BibitemOpen
  \bibfield  {author} {\bibinfo {author} {\bibfnamefont {A.}~\bibnamefont
  {Ac\'{\i}n}}, \bibinfo {author} {\bibfnamefont {S.}~\bibnamefont {Massar}},\
  and\ \bibinfo {author} {\bibfnamefont {S.}~\bibnamefont {Pironio}},\
  }\bibfield  {title} {\bibinfo {title} {Randomness versus nonlocality and
  entanglement},\ }\href {https://doi.org/10.1103/PhysRevLett.108.100402}
  {\bibfield  {journal} {\bibinfo  {journal} {Physical Review Letters}\
  }\textbf {\bibinfo {volume} {108}},\ \bibinfo {pages} {100402} (\bibinfo
  {year} {2012})}\BibitemShut {NoStop}%
\bibitem [{Note2()}]{Note2}%
  \BibitemOpen
  \bibinfo {note} {It was realised after the data acquisition that the detector
  efficiency should in fact not be accounted for, because correcting for that
  efficiency introduces biases in the measured empirical table that reduce the
  CHSH violation. In practice, the detector efficiencies differ only at the
  third significant digit, so the impact on the obtained results is negligible:
  correcting by these biases diminishes the maximum obtainable CHSH violation
  by only $10^{-4}$ (simulation using the Perceval package in Python). In
  addition, in principle these biases do not increase the signaling fraction.
  Therefore, we can still exploit the results from our experiment.}\BibitemShut
  {Stop}%
\end{thebibliography}%

\section*{Acknowledgements}

We thank Pascale Senellart for valuable discussions and feedback, Mathias Pont for experimental help and multi-photon model of the Perceval package used in our simulations, Damian Markham and Ivan Šupić for discussions on quantum correlations and signalling, and Swabian Instruments for helpful discussions.

This work has received funding by the European Union’s Horizon 2020 Research and Innovation Programme QUDOT-TECH under the Marie Skłodowska Curie Grant Agreement No. 861097 and from BPI France Concours Innovation PIA3 projects DOS0148634/00 and DOS0148633/00 – Reconfigurable Optical Quantum Computing.

\section*{Author contributions}
A.F.: experimental investigation, data analysis, methodology, visualisation and writing.
B.B.: conceptualisation, formal analysis, methodology, visualisation and writing.
A.M.: experimental investigation.
P-E.E.: conceptualisation and  formal analysis.
K.S.: randomness extraction. 
N.Margaria: QD device processing. 
M.M.: QD sample growth. 
A.L.: QD sample growth. 
I.S.: QD device processing. 
P.S.: QD device characterization. 
T.H.A.: QD device processing. 
S.B.: QD device processing.
N.S.: supervision.
N.Maring: experimental investigation and supervision.
N.B.: 
experimental investigation,
data analysis, methodology, formal analysis, visualisation and writing, supervision. 
S.M.: conceptualisation, formal analysis, visualisation and writing, supervision.

\section*{Competing interests}
A patent has been filed listing B.B., P.-E.E. and S.M. as inventors.

\section*{Data availability}

The data that support the findings of this study are available from the 
corresponding authors upon request. 

\section*{Code availability}

Code for the randomness extraction step is available from the corresponding authors upon request.

\section*{Methods}

\subsection*{Single-photon generation}

We illustrate the optical setup in Fig.\ \ref{fig:setup}. Single photons at $925.16$ nm are generated by a Quandela single photon source relying on an InAs/GaAs quantum dot embedded in a cavity \cite{Somaschi2016}. A voltage of the order of -1.5 V is applied on the dot, such that the emission line is in resonance with the cavity and to reduce charge noise. The source is pumped using the longitudinal-acoustic phonon assisted excitation scheme \cite{Barth2016, Cosacchi2019, Gustin2020, Thomas2021} at around 924.24 nm and with a pump spectral FWHM $\Delta\omega\approx\SI{0.1}{\nm}$, corresponding to a pulse duration of the order of $\SI{12}{ps}$. The pump is a mode-locked femtosecond laser with a repetition rate of 79.08 MHz, corresponding to a duration $\tau_{\text{rep}}\approx\SI{12,6}{ns}$ between two consecutive pulses.

The laser pulses are subsequently temporally shaped using a filtering setup based on a 4f line principle, which includes a grating splitting incoming wavelengths into different directions, a slit for wavelength selection, and another grating recombining the light into a single Gaussian beam, ensuring optimal pumping of the source. To increase the random number generation rate of the experiment, the pulse rate is doubled using a fibered Mach-Zehnder interferometer with an approximately $\tau_{\text{rep}}/2$ delay line on one arm. 

The excitation pulses are then sent to the photon source. Emitted single photons and reflected pump light are sent to a filtering stage consisting in three bandpass filters ($10^{-3}$ transmission at 924 nm, 805 pm FWHM) and a Fabry-Perot etalon (FSR 204 pm and finesse 14 at 925 nm, $\SI{59(1)}{\%}$ (error bar is one standard deviation) single-photon transmission).

The first lens brightness of our single-photon source (number of photons collected per excitation pulse at the level of the source \cite{Somaschi2016}) amounts to $\SI{39(3)}{\%}$ and the polarized fibered brightness (number of photons collected per excitation pulse after the filtering stage including the polarizer, see Fig.\ \ref{fig:setup}) is $\SI{8.3(8)}{\%}$, corresponding to a polarized photon output rate of $\SI{13.0(1)e6}{\s^{-1}}$. 

The purity and indistinguishability of the generated photons are increased by inserting a polarizer and a Fabry-Perot etalon. The purity with the etalon (see Sup III.A) is $g^{(2)}(0)\approx \SI{2.31(3)}{\%}$. The Hong-Ou-Mandel (HOM) visibility is $\SI{93.09(4)}{\%}$. We can deduce from these measurements the photon mean wave-packet overlap $M_s=\SI{97.65(6)}{\%}$ \cite{Ollivier2021}.

\subsection*{Single-photon manipulation}

A passive demultiplexing stage ($\SI{20}{\%}$ insertion loss) converts the photon stream into pairs of photons arriving simultaneously at the QRNG chip input. The demultiplexer consists in a polarizer set in the diagonal position, such that the subsequent polarizing beamsplitter on Fig. \ref{fig:setup} acts like a symmetric beamsplitter. The waveplates preceding the polarizer are set to maximise the number of transmitted photons through the polarizer. One of the outputs of the beamsplitter leads to a fibered delay loop. A pair of photons is successfully demultiplexed when the first photon of the pair takes the long path via the fibered loop and the second photon the short path. Hence, only 1/4 of photon pairs are successfully demultiplexed.

The silica glass QRNG chip features laser-written waveguides and four configurable thermo-optic phase shifters (see Sup. III.A for details and operation). The optical transmission of the chip is $\SI{58(1)}{\%}$ (averaged over the two inputs used). Its output is sent to a superconducting nanowire single photon detector (SNSPD, $\SI{70}{\%}$ detection efficiency). Photon times of arrival are processed by a time tagger module. We measure an overall setup transmission, ie;\ the probability of a photon being detected after an excitation pulse arriving on the source,  of $\SI{2.7}{\%}$ by using the photon countrate on the detectors.

\newpage

\subsection*{Selection of the measurement bases}

We compute the on-chip phases that should be applied in order to maximally violate the CHSH inequality and explain how to calibrate the voltages accordingly. Alice and Bob each control an interferometer (see Fig.\ \ref{subfig:chip_scheme}) consisting in a phase $\psi$, followed by a symmetric beamsplitter, a phase $\phi$ and a second symmetric beamsplitter. Their interferometer is described by the unitary matrix
\begin{multline}
\widehat{U}(\psi, \phi)= 
\frac{1}{\sqrt{2}}
\left[\begin{array}{ll}
1 &
i \\
i &
1
\end{array}\right]
\left[\begin{array}{ll}
e^{i\phi} &
0 \\
0 &
1
\end{array}\right]
\\ \\
\times
\frac{1}{\sqrt{2}}
\left[\begin{array}{ll}
1 &
i \\
i &
1
\end{array}\right]
\left[\begin{array}{ll}
e^{i\psi} &
0 \\
0 &
1
\end{array}\right]
\\ \\
= i e^{i\frac{\phi}{2}}
\left[\begin{array}{ll}
\sin\left(\frac{\phi}{2}\right)e^{i\psi} &
\cos\left(\frac{\phi}{2}\right) \\
\cos\left(\frac{\phi}{2}\right)e^{i\psi} &
-\sin\left(\frac{\phi}{2}\right)
\end{array}\right].
\label{eq:unitary}
\end{multline}
A possible choice of measurement bases that maximally violates the CHSH inequality is  described in Table~\ref{table:context_table}. With this choice, when Alice's measurement basis is $x=0$, we can measure the HOM visibility of the photons by recording the photon coincidences on her outputs. Indeed, for that basis, her interferometer behaves like a symmetric beamsplitter.

\begin{table}
\begin{tabular}{|l|l|l|l|l|l|}
\hline
Measurement context $(x,y)$ & $\psi^A$ & $\phi^A$ & $\psi^B$ & $\phi^B$ \\ \hline
$(0,0)$                      & 0            & $-\pi/2$      & 0            & $-\pi/4$    \\ \hline
$(0,1)$                      & 0            & $-\pi/2$      & 0            & $\pi/4$    \\ \hline
$(1,0)$                      & 0            & $0$        & 0            & $-\pi/4$    \\ \hline
$(1,1)$                      & 0            & $0$        & 0            & $\pi/4$    \\ \hline
\end{tabular}
\caption{Phase shifts for maximal CHSH inequality violation. Alice's measurement bases are labeled $x=0$ and $x=1$, Bob's bases are labeled $y=0$ and $y=1$. A measurement context is a pair $(x,y)$. From the measurement bases, we compute the phases, hence the voltages, that should be applied to Alice's and Bob's phase shifters.}
\label{table:context_table}
\end{table}

\subsection*{CHSH inequality value computation}

A behaviour for a CHSH Bell test can be described by a table containing four columns, one for each measurement context $(x,y)$, and the observed probability of each coincidence result $(a,b)$ on the corresponding row. For $V_\text{HOM}=\SI{93}{\%}$, which characterizes our single photon source, the expected behaviour is (see Sup.\ III.C):
\begin{table}[H]
	\centering
	\begin{tabular}{cc||cccc}
		 a & b & $00$ & $01$ & $10$ & $11$ \\ \hline
		 $0$ & $0$ & 0.414    & 0.086    & 0.073    & 0.073  \\
		 $0$ & $1$ & 0.086    & 0.414    & 0.427    & 0.427 \\
		 $1$ & $0$ & 0.086    & 0.414    & 0.427    & 0.427 \\
		 $1$ & $1$ & 0.414    & 0.086    & 0.073    & 0.073
	\end{tabular}
\caption*{}
\label{tab:behaviour_1}
\end{table}

From the $\SI{3.2e5}{}$ test rounds of our 94.5-hour-long randomness generation experiment, we construct the following observed behaviour:
\begin{table}[H]
	\centering
	\begin{tabular}{cc||cccc}
		 a & b & $00$ & $01$ & $10$ & $11$ \\ \hline
		 $0$ & $0$ & 0.424    & 0.090    & 0.085    & 0.077  \\
		 $0$ & $1$ & 0.087    & 0.416    & 0.418    & 0.429 \\
		 $1$ & $0$ & 0.084    & 0.410    & 0.418    & 0.423 \\
		 $1$ & $1$ & 0.405    & 0.084    & 0.079    & 0.071
	\end{tabular}
\caption*{}
\label{tab:behaviour_2}
\end{table}

The uncertainty due to finite statistics is of the order of 0.001 for each cell. Discrepancies between the two tables can be attributed to multi-photon emissions combined with optical losses, the reflectivities of the on-chip directional couplers estimated at 51 \%, errors on phase shifter phases and dark counts.

For the theoretical behaviour, we get as expected $S_\text{CHSH}\approx 0.84$ ($I_{\text{CHSH}}\approx 2.73$) and from our experimentally observed behaviour $S_\text{CHSH}\approx 0.835$ ($I_{\text{CHSH}}\approx 2.685$).

\section*{Appendix A: Contextual fraction and signalling fraction}
\label{app:fractions}

For a contextual $n$-partite game, we denote by $p$ an associated empirical behaviour, i.e. a set of probability distributions on the outputs conditioned on the inputs. We call a given choice of inputs a context, denoted by $C$. In an ideal contextual game, the measurements in a context are compatible, so the marginal distributions computed from two different distributions on the intersection of their contexts are equal, i.e. they obey the generalised no-signalling condition. In a real implementation however, the measurements of a context might not be perfectly compatible because of physical crosstalk. We nonetheless define the same contexts as in the ideal description of the game, which means that we can observe behaviours that do no satisfy the no-signalling conditions. We call $\mathcal{NS}$ the space of behaviours that does satisfy no-signalling and $\mathcal{E}$ the bigger space of all behaviours, i.e. the space of sets of real numbers for each contexts that satisfy the usual normalisation and positiveness conditions. 

A no-signalling behaviour is non-contextual if the distribution for each context $\{p_C\}$ can be obtained as the marginal of a single distribution on global assignments~\cite{AB2011Sheaf}. The contextual fraction $\CF$ quantifies how contextual a behaviour is and is the solution of a linear program (LP)~\cite{ABM2017Contextual}. It generalises the nonlocal fraction~\cite{EPR1992Quantum}. For all behaviours $p$ it holds that $\CF(p) \in [0,1]$, that $p$ is non-contextual if and only if $\CF(p) =0$, and that $p$ is maximally (or strongly) contextual if and only if $\CF(p)=1$. 

Similarly, the signalling fraction $\SF$ that we introduce in this work (simultaneously introduced in~\cite{EBMM2022Corrected}) quantifies how signalling a behaviour is. For an empirical behaviour $p$, we consider affine decompositions of the form $p = s \cdot p' + (1-s)\cdot p''$, for which $p'$ is a non-signalling behaviour. If $s^*$ is the maximum weight that can be assigned to a non-signalling $p'$ in such a decomposition then we define the non-signalling fraction as $\mathsf{NSF} (p):=s^*$. Then we define the signalling fraction $\SF(p):=(1-s*)$, as this is the irreducibly signalling `fraction' of $p$. In this way it quantifies how far from $\NS$ a behaviour is. We can thus decompose $p$ as:

\begin{equation} 
p = \mathsf{NSF}(p) \, p^{NS} + \SF(p) \, p^{SS} ,
\label{def:SF}
\end{equation} 
where $p^{NS}$ is a no-signalling behaviour and $p^{SS}$ is a strongly signalling behaviour, i.e. $\SF(p^{SS})=1$. For all behaviours $p$, it holds that $\SF(p) \in [0,1]$, and that $p$ is non-signalling if and only if $\SF(p)=0$. The signalling and non-signalling fractions $\SF(p)$ can be computed by a linear program, similarly to $\CF(p)$.

One can study the potential decompositions of behaviours into hidden-variable models (HVMs). An HVM for a measurement scenario is defined by a set of hidden variables $\Lambda$, a distribution $h(\lambda)$ over $\Lambda$ and, for each $\lambda$, a behaviour $h^{\lambda}$ on the same measurement scenario. A behaviour $e$ is said to be realised by an HVM $\{\Lambda, h(\lambda), h^{\lambda} \} $ if it arises as the weighted average
\begin{equation}
p=\sum_{\lambda \in \Lambda} h(\lambda) \cdot h^{\lambda}.
\label{eq:HVM}
\end{equation}
In this way HVMs provide a framework to reason about deeper or more fine-grained descriptions of hypothetical underlying processes that may be giving rise to an observed empirical behaviour, and that in principle an adversary may seek to exploit. Note that for finite measurement scenarios signalling, non-signalling, and non-contextual behaviours each arise as convex combinations of a finite number of vertex behaviours. These provide canonical hidden variable spaces, and when we wish to consider these properties it thus suffices to consider averages over finite hidden variable spaces (for extensions to the continuum see \cite{barbosa2022continuous}).

In the device-independent approach, it is thus crucial to examine all acceptable HVMs, i.e.\ all HVMs that are compatible with the underlying theory and the description of the experimental setup we implement, to bound the power of the eavesdropper. In particular, in the case of certified randomness, an HVM decomposition describes the eavesdropper's ability to predict the outputs of an experiment. In the case of a multi-partite game implemented in a space-like separated manner, all the behaviours of the HVMs must for instance be no-signalling (the term parameter-independent' is sometimes used instead, when talking about hidden-variable properties~\cite{EBMM2022Corrected}). Under these conditions, non-contextuality is equivalent to perfect predictability, because any non-contextual behaviour can be decomposed into no-signalling deterministic behaviours~\cite{Fine1982Hidden}.

The constraint that the elements of the HVM are no-signalling can be replaced by an upper-bound on their signalling fraction, if one has reasons to believe that the physical setup underlying the examined behaviour allows some flow of information between the physical components implementing the measurements of each context, but only in a limited amount. Indeed, crosstalk between the components, which happens at the physical level, is related to signalling at the level of the behaviours. We assume here that the amount of signalling allowed at the hidden-variable level $\sigma$ is not greater than the signalling $\SF_{\ell}$ observed in the estimated behaviour, which we define in the following way: we look for decompositions of the form $p = s \cdot p' + (1-s)\cdot p''$, for which $p'$ is contained in $\mathsf{NPA}_{\ell} \subseteq \mathcal{NS}$, the $\ell$-th level of the NPA hierarchy~\cite{NPA2007Bounding, NPA2008Convergent}, which approximates the set of behaviours achievable by performing quantum measurements on quantum states. The set $\mathsf{NPA}_{\ell}$ is a strict subset of the no-signalling space. We project into $\mathsf{NPA}_{\ell}$ rather than into $\mathcal{NS}$ because a no-signalling but supra-quantum behaviour can also be achieved by a signalling and quantum behaviour.
Similarly to Eq.\ \ref{def:SF}, we can thus decompose $p$ as
\begin{equation}
p = (1-\SF_{\ell}(p)) \, p^{\mathsf{NPA}_{\ell}} + \SF_{\ell}(p) \, p'' \, .
\label{def:sigma}
\end{equation} 
$\SF_{\ell}$ satisfies the following properties: $\mathsf{SF}_{\ell}$ increases with $\ell$, by definition $\sigma \geq \mathsf{SF}$ and $\sigma \geq \mathsf{SF}_{\ell}$ for all $\ell$ if we assume that quantum mechanics is valid.

Note that the security of the protocol described in the next section relies only on Propositions \ref{prop:SclBound} and \ref{prop:SclSxi}, which do not rely on any assumption on the signalling $\sigma$. It can thus be straightforwardly adapted to another choice of $\sigma$, based on a semi-device-dependent characterisation of the devices or related to different cryptographic assumptions, by replacing $\SF_{\ell}$ in Step 6 by the appropriate choice. 

Another approach to derive a lower bound on the physical crosstalk from the observed behaviour was proposed by Silman et al.~\cite{SPM2013Device}. It puts a constraint at the level of the measurements rather than on the behaviours (see Eq.~(6) in ~\cite{SPM2013Device}). This problem being computationally intractable, the authors also propose to use the NPA hierarchy to obtain a lower bound. The crosstalk measure $\chi$ is smaller than $\SF_{\ell}$ because the optimisation space for $\SF_{\ell}$  is a subspace of the optimisation space for the former. Indeed, the second constraint of Eq.~(6) in~\cite{SPM2013Device} is equivalent to 
\begin{equation}
- \chi \leq p(ab|xy) - p^{\mathsf{NPA}_{\ell}}(ab|xy) \leq \chi,
\label{eq:chiSilman}
\end{equation}
with $p^{\mathsf{NPA}_{\ell}}$ in $\mathsf{NPA}_{\ell}$, where, without loss of generality, we took the same state for both behaviours, while Eq.~\eqref{def:sigma} is equivalent to 
\begin{equation}
    p-p^{\mathsf{NPA}_{\ell}}=\SF_{\ell}(p)(p''-p^{\mathsf{NPA}_{\ell}}),
\end{equation}
which implies Eq.~\ref{eq:chiSilman} for $\chi=\SF_{\ell}(p)$, as $-1\leq p''-p^{\mathsf{NPA}_{\ell}} \leq 1$.

Both metrics provide valid lower bounds on the crosstalk, as long as they are used coherently (i.e.\ one should compute the maximal score in the presence of crosstalk using $\chi$ if one uses $\chi$ to estimate crosstalk). The $\chi$ metrics are based on a statistical distance, while the $\SF$ metrics look for a convex decomposition. The advantages of our chosen approach are that it integrates well with convex optimisation techniques used to compute maximal scores, Bell inequality violations and guessing probabilities, and it allows us to leverage all duality properties that relate these quantities to Bell inequalities~\cite{BSS2014More, NPS2014Using, ABM2017Contextual}.

\section*{Appendix B: Relation between score with bounded contextuality and score with bounded signalling}
\label{app:scores} 

We reformulate Proposition~\ref{prop:SclSxi} in the more general case of a $n$-partite contextual game.

\begin{prop}
For a contextual $n$-partite game with binary input choices for all players, let $C$ be a distinguished context. Then 
\begin{equation}
S^{\sigma}_{C}  \leq S_{\mathrm{cl}}^{\xi = \sigma}. 
\end{equation}
\label{prop:scores}
\end{prop}

\begin{proof}
Let $p^*$ be a solution for the optimisation problem defining $S^{\sigma}_{C}$. We can decompose $p^*$ as 
\begin{equation} 
p^* = (1- \tau) \cdot p' + \tau \cdot p''
\label{eq:decompNS}
\end{equation}
with $\tau \in [0,\sigma]$ and $p'$ a no-signalling behaviour. Moreover, by definition, we must have $p'_C=p''_C=1$. In the case of binary inputs, using the same construction as in Appendix D of~\cite{MS2017Universal}, we can find a local HVM for any no-signalling behaviour with binary inputs that is deterministic on a context, which implies that $\CF(p')=0$. Eq.~\eqref{eq:decompNS} is then a feasible point for the LP defining the contextual fraction of $p^*$, which in turn implies
\begin{equation}
    \CF(p^*) \leq \tau \leq \sigma,
\end{equation}
and $p^*$ is thus a feasible point for the optimisation problem corresponding to $S_{\mathrm{cl}}^{\xi = \sigma}$.
\end{proof}

Combining Propositions \ref{prop:SclBound} and \ref{prop:scores}, we see that, for randomness to be certified, the observed score $c/qN$ (see Fig. \ref{fig:protocol}) must satisfy $\frac{c}{qN} > \frac{\sigma(2^n-k)+k}{2^n}$, which in turn implies $\sigma < \frac{2^n (c/qN)-k}{2^n-k}$. That gives the upper bound on the amount of crosstalk that our analysis can tolerate while  certifying randomness.

\section*{Appendix C: Protocols for randomness generation and expansion against classical and quantum side information}
\label{app:protocol} 

\begin{figure}
\setlength{\fboxrule}{1pt}
\fbox{\parbox{\linewidth}{
\raggedright
\textit{Arguments:}
\begin{enumerate}
\item[$G$]: An $k$-consistent $n$-party contextual game with binary inputs for each player and a distinguished context $C$
\item[$N$]: The output length (a positive integer)
\item[$q$]: The test probability (a real number in $[0, 1]$)
\item[$\chi$]: The score threshold (a real number in $[ 0, 1]$)
\item[$\ell$]: The level of the NPA hierarchy (a positive integer)
\end{enumerate} 
\textit{Variables:}
\begin{enumerate}
\item[$c$]: The number of wins in test rounds (a positive integer that we set to 0)
\item[$\hat{p}$]: the estimated behaviour (a table indexed by the input and output choices that we set to 0)
\end{enumerate}
\textbf{Protocol:}
\begin{enumerate}
\item Choose a bit $t \in \{ 0, 1 \}$ according
to Bernoulli distribution $(1-q,q )$.  
\item If $t = 1$ (``test round''), play
$G$, record input and output in $\hat{p}$, add score to $c$.
\item If $t = 0$ (``generation round''), input $C$ and record output.
\item Steps 2--4 are repeated $(N-1)$ more times.
\item Normalise $\hat{p}$ and compute $\mathsf{SF}_{\ell}(\hat{p})$.
\item If $c/ (q N)  - \mathsf{SF}_{\ell}(\hat{p}) (2^n-k)/2^n   <  \chi $, 
then the protocol aborts. Otherwise, it succeeds.
\end{enumerate}
}}
\caption{\textbf{Protocol for randomness generation protocol with nonzero crosstalk}. The protocol is based on an $n$-player $k$-consistent nonlocal game with binary inputs, and it is secure provided that the signalling at the HV level $\sigma$ is not greater than $\mathsf{SF}_{\ell}$. Compared to previous protocols that ignored the effect of crosstalk, the observed average score $c/qN$ has to be greater than a fixed threshold plus the correction due to signalling for the protocol to succeed.}
\label{fig:protocol}
\end{figure}

To certify randomness against quantum side information, we use the protocol described in Fig.~\ref{fig:protocol}. We call $\mathbf{C}$ and $\mathbf{S}$ the sequence of context choices and outputs produced when the protocol is implemented, and $E$ the side information accessible to an eavesdropper. Then, when the protocol succeeds, the $\delta$-smooth min-entropy of the outputs $\mathbf{S}$ is at least:
\begin{equation}
H^{\delta}_{\textrm{min}}(\mathbf{S}|\mathbf{C},E)  \geq N (\pi(\chi)-\Delta),
\label{eq:Bound}
\end{equation}
with
\begin{equation}
\label{eq:RateCurve}
\pi(\chi) =2\frac{\log(e)(\chi-S_{\mathrm{cl}})^2}{nd-1}
\end{equation}
and
\begin{align}
\label{eq:Delta}
\Delta & = \frac{\log(2/\delta^2)}{N\varepsilon}+2ndq + \frac{\varepsilon}{q} \frac{8\log(e)(\chi-S_{\mathrm{cl}})^2}{(nd-1)^2} +  \\
& \quad + \left(\frac{\varepsilon}{q}\right)^2 \frac{32 \log(e)(\chi-S_{\mathrm{cl}})^3}{3(nd-1)^3}  \cdot 2^{\frac{\varepsilon}{q} \frac{4\log(e)(\chi-S_{\mathrm{cl}})}{nd-1}}  \nonumber
\end{align}
for any $\varepsilon \in ]0,1]$, where $\log(e)$ is the base-2 logarithm of the exponential and $d$ is the total number of outputs.

The bound~\eqref{eq:Bound} is derived from the security proof introduced in~\cite{MS2017Universal}. It is valid because the compatibility assumption in~\cite{MS2017Universal}, which translates into the no-signalling assumption at the behaviour's level, is required only to bound the distinguished score for the game. We derived a new relaxed distinguished score with signalling, $S_{C}^{\sigma}$, and can then use the tools of~\cite{MS2017Universal} to lower-bound the min-entropy with that new score. This score is, according to Propositions \ref{prop:SclBound} and \ref{prop:scores}, the classical score in the absence of signalling $S_{\mathrm{cl}}$ increased by $\sigma (2^n-k)/2^n$. An alternative way to present our protocol and the associated bound would be to test instead $c/ (q N)<\chi$ in Step. 6 and to replace $S_{\mathrm{cl}}$ by $S_{C}^{\sigma}$ in Eq.~\eqref{eq:RateCurve} and~\eqref{eq:Delta}. The effect of signalling would then be visible in the expression of the bound rather than, as it is now, in the description of the protocol itself through the requirement on $\chi$. These two approaches are equivalent because what matters in the min-entropy bound if the difference $\chi-S_{\mathrm{cl}}$.

The expressions for $\pi$ and $\Delta$ given by Eqs.~\eqref{eq:RateCurve} and~\eqref{eq:Delta} are obtained in the following way. The rate curve $\pi(\chi)$ is described by Theorem 5.8 of~\cite{MS2017Universal}, where, for the CHSH game, $r=4$ (the size of the output alphabet, $n\cdot d$ in our notation) and $W_{G,\bar{a}}=3/4$ (the score with distinguished input, which, for binary inputs, is equal to the classical score, $S_{\mathrm{cl}}$ in our notation, as proved in~\cite{MS2017Universal}, Appendix D). The term $\log(2/\delta^2)/N\epsilon$ in Eq.~\eqref{eq:Delta} comes Theorem 3.2 of~\cite{MS2017Universal}, rewriting it as $1+2\log(1/\delta)=\log(2/\delta^2)$. The other terms are expressed in big O notation in~\cite{MS2017Universal}, where the authors were interested in the asymptotic case, and we replace them by actual bounds on finite number of rounds using a derivation similar to~\cite{UZZ+2020Randomness}, Appendix G. More precisely, in~\cite{MS2017Universal}, the term $O(q)$ in Proposition 6.8 comes from an induction on Proposition B.2 and B.3 together with the fact that $(1-x)^{\alpha}\geq 1-\alpha x$, to bound the term $ \sum_x \left< \rho_{\overline{a}}^x
\right>_{1+\epsilon} / \left< \rho \right>_{1+\epsilon}$ in Eq. (6.24), and can thus be replaced by $2ndq$ in our Eq.~\eqref{eq:Delta}. The term $O(\epsilon/q)$ comes from the Taylor expansion at the order 3 of $x \mapsto 2^x$ around 0 in the proof of Proposition 6.3., and can thus be replaced by the last two terms of our Eq.~\eqref{eq:Delta}. The bound given by our Eq.~\eqref{eq:BoundMain} in the main text was optimised on $\varepsilon$, which led to taking $\varepsilon=\SI{6e-5}{}$.

For the min-entropy bound to hold, the following assumptions have to be satisfied:
\begin{enumerate}
\item The user implements the protocol in a secure lab from which information leakage can be prevented.
\item This lab can be partitioned in two sites, corresponding to Alice and Bob, and the information transfer between these two sites can be upper bounded
\item Quantum mechanics is correct.
\item The user has access to a trusted classical computer
\item The user has access to a source of private random numbers or of public random numbers that are independent of the state of the devices used to implement the protocol.
\end{enumerate}
The first assumption has to be satisfied for any cryptographic applications. The second assumption enables us to address the locality loophole by introducing a measure of signalling, and the third assumption allows us to quantify it via the distance to the set of quantum correlations. The fourth assumption is required for the processing of the data output by the protocol. The last assumption is needed to achieve random expansion (first case) or private randomness generation (second case). No additional assumption, in particular on the inner working of the device, is required.

In order to derive a bound on the min-entropy against classical side information, one can use the protocols proposed in~\cite{PM2013Security, FGS2013Security, NBSP2018Device}. To adapt them to nonzero signalling and our measure of crosstalk, the adequate modified version of the guessing probability problem~\cite{AMP2012Randomness, BSS2014More, NPS2014Using} with a fixed Bell inequality $\beta$ is: 
\begin{equation}
    \begin{split}
        G_{C} (I,\sigma, \ell) = & \max_p \ \max_{s} e_C(p) \\
        & \ \ \text{s.t.} \ \ p \in \mathcal{E}, \\
        & \ \ \phantom{ \text{s.t.} \ \ } \beta(p)=I, \\
        & \ \  \phantom{ \text{s.t.} \ \ } \SF_\ell(p) \leq \sigma,
    \end{split} 
    \label{eq:GuessProba}
\end{equation}
which is similar to the one introduced by Silman et al.~\cite{SPM2013Device}, but where the amount of signalling is bounded at the behaviour level. It would be interesting to study how $G_{C}$ varies as a function of $\sigma$ and to compare it to the $P^*_{xy}(I,\chi)$ introduced in Eq. (2) in~\cite{SZB+2021Device}. We leave this as future work.

\section*{Appendix D: Relation between CHSH violation and photon indistinguishability}
\label{app:chsh_hom}

In this section, we derive the relation between the measured CHSH expression $I_\text{CHSH}$ in our setup and the photon indistinguishability characterized by the HOM visibility $V_\text{HOM}$:
\begin{equation}
I_\text{CHSH}=\sqrt{2}(V_\text{HOM}+1).
\label{eq:CHSH_HOM}
\end{equation}
This allows to set the requirements on the single photon sources used to violate Bell inequalities. We use the formalism of quantum creation operators to predict the behaviour of two partially distinguishable photons in the photonic chip used in our optical setup. We calculate the expression of the output state, and use it to compute the coincidence probabilities. We then compute the expected behaviour as a function of $V_\text{HOM}$, which yields the relation between $I_\text{CHSH}$ and $V_\text{HOM}$.

We initialise the photonic chip by injecting one photon in the spatial mode $0_A$ and a second one in mode $0_B$ (see Fig.\ 1 in main article). We assume in addition that the two injected photons are not necessarily identical. They could for instance not arrive exactly at the same time in the chip, have slightly different wavelengths or not share the same polarisation. The information about these degrees of freedom is encoded in configuration wavefunctions $\ket{\alpha}$ for the first photon and $\ket{\beta}$ for the second one. A photon is thus completely described by its mode and configuration e.g.\ $\ket{0_B:\alpha}$ for a photon in mode $0_B$ and configuration $\ket{\alpha}$. The wavefunction overlap is then $\braket{\alpha}{\beta}$. If the photons are completely distinguishable, $\braket{\alpha}{\beta}=0$ and on the contrary, if they are identical $\braket{\alpha}{\beta}=1$ (up to a phase).

We denote by $a^\dagger_{0,\alpha}$, $a^\dagger_{1,\alpha}$, $b^\dagger_{0,\beta}$ and $b^\dagger_{0,\beta}$ and creation operators respectively associated to the states $\ket{0_A:\alpha}$, $\ket{1_A:\alpha}$, $\ket{0_B:\beta}$ and $\ket{1_B:\beta}$. We use the notation $\ket{\text{vac}}$ for the vacuum state. The chip input state is then $\ket{\psi_0}=\ket{0_A:\alpha,0_B:\beta} = a^\dagger_{0,\alpha}b^\dagger_{0,\beta}\ket{\text{vac}}$. The first part of the chip is dedicated to generating a Bell state. We can write the full quantum state after the first column of symmetric beamsplitters and the swap operation, before post-selection, as:
\begin{multline}
    \ket{\psi_1} = \\ \\ \left(
    \frac{1}{2} a^\dagger_{0,\alpha} a^\dagger_{1,\beta}
    + \frac{i}{2} a^\dagger_{0,\alpha} b^\dagger_{1,\beta} 
    + \frac{i}{2} b^\dagger_{0,\alpha} a^\dagger_{1,\beta}
    - \frac{1}{2} b^\dagger_{0,\alpha} b^\dagger_{1,\beta}
    \right) \ket{\text{vac}}
\end{multline}
The Bell state is generated from this state by post-selection, by keeping the results where Alice and Bob simultaneously measure a photon, which yields:
\begin{equation}
    \ket{\psi_{\text{ent}}} = 
    \frac{1}{\sqrt{2}} a^\dagger_{0,\alpha} b^\dagger_{1,\beta}
    \ket{\text{vac}}
    + \frac{1}{\sqrt{2}} b^\dagger_{0,\alpha} a^\dagger_{1,\beta}
     \ket{\text{vac}},
\end{equation}
where we discarded the global phase factor $i$.

When we set the phase shift $\psi=0$, as it the case in our experiment, the unitary matrix associated to Alice and Bob's interferometer is (see Eq.\ \ref{eq:unitary}): 
\begin{equation}
    \widehat{U}(\phi, \psi=0) = 
    \left[\begin{array}{ll}
    \sin\left(\phi/2\right) &
    \cos\left(\phi/2\right) \\
    \cos\left(\phi/2\right) &
    -\sin\left(\phi/2\right)
    \end{array}\right].
\end{equation}
where we discarded the global phase factor, because the interferometer output are connected to detectors. The quantum state at the exit of the chip is then:
\begin{multline}
    \ket{\psi_{\text{out}}} = 
    \frac{1}{\sqrt{2}}
    \Big[
    \sin \left( \frac{\phi^A}{2}  \right) a^\dagger_{0,\alpha}
    +\cos \left( \frac{\phi^A}{2}  \right) a^\dagger_{1,\alpha}
    \Big]
    \\ \times
    \Big[
    \cos \left( \frac{\phi^B}{2}  \right) b^\dagger_{0,\beta}
    -\sin \left( \frac{\phi^B}{2}  \right) b^\dagger_{1,\beta}
    \Big]
    \ket{\text{vac}} \\
    +
    \frac{1}{\sqrt{2}}
    \Big[
    \sin \left( \frac{\phi^B}{2}  \right) b^\dagger_{0,\alpha}
    +\cos \left( \frac{\phi^B}{2}  \right) b^\dagger_{1,\alpha}
    \Big]
    \\ \times
    \Big[
    \cos \left( \frac{\phi^A}{2}  \right) a^\dagger_{0,\beta}
    -\sin \left( \frac{\phi^A}{2}  \right) a^\dagger_{1,\beta}
    \Big]
    \ket{\text{vac}}.
\end{multline}
Writing
$\ket{\beta} = c_\parallel\ket{\alpha} +c_\perp\ket{\alpha_\perp}
$ with $\ket{\alpha_\perp}$ a unit vector such that $\braket{\alpha}{\alpha_\perp}=0$, and $|c_\parallel|^2+|c_\perp|^2=1$, we have $a^\dagger_\beta = c_\parallel a^\dagger_\alpha + c_\perp a^\dagger_{\alpha_\perp}$ and the output state is equal to
\begin{multline}
    \ket{\psi_{\text{out}}} = 
    \frac{1}{\sqrt{2}}
    \Big[
    \sin \left( \frac{\phi^A}{2}  \right) a^\dagger_{0,\alpha}
    +\cos \left( \frac{\phi^A}{2}  \right) a^\dagger_{1,\alpha}
    \Big]
    \\ \times
    \Big[
    c_\parallel \cos \left( \frac{\phi^B}{2}  \right) b^\dagger_{0,\alpha}
    +c_\perp \cos \left( \frac{\phi^B}{2}  \right) b^\dagger_{0,\alpha_\perp}
    \\
    -c_\parallel \sin \left( \frac{\phi^B}{2}  \right) b^\dagger_{1,\alpha}
    -c_\perp \sin \left( \frac{\phi^B}{2}  \right) b^\dagger_{1,\alpha_\perp}
    \Big]
    \ket{\text{vac}} \\
    +
    \frac{1}{\sqrt{2}}
    \Big[
    \sin \left( \frac{\phi^B}{2}  \right) b^\dagger_{0,\alpha}
    +\cos \left( \frac{\phi^B}{2}  \right) b^\dagger_{1,\alpha}
    \Big]
    \\
    \Big[
    c_\parallel \cos \left( \frac{\phi^A}{2}  \right) a^\dagger_{0,\alpha}
    + c_\perp \cos \left( \frac{\phi^A}{2}  \right) a^\dagger_{0,\alpha_\perp}
    \\
    - c_\parallel \sin \left( \frac{\phi^A}{2}  \right) a^\dagger_{1,\alpha}
    - c_\perp \sin \left( \frac{\phi^A}{2}  \right) a^\dagger_{1,\alpha_\perp}
    \Big]
    \ket{\text{vac}}.
\end{multline}

In the protocol, the chip modes $0_A$ and $1_A$ correspond to Alice's results 0 and 1 respectively, and the modes $0_B$ and $1_B$ correspond to Bob's results 0 and 1. Let $p(a,b|\phi^A, \phi^B)$ be the probability that Alice measures $a$ and Bob measures $b$ at the same time, knowing that their interferometer phases are $\phi^A$ and $\phi^B$ respectively. We can compute it from the expression of $\ket{\psi_{\text{out}}}$ by pairing the creation operators according to $(a,b)$ and summing the modulus square of the coefficients. For instance, to compute $p(0,0|\phi^A, \phi^B)$, we identify the pairs $a^\dagger_{0,\alpha} b^\dagger_{0,\alpha}$, $a^\dagger_{0,\alpha} b^\dagger_{0, \alpha_\perp}$, $b^\dagger_{0,\alpha} a^\dagger_{0,\alpha}$ and $b^\dagger_{0,\alpha} a^\dagger_{0,\alpha_\perp}$ in $\ket{\psi_\text{out}}$. Each of these pairs of creation operators, when applied on $\ket{\text{vac}}$, yields a state where Alice and Bob both measure the result 0. Notice that $a^\dagger_{0,\alpha} b^\dagger_{0,\alpha}=b^\dagger_{0,\alpha} a^\dagger_{0,\alpha}$. We thus obtain:
\begin{multline}
    p(0,0|\phi^A, \phi^B) = \\
    \bigg|
    \frac{c_\parallel}{\sqrt{2}}
    \left(
    \sin \left( \frac{\phi^A}{2}  \right)
    \cos \left( \frac{\phi^B}{2}  \right)
    + \sin \left( \frac{\phi^B}{2}  \right)
    \cos \left( \frac{\phi^A}{2}  \right)
    \right)
    \bigg|^2
    \\
    +
    \bigg|
    \frac{c_\perp}{\sqrt{2}}
    \sin \left( \frac{\phi^A}{2}  \right)
    \cos \left( \frac{\phi^B}{2}  \right)
    \bigg|^2
    \\
    +
    \bigg|
    \frac{c_\perp}{\sqrt{2}}
    \sin \left( \frac{\phi^B}{2}  \right)
    \cos \left( \frac{\phi^A}{2}  \right)
    \bigg|^2
\end{multline}
which amounts to:

\begin{multline}
    p(0,0|\phi^A, \phi^B) = p(1,1|\phi^A, \phi^B) = \\ \\
    \frac{|c_\parallel|^2-1}{8} \cos \left( \phi^A -\phi^B  \right)
    -
    \frac{|c_\parallel|^2+1}{8} \cos \left( \phi^A +\phi^B  \right)
    +
    \frac{1}{4},
\end{multline}

\begin{multline}
    p(0,1|\phi^A, \phi^B) = p(1,0|\phi^A, \phi^B) = \\ \\
    -\frac{|c_\parallel|^2-1}{8} \cos \left( \phi^A -\phi^B  \right)
    +
    \frac{|c_\parallel|^2+1}{8} \cos \left( \phi^A +\phi^B  \right)
    +
    \frac{1}{4}.
\end{multline}

Photon indistinguishability is commonly quantified with the HOM visibility. Consider a symmetric beamsplitter. We inject two photons, one in each beamsplitter input, and measure the probability $p_{\text{coinc}}$ of measuring a coincidence, that is two photons going out of the beamsplitter on different outputs. The HOM visibility is then defined as $V_{\text{HOM}}=1-2p_{\text{coinc}}$. Indistinguishable photons will in this case always both come out of the same beamsplitter output, and no coincidences will be measured, yielding $V_{\text{HOM}}=1$. For perfectly distinguishable photons, $p_{\text{coinc}}=1/2$ and $V_{\text{HOM}}=0$.

To relate $c_\parallel$ to $V_{\text{HOM}}$, we use the same mathematical treatment as above, i.e. we write the input state with creation operators, propagate the state in a symmetric beamsplitter and extract the probability of measuring one photon on each beamsplitter output, yielding $V_{\text{HOM}} = |\braket{\alpha}{\beta}|^2 = |c_\parallel|^2 $.

The expressions output probabilities as a function of $V_{\text{HOM}}$ are then: 

\begin{multline}
    p(0,0|\phi^A, \phi^B) = p(1,1|\phi^A, \phi^B) = \\ \\
    \frac{V_{\text{HOM}}-1}{8} \cos \left( \phi^A -\phi^B  \right)
    -
    \frac{V_{\text{HOM}}+1}{8} \cos \left( \phi^A +\phi^B  \right)
    +
    \frac{1}{4}
    \label{eq:hom_1}
\end{multline}

\begin{multline}
    p(0,1|\phi^A, \phi^B) = p(1,0|\phi^A, \phi^B) = \\ \\
    -\frac{V_{\text{HOM}}-1}{8} \cos \left( \phi^A -\phi^B  \right)
    +
    \frac{V_{\text{HOM}}+1}{8} \cos \left( \phi^A +\phi^B  \right)
    +
    \frac{1}{4}.
    \label{eq:hom_2}
\end{multline}

\begin{table}
	\centering
	\begin{tabular}{cc||cccc}
		 a & b & $00$ & $01$ & $10$ & $11$ \\ \hline
		 $0$ & $0$ & $\frac{1}{4}+\frac{\sqrt{2}}{8}V_{\text{HOM}}$    & $\frac{1}{4}-\frac{\sqrt{2}}{8}V_{\text{HOM}}$    & $\frac{1}{4}-\frac{\sqrt{2}}{8}$    & $\frac{1}{4}-\frac{\sqrt{2}}{8}$  \\
		 $0$ & $1$ & $\frac{1}{4}-\frac{\sqrt{2}}{8}V_{\text{HOM}}$    & $\frac{1}{4}+\frac{\sqrt{2}}{8}V_{\text{HOM}}$    & $\frac{1}{4}+\frac{\sqrt{2}}{8}$    & $\frac{1}{4}+\frac{\sqrt{2}}{8}$ \\
		 $1$ & $0$ & $\frac{1}{4}-\frac{\sqrt{2}}{8}V_{\text{HOM}}$    & $\frac{1}{4}+\frac{\sqrt{2}}{8}V_{\text{HOM}}$    & $\frac{1}{4}+\frac{\sqrt{2}}{8}$    & $\frac{1}{4}+\frac{\sqrt{2}}{8}$ \\
		 $1$ & $1$ & $\frac{1}{4}+\frac{\sqrt{2}}{8}V_{\text{HOM}}$    & $\frac{1}{4}-\frac{\sqrt{2}}{8}V_{\text{HOM}}$    & $\frac{1}{4}-\frac{\sqrt{2}}{8}$   & $\frac{1}{4}-\frac{\sqrt{2}}{8}$
	\end{tabular}
\caption{}
\label{tab:behaviour_expr}
\end{table}

The corresponding behaviour is displayed in Table~\ref{tab:behaviour_expr}, and its associated CHSH value is:
\begin{equation}
    I_\text{CHSH} = \sqrt{2}(V_\text{HOM}+1).
\end{equation}

\section*{Appendix E: Photonic chip crosstalk matrices}
\label{app:crosstalk}

Heat generated by Alice's and Bob's heating resistors propagates in the chip. We denote $V_1$, $V_2$, $V_3$ and $V_4$ the voltages applied respectively on the phase shifters $\psi^A_x$, $\psi^B_y$, $\phi^A_x$ and $\phi^B_x$ (see Fig.\ \ref{subfig:chip_scheme}). The full characterization of the phases implemented by the phase shifters as a function of applied voltages is given with a typical error on the order of 0.1 rad by

\begin{multline}
    \begin{bmatrix}
    \phi^{(A)}_Z \\
    \phi^{(B)}_Z 
    \end{bmatrix}
    =
    \begin{bmatrix}
    1.2890 & -0.0785 \\
    0.0988 & -1.2777
    \end{bmatrix}
    \begin{bmatrix}
    V_1^2 \\
    V_2^2
    \end{bmatrix} 
    \\
    -
    \begin{bmatrix}
    0.0192 & -0.0009 \\
    0.0012 & -0.0203
    \end{bmatrix}
    \begin{bmatrix}
    V_1^4 \\
    V_2^4
    \end{bmatrix}
    \label{eq:Z}
\end{multline}

\begin{multline}
    \begin{bmatrix}
    \phi^{(A)}_Y \\
    \phi^{(B)}_Y 
    \end{bmatrix}
    =
    \begin{bmatrix}
    0.2703 \\
    -0.2799 
    \end{bmatrix}
    +
    \begin{bmatrix}
    1.4693 & -0.1111 \\
    0.1120 & -1.4776
    \end{bmatrix}
    \begin{bmatrix}
    V_3^2 \\
    V_4^2
    \end{bmatrix} 
    \\
    -
    \begin{bmatrix}
    0.0351 & -0.0026 \\
    0.0027 & -0.0343
    \end{bmatrix}
    \begin{bmatrix}
    V_3^4 \\
    V_4^4
    \end{bmatrix}
    \label{eq:Y}
\end{multline}

\noindent
 where the phases are expressed in radians as a function of the voltages in Volts applied on each heating resistor. Injecting a continuous diode laser in the chip and measuring the outputs with photodiodes confirms that there is 
 measurable thermal crosstalk only between resistors belonging to the same column on Fig.\ \ref{subfig:chip_scheme} in the main text.

\section*{Appendix F: Phase calibration protocol}
\label{app:calibration}

The goal of the voltage calibration protocol is
to compute the voltages to apply on Alice's and Bob's Mach-Zehnder interferometer (MZI) phase, such that the measurement contexts presented in Table \ref{table:context_table} can be implemented with high precision. 
The voltages are regularly calibrated during the execution of the protocol because the polarisation of photons in the fibers between the chip and the detectors fluctuate, causing detection efficiency fluctuations on the detectors which are sensitive to polarisation. We compensate for these by shifting the MZI phase. Because we cannot rely on the phases computed from the crosstalk matrix relations (Eqs.~\eqref{eq:Z},~\eqref{eq:Y}) to accurately apply a phase on the interferometer, we use Alice's and Bob's MZI splitting as a measure of the implemented phase (see Fig.~\ref{fig:MZI}). To do so, the motorized shutter in Fig. \ref{fig:setup} is closed, such that only the upper input modes of Alice's and Bob's MZI provide photons. 

\begin{figure}
\includegraphics[height=1.5cm]{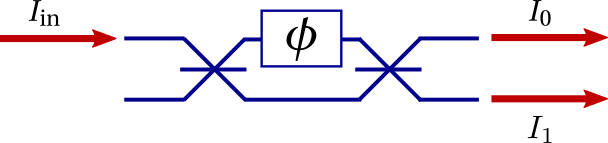}
\centering
\caption{Mach-Zehnder interferometer (MZI) in the configuration used to calibrate the voltages. $I_{in}$ is the input power and $I_0, I_1$ the output ones.  The MZI splitting  $n_0=I_0/(I_0+I_1)$ is related to the MZI's phase $\phi$ by $n_0(\phi) = \sin^2 \left( \frac{\phi}{2} \right)$ according to Eq.\ \ref{eq:unitary}.}
\label{fig:MZI}
\end{figure}
\subsection{Relative detector efficiency measurement.}
First, $\phi^A$ and $\phi^B$ are swept simultaneously while recording the countrate on Alice's and Bob's outputs, which produces the data presented in Fig.~\ref{fig:MZI_sweep}. 

\begin{figure}
\includegraphics[width=\textwidth]{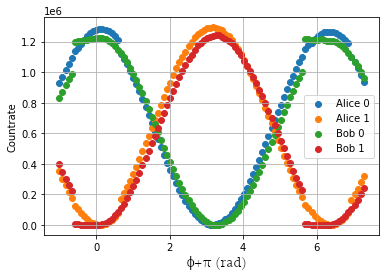}
\centering
\caption{Measured count rates on Alice’s and Bob’s outputs  while sweeping $\phi^A$ and $\phi^B$. Notice that the heating resistors do not achieve a full $2\pi$ sweep.}
\label{fig:MZI_sweep}
\end{figure}

We call $N^A_{0,\text{max}}$ the maximum countrate recorded during the sweep by Alice on her 0 output, and similarly for $N^A_{1,\text{max}}$, $N^B_{0,\text{max}}$ and $N^B_{1,\text{max}}$. Instead of using the raw detector countrates $N^A_0$, we worked with the normalized countrates $\tilde{N}^A_0 = N^A_0/N^A_{0,\text{max}}$ to compensate for the different detector efficiencies. The corrected MZI splittings are 
$\tilde{n}^A_0 = \tilde{N}^A_0/(\tilde{N}^A_0+\tilde{N}^A_1), 
\tilde{n}^B_0 = \tilde{N}^B_0/(\tilde{N}^B_0+\tilde{N}^B_1)$
\footnote{It was realised after the data acquisition that the detector efficiency should in fact not be accounted for, because correcting for that efficiency introduces biases in the measured empirical table that reduce the CHSH violation. In practice, the detector efficiencies differ only at the third significant digit, so the impact on the obtained results is negligible: correcting by these biases diminishes the maximum obtainable CHSH violation by only $10^{-4}$ (simulation using the Perceval package in Python). In addition, in principle these biases do not increase the signaling fraction. Therefore, we can still exploit the results from our experiment.}.

\subsection{Local phase sweeps}

For each measurement context, $\phi^A$ and $\phi^B$ are swept simultaneously around the target phase while recording the corrected MZI splittings. The target phases and MZI splittings for each context are displayed in Table~\ref{tab:target_phases}. From the measured MZI splittings, we deduce the measured interferometer phase:
$\phi^A_{\text{measured}} = 2 \arcsin(\sqrt{\tilde{n}^A_0})$, $\phi^B_{\text{measured}} = 2 \arcsin(\sqrt{\tilde{n}^B_0}).$ 
The plots of the measured phases as a function of the input phases are displayed in Fig.~\ref{fig:phases}. The data is processed with a linear or triangular fit depending on the context, and the input phase that should be applied to the chip is retrieved from the vertical lines on the plots. 

\begin{table}
\begin{tabular}{lllll}
\hline
Measurement \\ context & Target $\phi^A$ & Target $n^A_0$ & Target $\phi^B$ & Target $n^B_0$ \\ \hline
(x=0,y=0)       & $-\pi/2   $                        & $0.5  $                                &$ -\pi/4    $                    &$ 0.146    $                          \\
(x=0,y=1)       & $-\pi/2 $                          & $0.5$                                  & $\pi/4  $                      &$ 0.146  $                            \\
(x=1,y=0)       & $0  $                           & $0  $                                  & $-\pi/4  $                      & $0.146  $                            \\
(x=1,y=1)       & $0$                             & $0$                                    & $\pi/4$                        & $0.146$                              \\ \hline
\end{tabular}
\caption{Target phases and MZI splittings for each measurement context.}
\label{tab:target_phases}
\end{table}

\begin{figure*}
\includegraphics[width=0.85\textwidth]{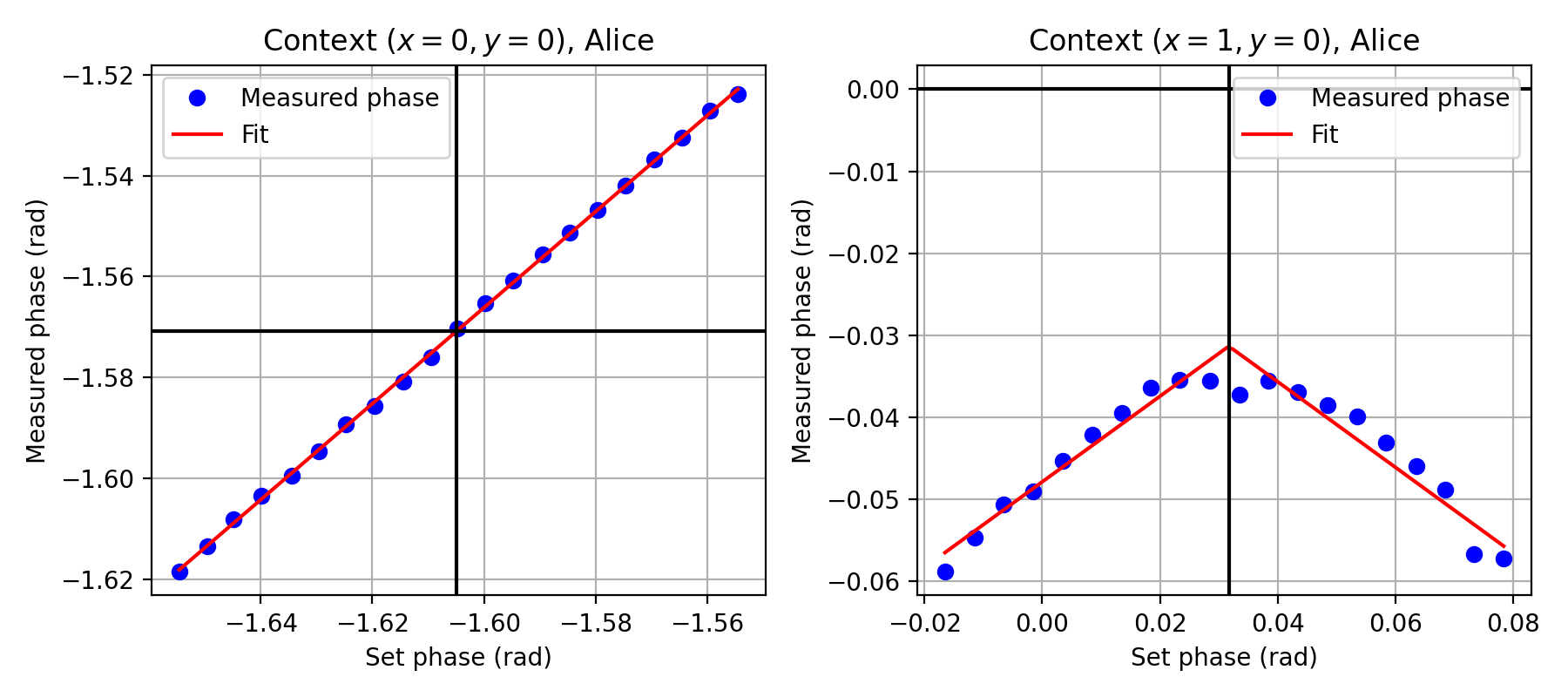}
\centering
\caption{For each measurement context, the input phase of Alice's and Bob's MZI is swept around the target phase and we record the measured phase from the MZI splitting. Data is fitted with a line or a triangle depending on the measurement context. Here we show the result for Alice in contexts 00 and 01. The horizontal black line indicates the target phase for each context, and the vertical one represents the input phase that should be used to implement the context, knowing that in general the input phase is not equal to the measured phase. Notice that for Alice, in contexts $(x=1,y=0)$ and  $(x=1,y=1)$, the measured phase should be $\pi$ at the triangle’s peak. It does not because of dark counts, which prevent the measured interferometer balance from going to 0, and thus we rely on a triangular fit of the sweep.}
\label{fig:phases}
\end{figure*}

\subsection{Phase drifts}

We measured the phases of Alice's and Bob's MZI for 14 hours using always the same voltages to implement the four measurement contexts, and by cycling through them in the order ($x=0,y=0$), ($x=0,y=1$), ($x=1,y=0$) and ($x=1,x=1$). The results are summarized in Fig. \ref{fig:mzi_stability_1}. We observe overall a typical drift of the implemented phase of the order of $\SI{0.25}{mrad\per\hour}$. 

\begin{figure*}
\includegraphics[width=\textwidth]{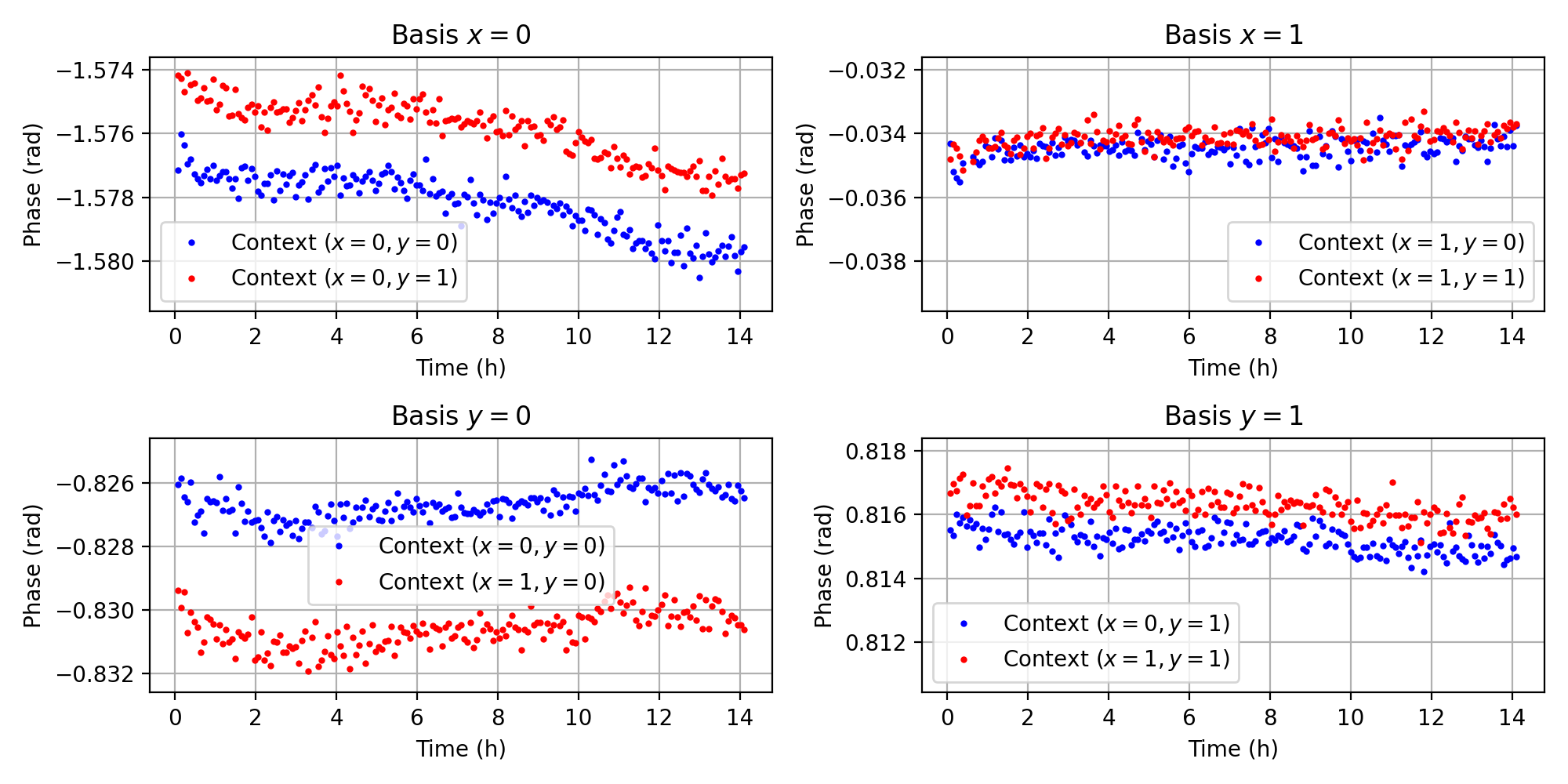}
\centering
\caption{For each measurement context, the current MZI phase is measured using the MZI’s balance at Alice's and Bob's outputs. For each plot, the vertical scale is 2 mrad per division.}
\label{fig:mzi_stability_1}
\end{figure*}

\subsection{Phase stabilization}

In Fig.~\ref{fig:mzi_stability_2} we show Alice's and Bob's MZI phases measured every 6 hours before and after calibration during our main experiment. As a result of these frequent calibrations, the implemented phases stayed confined in an interval of 3 mrad around the targets.

\begin{figure*}
\includegraphics[width=\textwidth]{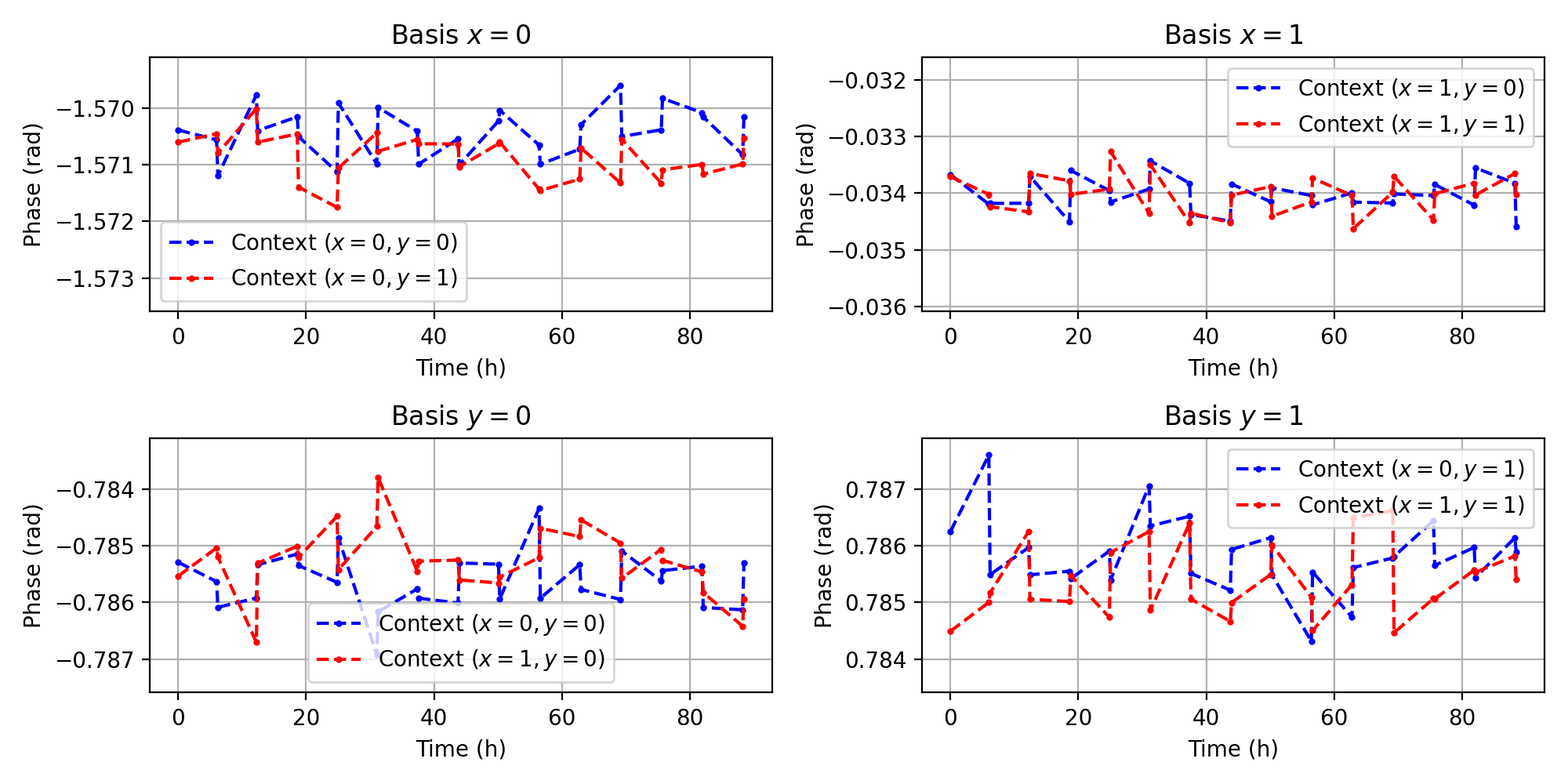}
\centering
\caption{For each measurement context, the current MZI phase is measured every 6 hours before and after a voltage calibration using the MZI’s balance at Alice's and Bob's outputs. For each plot, the vertical scale is 1 mrad per division.}
\label{fig:mzi_stability_2}
\end{figure*} 

\end{document}